\documentclass[acmsmall]{acmart}

\setcopyright{none}
\acmPrice{}
\acmDOI{10.1145/3371105}
\acmYear{2020}
\copyrightyear{2020}
\acmJournal{PACMPL}
\acmVolume{4}
\acmNumber{POPL}
\acmMonth{1}

\DeclareMathAlphabet{\mathpzc}{OT1}{pzc}{m}{it}

\renewcommand{\phi}{\varphi}

\renewcommand{\epsilon}{\varepsilon}

\newcommand{\sfsymbol}[1]{\textsf{\upshape {#1}}}
\newcommand{\ttsymbol}[1]{\texttt{\upshape {#1}}}

\newcommand{\wpsymbol}{\sfsymbol{wp}}
\newcommand{\boldwpsymbol}{\textbf{\sfsymbol{wp}}}
\renewcommand{\wp}[2]{\wpsymbol \left\llbracket {#1} \right\rrbracket \left( {#2} \right)}
\newcommand{\boldwp}[2]{\boldsymbol{\boldwpsymbol \left\llbracket {#1} \right\rrbracket \left( {#2} \right)}}
\newcommand{\wpC}[1]{\wpsymbol \left\llbracket {#1} \right\rrbracket}

\newcommand{\ertsymbol}{\sfsymbol{ert}}
\newcommand{\boldertsymbol}{\textbf{\sfsymbol{ert}}}
\newcommand{\ert}[2]{\ertsymbol \, \left\llbracket {#1} \right\rrbracket \, \left( {#2} \right)}
\newcommand{\boldert}[2]{\boldsymbol{\boldertsymbol \left\llbracket {#1} \right\rrbracket \left( {#2} \right)}}

\newcommand{\conditionalPair}[2]{{\let\oldarraystretch\arraystretch}\renewcommand{\arraystretch}{1}~\holter{~\raisebox{.5ex}{${#1}$}~}{~\raisebox{.125ex}{${#2}$}~}~\renewcommand{\arraystretch}{\oldarraystretch}}

\newcommand{\SKIP}{\ttsymbol{skip}}

\newcommand{\AssignSymbol}{\mathrel{\textnormal{\texttt{:=}}}}
\newcommand{\ASSIGN}[2]{\ensuremath{#1 \AssignSymbol #2}}

\newcommand{\COMPOSE}[2]{\ensuremath{{#1}{\,\fatsemi}~ {#2}}}
\newcommand{\PCHOICE}[3]{\ensuremath{\left\{\, {#1} \,\right\}\mathrel{\left[{#2}\right]}\left\{\, {#3} \,\right\}}}

\newcommand{\IFSYMBOL}{\ensuremath{\textnormal{\texttt{if}}}}

\newcommand{\ELSESYMBOL}{\ensuremath{\textnormal{\texttt{else}}}}

\newcommand{\ITE}[3]{\ensuremath{\IFSYMBOL\,\left(\, {#1} \,\right)\,\left\{\, {#2} \,\right\}\,\ELSESYMBOL\,\left\{\, {#3} \,\right\}}}
\newcommand{\TRUE}{\ensuremath{\textnormal{\texttt{true}}}}

\newcommand{\WHILESYMBOL}{\ensuremath{\textnormal{\texttt{while}}}}
\newcommand{\WHILE}[1]{\ensuremath{\WHILESYMBOL \left(\, {#1} \,\right)\left\{\right.}}

\newcommand{\WHILEDO}[2]{\ensuremath{\WHILESYMBOL \left(\, {#1} \,\right)\left\{\, {#2} \,\right\}}}

\newcommand{\pgcl}{\textnormal{\sfsymbol{pGCL}}\xspace} 
\newcommand{\Vars}{\ensuremath{\mathsf{Vars}}\xspace} 

\newcommand{\Nats}{\ensuremath{\mathbb{N}}\xspace}
\newcommand{\NatsInf}{\ensuremath{\overline{\Nats}}}

\newcommand{\Reals}{\ensuremath{\mathbb{R}}\xspace}
\newcommand{\PosReals}{\mathbb{R}_{\geq 0}}
\newcommand{\PosRealsInf}{\overline{\mathbb{R}}_{\geq 0}}

\newcommand{\States}{\Sigma}
\newcommand{\State}{s}

\newcommand{\E}{\mathbb{F}}

\newcommand{\abs}[1]{\ensuremath{\left\lvert #1 \right\rvert}}

\newcommand{\iverson}[1]{\left[ {#1} \right]}

\newcommand{\Min}[2]{\min\left\{\,{#1},\: {#2}\,\right\}}

\newcommand{\Max}[2]{\max\left\{\,{#1},\: {#2}\,\right\}}

\newcommand{\subst}[2]{\left[ {#1} \middle/ {#2}\right]}

\newcommand{\charfun}[4]{\tensor*[^{\smash{#1}}_{\smash{\langle #2, #3 \rangle}}]{\Phi}{_{{#4}}}}
\newcommand{\charfunn}[5]{\tensor*[^{#1}_{{\langle #2, #3 \rangle}}]{\Phi}{_{{#4}}^{{#5}}}}

\newcommand{\charwp}[3]{\charfun{\wpsymbol}{#1}{#2}{#3}}

\newcommand{\charert}[3]{\charfun{\ertsymbol}{#1}{#2}{#3}}
\newcommand{\charwpnoindex}[3]{{}^{\wpsymbol}\Phi_{#3}}
\newcommand{\charertnoindex}[3]{{}^{\ertsymbol}\Phi_{#3}}

\newcommand{\charwpn}[4]{\charfunn{\wpsymbol}{#1}{#2}{#3}{#4}}

\newcommand{\charwpnnoindex}[4]{{}^{\wpsymbol}\Phi^{#4}_{#3}}
\newcommand{\charertnnoindex}[4]{{}^{\ertsymbol}\Phi^{#4}_{#3}}

\newcommand{\To}{\rightarrow}

\newcommand{\mydot}{{}\scalebox{1.05}{$\centerdot$}{}~}

\newcommand{\guard}{\varphi}

\newcommand{\qiff}{\quad\textnormal{iff}\quad}
\newcommand{\qqiff}{\qquad\textnormal{iff}\qquad}

\newcommand{\qand}{\quad\textnormal{and}\quad}
\newcommand{\qqand}{\qquad\textnormal{and}\qquad}

\newcommand{\qimplies}{\quad\textnormal{implies}\quad}

\newcommand{\qqimplies}{\qquad\textnormal{implies}\qquad}

\newcommand{\pprec}{\mathrel{{\prec}\hspace{-.5ex}{\prec}}}

\newcommand{\ppreceq}{~{}\preceq{}~}

\newcommand{\eeq}{~{}={}~}
\newcommand{\ccoloneqq}{~{}\coloneqq{}~}

\newcommand{\lleq}{~{}\leq{}~}

\newcommand{\ggeq}{~{}\geq{}~}
\newcommand{\ssqsubseteq}{~{}\sqsubseteq{}~}

\newcommand{\pplus}{~{}+{}~}

\newcommand{\setcomp}[2]{\left\{\, {#1} ~\middle|~ {#2} \,\right\}}

\newcommand{\annocolor}[1]{\textcolor{darkgray!50!gray}{#1}}
\newcommand{\annotate}[1]{\boldsymbol{\annocolor{\!\!{\fatslash}\!\!{\fatslash}~~\vphantom{G'} {#1}}}}
\newcommand{\starannotate}[1]{\boldsymbol{\annocolor{{\talloblong}\!{\talloblong}\:\vphantom{G'} {#1}}}}
\newcommand{\phiannotate}[1]{\boldsymbol{\annocolor{\!\!\hspace{-.15ex}{}^{\annocolor{\Phi}}\!\!\!{\fatslash}\!\!{\fatslash}~~\vphantom{G'} {#1}}}}

\newcommand{\bowtieannotate}[1]{\annocolor{\!\!\hspace{-.5ex}{}^{\annocolor{{\bowtie}}}\boldsymbol{\!\!\!{\fatslash}\!\!{\fatslash}~~\vphantom{G'} {#1}}}}

\newcommand{\preceqannotate}[1]{\annocolor{\!\!\hspace{-.25ex}{}^{\annocolor{{\preceq}}}\boldsymbol{\!\!\!{\fatslash}\!\!{\fatslash}~~\vphantom{G'} {#1}}}}

\newcommand{\eqannotate}[1]{\annocolor{\!\!\hspace{-.1ex}{}^{\annocolor{{=}}}\boldsymbol{\!\!\!{\fatslash}\!\!{\fatslash}~~\vphantom{G'} {#1}}}}

\newcommand{\wpannotate}[1]{\boldsymbol{\annocolor{\!\!\hspace{-0.75ex}{}^{\annocolor{\text{\tiny $\wpsymbol$}}}\!\!\!{\fatslash}\!\!{\fatslash}~~\vphantom{G'} {#1}}}}

\newcommand{\ertannotate}[1]{\boldsymbol{\annocolor{\!\!\hspace{-.6ex}{}^{\annocolor{\text{\tiny $\ertsymbol$}}}\!\!\!{\fatslash}\!\!{\fatslash}~~\vphantom{G'} {#1}}}}

\newcounter{computationarrowsone}
\newcounter{computationarrowstwo}

\newcounter{sarrow}

\newcommand{\lfp}{\ensuremath{\textnormal{\sfsymbol{lfp}}~}}
\newcommand{\gfp}{\ensuremath{\textnormal{\sfsymbol{gfp}}~}}

\newcommand{\IP}[2]{\ensuremath{\phantom{}^{#1}\mathbb{P}\left(#2\right)}}

\newcommand{\IR}{\ensuremath{\mathbb{R}}}
\newcommand{\expec}[2]{\ensuremath{\phantom{}^{#1}\mathbb{E}\left(#2\right)}}
\newcommand{\IE}{\ensuremath{\mathbb{E}}}
\newcommand{\F}[1]{\ensuremath{\mathfrak{#1}}}
\newcommand{\sigmagen}[1]{\ensuremath{\left\langle #1 \right \rangle _{\sigma}}}
\newcommand{\primed}[1]{\ensuremath{{#1}^{\prime}}}
\newcommand{\ddashed}[1]{\ensuremath{{#1}^{\prime\prime}}}
\renewcommand{\subset}{\subseteq}

\renewcommand{\comment}[2]{\footnote{\textcolor{orange}{\textbf{ #1 \normalsize !!!
[}} {#2} \textcolor{orange}{\textbf{\normalsize ] !!!}}}}
\renewcommand{\comment}[2]{}

\newcommand{\oldcomment}[1]{}
\newcommand{\termtime}[1]{\ensuremath{T^{\neg #1}}}
\newcommand{\run}{\vartheta}
\newcommand{\indicator}[1]{\ensuremath{\left[{#1}\right]}}
\newcommand{\lp}{{\mbox{\tiny \rm loop}}}
\newcommand{\harm}[1]{\ensuremath{\mathcal{H}_{#1}}}
\newcommand{\Diff}[1]{\ensuremath{\Delta #1}}

\newcommand{\lowerpenalties}{\allowdisplaybreaks\predisplaypenalty=0\clubpenalty=0\widowpenalty=0\displaywidowpenalty=0\brokenpenalty=0{}}
\newcommand{\restorepenalties}{\allowdisplaybreaks[0]\predisplaypenalty=10000\clubpenalty=150\widowpenalty=150\displaywidowpenalty=50\brokenpenalty=100{}}
\newcommand{\unlowerpenalties}{\restorepenalties}

\makeatletter \newlength{\negph@wd}
\DeclareRobustCommand{\negphantom}[1]{%
	\ifmmode \mathpalette\negph@math{#1}%
	\else \negph@do{#1}%
	\fi
}
\newcommand{\negph@math}[2]{\negph@do{$\m@th#1#2$}}
\newcommand{\negph@do}[1]{%
	\settowidth{\negph@wd}{#1}%
	\hspace*{-\negph@wd}%
}
\makeatother

\usepackage{amsmath,amsthm,mathtools}
\usepackage{hyperref}

\usepackage[capitalize,nameinlink]{cleveref}
\crefname{lemma}{Lem.}{Lem.}
\crefname{example}{Ex.}{Ex.}
\crefname{section}{Sect.}{Sect.}
\crefname{appendix}{App.}{App.}
\crefname{definition}{Def.}{Def.}
\crefname{theorem}{Thm.}{Thm.}
\crefname{corollary}{Cor.}{Cor.}
\crefname{algorithm}{Alg.}{Alg.}
\crefname{equation}{}{}
\crefname{counterexample}{Counterex.}{Counterex.}
\usepackage{microtype}
\usepackage{nicefrac}
\usepackage{tensor}
\usepackage{xspace}
\usepackage{pifont}
\usepackage{proof}
\usepackage{stmaryrd}
\usepackage{tabularx}
\usepackage{tensor}
\usepackage{wasysym}
\usepackage{xfrac}
\usepackage{tikz}
\usetikzlibrary{decorations.pathmorphing,decorations.fractals,shapes,calc,patterns,arrows.meta,arrows,decorations.pathmorphing,positioning,fit,trees,shapes,shadows,automata,calc,decorations.pathreplacing}
\usepackage{tikz-cd}
\usepackage{adjustbox}
\usepackage{thmtools, thm-restate}
\usepackage{subcaption}
\usepackage{wrapfig}
\usepackage{cutwin}

\theoremstyle{plain}
\newtheorem*{thm*}{Theorem}
\newtheorem*{exmp*}{Example}
\newtheorem*{counterexmp*}{Counterexample}

\newtheoremstyle{ourstyle}{}{}{\itshape}{}{\bfseries}{.}{ }{\thmname{#1}\thmnumber{ #2}\thmnote{ (#3)}}
\theoremstyle{ourstyle}

\newtheorem{theorem}{Theorem}
\newtheorem{counterexample}[theorem]{Counterexample}

\hypersetup{
    colorlinks,
    linkcolor={red!80!black},
    citecolor={green!60!black},
    urlcolor={blue}
}

\newtheorem{definition}[theorem]{Definition}
\newtheorem{lemma}[theorem]{Lemma}
\newtheorem{example}[theorem]{Example}
\newtheorem{corollary}[theorem]{Corollary}

\newcommand{\eat}[1]{}
\usepackage{booktabs}   
\usepackage{subcaption} 

\begin{document}

\title{Aiming Low Is Harder}         
\titlenote{This technical report supplements a paper of the same title which appeared at
  POPL 2020.}
\subtitle{Induction for Lower Bounds in Probabilistic Program Verification}  


\author{Marcel Hark}
\affiliation{
  \institution{RWTH Aachen University}            
  \country{Germany}
}
\email{marcel.hark@cs.rwth-aachen.de}          

\author{Benjamin Kaminski}
\affiliation{
  \institution{RWTH Aachen University}            
  \country{Germany}
}
\email{benjamin.kaminski@cs.rwth-aachen.de}         

\author{J\"urgen Giesl}
\affiliation{
  \institution{RWTH Aachen University}            
  \country{Germany}
}
\email{giesl@cs.rwth-aachen.de}         

\author{Joost-Pieter Katoen}
\affiliation{
  \institution{RWTH Aachen University}            
  \country{Germany}
}
\email{katoen@cs.rwth-aachen.de}         

\allowdisplaybreaks
\begin{abstract}
	We present a new inductive rule for verifying lower bounds on expected values of random variables after execution of probabilistic loops as well as on their expected runtimes. 
	Our rule is \emph{simple} in the sense that loop body semantics need to be applied only finitely often in order to verify that the candidates are indeed lower bounds.
	In particular, it is not necessary to find the limit of a sequence as in many previous rules.
\end{abstract}

\ccsdesc[500]{Mathematics of computing~Probabilistic algorithms}
\ccsdesc[300]{Mathematics of computing~Markov processes}
\ccsdesc[100]{Theory of computation~Denotational semantics}

\keywords{probabilistic programs, verification, weakest precondition, weakest preexpectation, lower bounds, optional stopping theorem, uniform integrability}  

\newcommand{\cameraready}[1]{#1}
\newcommand{\techrep}[1]{#1}

\renewcommand{\cameraready}[1]{}

\maketitle

\section{Introduction and Overview}
We study probabilistic programs featuring discrete probabilistic choices
as well as \emph{unbounded loops}.
Randomized algorithms are the classical application of such programs.
Recently, applications in \emph{biology}, \emph{quantum computing}, \emph{cyber security}, \emph{machine learning}, and \emph{artificial intelligence}
led to rapidly growing interest in probabilistic programming\ \cite{DBLP:conf/icse/GordonHNR14}.

Formal verification of probabilistic programs is strictly harder than for nonprobabilistic programs\ \cite{DBLP:journals/acta/KaminskiKM19}.
Given a random variable $f$, a key verification task is to reason about the \emph{expected value of $f$} after termination of a program\ $C$ on input\ $\State$.
If $f$ is the indicator function of an event\ $A$, then this expected value is the probability that $A$ has occurred on termination of $C$.

For verifying probabilistic loops, most approaches share a common, \emph{conceptually very simple}, technique: an \emph{induction rule} for verifying \emph{upper bounds} on expected values, which are characterized as least fixed points ($\textsf{lfp}$) of a suitable function\ $\Phi$.
This rule, called \mbox{``Park induction'', reads}
\begin{align*}
	\Phi(I) \ssqsubseteq I \qqimplies \lfp\, \Phi \ssqsubseteq I\ ,
\end{align*}
i.e., for a candidate upper bound $I$ we check $\Phi(I) \sqsubseteq I$ (for a suitable partial order $\sqsubseteq$) to prove that $I$ is indeed an upper bound on the least fixed point, and hence on the sought--after expected value.

For \emph{lower bounds}, a simple proof principle analogous to Park induction, namely
\begin{align*}
	I \ssqsubseteq \Phi(I) \qqimplies I \ssqsubseteq \lfp\, \Phi\ , \tag*{\Large\lightning}
\end{align*}
is \emph{unsound} in general.
\emph{Sound} rules (see \cref{sec:related}), on the other hand, often suffer from the fact that either $f$ needs to be \emph{bounded}, or that one has to find the \emph{limit of some sequence}, as well as the sequence itself, rendering those rules conceptually much more involved than Park induction.

Our main contribution (\cref{sec:optional-stopping-wp}, \cref{thm:optional_stopping_probabilistic_programs}) is to provide relatively \emph{simple} side conditions that can be added to the (unsound) implication above, such that the implication becomes true, i.e.,
\begin{align*}
	I \ssqsubseteq \Phi(I) \ \,{}\wedge{}\ {
	\begin{array}{c}
		\textnormal{\small some side} \\[-.5ex]
		\textnormal{\small conditions}
	\end{array}
	} \qqimplies I \ssqsubseteq \lfp\, \Phi\ . \tag*{\Large$\checkmark$}
\end{align*}
In particular, our side conditions will be simple in the sense that (a variation of) $\Phi$ needs to be applied to a candidate $I$ only a \emph{finite} number of times, \mbox{which is beneficial for potential \emph{automation}}.

The need for verifying lower bounds on expected values is quite natural: First of all, they help to assess the quality and tightness of upper bounds. Moreover, giving \emph{total correctness} guarantees for probabilistic programs amounts to lower--bounding the correctness probability, e.g., in order to establish membership in complexity classes like \textsf{RP} and \textsf{PP}.

In addition to expected values of random variables at program termination, lower bounds on expected runtimes are also of significant interest: Lower bounds on expected runtimes which depend on secret program variables may compromise the secret, thus allowing for timing side--channel attacks; ``very large'' lower bounds could indicate \mbox{potential denial--of--service attacks}.

In order to enable practicable reasoning about lower bounds on expected runtimes, we will show how our inductive lower bound rule carries over to expected runtimes (\cref{sec:runtime}, \cref{thm:lower_bounds_ert}). As an example to show the applicability of our rule, we will verify that the well--known and notoriously difficult coupon collector's problem \cite{DBLP:books/cu/MotwaniR95}
(\cref{sec:runtime}, \cref{ex:coupon-collector}), modeled by \mbox{the probabilistic program}\footnote{The random assignment $\ASSIGN{i}{\mathrm{Unif}[1..N]}$ does not --- strictly speaking --- adhere to our syntax of binary probabilistic choices, but it can be modeled in our syntax. For the sake of readability, we opted for $\ASSIGN{i}{\mathrm{Unif}[1..N]}$.}
\begin{align*}
	 & \COMPOSE{\ASSIGN{x}{N}}{}                                             \\
	 & \WHILE{0 < x}                                                         \\
	 & \qquad \COMPOSE{\ASSIGN{i}{N+1}}{}                                    \\
	 & \qquad \COMPOSE{\WHILEDO{x < i}{ \ASSIGN{i}{\mathrm{Unif}[1..N]} }}{} \\
	 & \qquad \ASSIGN{x}{x-1}                                                \\
	 & \}\ ,
\end{align*}
has an expected runtime of at least $N \harm{N}$, where $\harm{N}$ is the $N$-th harmonic number.

Our new inductive rules will be stated in terms of so--called expectation transformers \techrep{\linebreak}\cite{DBLP:series/mcs/McIverM05} (\cref{sec:wp}) and rely on the notions of \emph{uniform integrability} (\cref{sec:bounds}, in particular \ref{sec:series_perspective_on_loops}, and \cref{sec:expectations_processes}), \emph{martingales}, \emph{conditional difference boundedness}, and the \emph{Optional Stopping Theorem} (\cref{sec:optional-stopping-wp}) from the theory of stochastic processes.
However, we do not only \emph{make use} of these notions in order to prove soundness of our induction rule, but instead establish tight connections in terms of these notions between expectation transformers and certain canonical stochastic processes (\cref{sec:expectations_processes}, \cref{thm:lfp_expectation} and \cref{sec:optional-stopping-wp}, \cref{thm:conditional_difference_boundedness}).
In particular, we will build upon the key result of this connection (\cref{thm:lfp_expectation}) to study exactly how inductive proof rules for both upper and lower bounds can be understood in the realm of these stochastic processes and vice versa (\cref{sec:optional-stopping-wp}, \cref{thm:optional_stopping_probabilistic_programs} and \cref{sec:fatou}).
We see those connections between the theories of expectation transformers and stochastic processes as a stepping stone for applying further results from stochastic process theory to probabilistic program analysis and possibly also vice versa.

As a final contribution, we revisit one of the few existing rules for lower bounds due to \techrep{\linebreak}\cite{DBLP:series/mcs/McIverM05}, which gives \emph{sufficient} criteria for a candidate being a lower bound on the expected value of a \emph{bounded} function $f$. We show that their rule is also a consequence of uniform integrability and we are moreover able to generalize their rule to a \emph{necessary} and sufficient criterion (\cref{sec:mciver_morgan}, \cref{thm:generalization_mciver_morgan}). We demonstrate the usability of our generalization by an example (\cref{sec:mciver_morgan}, \cref{ex:mciver_morgan}).

\techrep{The appendix contains }\cameraready{We refer to \cite{arxiv} for }more case studies illustrating the effectiveness of our lower bound proof rule, a more detailed introduction to probability theory, and more detailed proofs of our results.

\section{Weakest Preexpectation Reasoning}
\label{sec:wp}
\emph{Weakest preexpectations} for probabilistic programs are a generalization of Dijkstra's \emph{weakest preconditions} for nonprobabilistic programs.
Dijkstra employs \emph{predicate transformers}, which push a \emph{postcondition~$F$} (a predicate) backward through a nonprobabilistic program~$C$ and yield the \emph{weakest precondition}~$G$ (another predicate) describing the largest set of states such that whenever~$C$ is started in a state satisfying~$G$, $C$ terminates in a state satisfying~$F$.\footnote{We consider total correctness, i.e., from any state satisfying the weakest precondition~$G$, $C$ definitely terminates.}

The \emph{weakest preexpecation calculus} on the other hand employs \emph{expectation transformers} which act on real--valued functions called \emph{expectations}, mapping program states to non--negative reals.\footnote{For simplicity of the presentation, we study the standard case of positive expectations.
	Mixed--sign expectations mapping to the \emph{full} extended reals require much more technical machinery, see~\cite{DBLP:conf/lics/KaminskiK17}.}
These transformers push a \emph{postexpectation~$f$} backward through a probabilistic program~$C$ and yield a \emph{preexpectation~$g$}, such that~$g$ represents the expected value of~$f$ after executing~$C$.
The term \emph{expectation} coined by \cite{DBLP:series/mcs/McIverM05} may appear somewhat misleading at first.
We \emph{clearly distinguish between expectations and expected values}: An expectation is hence not an expected value, per se.
Instead, we can think of an expectation as a \emph{random variable}.
In Bayesian network jargon, expectations are also called \emph{factors}.%
\begin{definition}[Expectations~\textnormal{\cite{thesis:kaminski,DBLP:series/mcs/McIverM05}}]
	\label{defQQQexpectations}
	Let $\Vars$ denote the finite set of program variables and let $\States = \{ \State \mid \State\colon \Vars \To \mathbb{Q} \}$
	denote the set of program states.\footnote{We choose rationals to have some range of values at hand which are conveniently represented in a computer.}

	The set of expectations, denoted by $\E$, is defined as
	\begin{align*}
		\E \eeq \Bigl \{ f ~\Big|~ f\colon \States \To \PosRealsInf\Bigr\} ~,
	\end{align*}
	where $\PosRealsInf = \setcomp{r \in \Reals}{r \geq 0} \cup \{\infty\}$.
	We say that $f\in \E$ is \emph{finite} and write $f \pprec \infty$, if $f(\State) < \infty$ for all $\State \in \States$.
	A partial order $\preceq$ on $\E$ is obtained by point--wise lifting the usual order~$\leq$ on $\PosRealsInf$, i.e.,
	\begin{align*}
		f_1 \ppreceq f_2 \qiff \forall \State\in\States\colon~~ f_1(\State) ~{}\leq{}~ f_2(\State)~.
	\end{align*}
	$(\E,\, {\preceq})$ is a complete lattice where suprema and infima are constructed point--wise.
\end{definition}
\noindent
We note that our notion of expectations is more general than the one of McIver and Morgan: Their work builds almost exclusively on \emph{bounded} expectations, i.e., non--negative real--valued functions which are bounded from above by some constant, whereas we allow \emph{unbounded} expectations.
As a result, we have that $(\E,\, {\preceq})$ forms a complete lattice, whereas McIver and Morgan's space of bounded expectations does not.

\subsection{Weakest Preexpectations}
\label{Weakest Preexpectations}
Given program $C$ and postexpectation $f \in \E$, we are interested in the \emph{expected value} of $f$ evaluated in the final states reached after termination of $C$.
More specifically, we are interested in a \mbox{\emph{function}~$g\colon \States \To \PosRealsInf$} mapping each \emph{initial} state~$\State_0$ of~$C$ to the respective expected value of~$f$ evaluated in the final states reached after termination of $C$ on input $\State_0$.
This function $g$ is called the \emph{weakest preexpectation of $C$ with respect to $f$}, denoted $\wp{C}{f}$.
Put as an equation, if $\phantom{}^{\State_0}\mu_C$ is the probability (sub)measure\footnote{\label{distribution explanation}
	$\phantom{}^{\State_0}\mu_C(s) \in [0,1]$ is the probability that $s$ is the final state reached after termination of $C$ on input $\State_0$.
	We have $\sum_{\State \in \States} \phantom{}^{\State_0}\mu_C(s) \leq 1$, where the ``missing'' probability mass is the probability of \emph{nontermination} of $C$ on $\State_0$.} over final states reached after termination of~$C$ on initial state~$\State_0$, then\footnote{As $\States$ is countable, the integral can be expressed as $\sum_{\State \in \States} \phantom{}^{\State_0}\mu_C(\State)\cdot f(\State)$.}
\begin{align*}
	g(\State_0) \eeq \wp{C}{f}(\State_0) \eeq \int_\States~f~d\,(\phantom{}^{\State_0}\mu_C)~.
\end{align*}
While $\wp{C}{f}$ in fact represents an expected value, $f$ itself does not.
In an analogy to Dijkstra's pre-- and postconditions, as $f$ is evaluated in the final states after termination of $C$ it is called the \emph{postexpectation}, and as $\wp{C}{f}$ is evaluated in the initial states of $C$ it is called the \emph{preexpectation}.

\subsection{The Weakest Preexpectation Calculus}
We now show how to determine weakest preexpectations in a systematic and \emph{compositional} manner by recapitulating the \emph{weakest preexpectation calculus} \`{a} la McIver and Morgan.
This calculus employs expectation transformers which move backward through the program in a \emph{continuation--passing} style, see \cref{fig:wptrans}.
\begin{figure*}[t!]
	\hrule \vspace{1.5ex}
	\begin{center}
		\begin{adjustbox}{max totalheight=.382\textheight, max width=\textwidth}
			\begin{tikzpicture}[node distance=2.2cm,decoration={snake,pre=lineto,pre length=.5mm,post=lineto,post length=1mm}]
				\node (C) {$~~~~C_2$}; \node (dummy) [left of=C] {}; \node (t) [right of=C] {$f$}; \node (ert) [left of=C] {$\wp{C_2}{f}$}; \draw (t) edge[decorate,->,bend right] (ert); \node (desct) [below of=t]{$
						\begin{array}{c}
							\text{\footnotesize postexpectation $f$}       \\[-.3em]
							\text{\footnotesize evaluated in final states} \\[-.3em]
							\text{\footnotesize after termination of $C_2$}
						\end{array}
					$}; \draw (desct) edge[->,gray] (t); \node (descert) [below of=ert]{$
						\begin{array}{c}
							\text{\footnotesize weakest preexpectation of $C_2$} \\[-.3em]
							\text{\footnotesize with respect to $f$}
						\end{array}
					$}; \draw (descert) edge[->,gray] (ert); \node (C2) [left of=ert] {$~~~~C_1~~~$}; \node (dummy) [left of=C2] {}; \node (ert2) [left of=C2] {$\wp{C_1}{\vphantom{\bigl(}\wp{C_2}{f}}$}; \draw (ert) edge[decorate,->,bend right] (ert2); \node (descert2) [below of=ert2]{$
						\begin{array}{c}
							\text{\footnotesize weakest preexpectation of $C_1$}                \\[-.3em]
							\text{\footnotesize with respect to $\wp{C_2}{f}$}                  \\
							\text{\footnotesize or in other words:}                             \\
							\text{\footnotesize weakest preexpectation of $\COMPOSE{C_1}{C_2}$} \\[-.3em]
							\text{\footnotesize with respect to $f$}
						\end{array}
					$}; \draw (descert2) edge[->,gray] (ert2);
			\end{tikzpicture}
		\end{adjustbox}
	\end{center}
	\vspace{.5ex}
	\hrule \vspace{-2ex}
	\caption{Backward--moving continuation--passing style weakest preexpectation transformer.}
	\vspace{.75ex}
	\hrule \label{fig:wptrans}
\end{figure*}
If we are interested in the expected value of some postexpectation $f$ after executing the sequential composition $\COMPOSE{C_1}{C_2}$, then we can first determine the weakest preexpectation of $C_2$ with respect to $f$, i.e., $\wp{C_2}{f}$. Thereafter, we can use the intermediate result $\wp{C_2}{f}$ as \emph{postexpectation} to determine the weakest preexpectation of $C_1$ with respect to $\wp{C_2}{f}$. Overall, this gives the weakest preexpectation of $\COMPOSE{C_1}{C_2}$ with respect to the postexpectation $f$. The above explanation illustrates the compositional nature of the weakest preexpectation calculus. $\wpsymbol$--transformers for all language constructs can be defined by induction on \mbox{the program structure:}%
\begin{definition}[The \textnormal{$\boldwpsymbol$}--Transformer~\textnormal{\cite{DBLP:series/mcs/McIverM05}}]
	\label{defQQQwp}
	Let $\pgcl$ be the set of programs in the \emph{probabilistic guarded command language}. 
	Then the weakest~\mbox{preexpectation~transformer}
	\begin{align*}
		\wpsymbol\colon \pgcl \To \E \To \E
	\end{align*}
	is defined according to the rules given in \textnormal{\cref{table:wp}}, where $\iverson{\guard}$ denotes the Iverson--bracket of $\guard$, i.e., $\iverson{\guard}(\State)$ evaluates to $1$ if $\State \models \guard$ and to $0$ otherwise. Moreover, for any variable $b \in \Vars$ and any expression~$e$, let $f\subst{b}{e}$ be the expectation with $f\subst{b}{e}(\State) = f(\State\subst{b}{e})$ for any $\State \in \States$, where $\State\subst{b}{e}(b) = \State(e)$ and $\State\subst{b}{e}(x) = \State(x)$ for all $x \in \Vars\setminus\{b\}$.
	\begin{table}[t]
		\renewcommand{\arraystretch}{1.5}
		\begin{tabular}{@{\hspace{.1em}}l@{\hspace{1em}}l@{\hspace{.1em}}}
			\hline $\boldsymbol{C}$        & $\boldwp{C}{f}$                                                                \\[.25em]
			\hline $\SKIP$                 & $f$                                                                            \\
			$\ASSIGN{b}{e}$                & $f\subst{b}{e}$                                                                \\
			$\ITE{\guard}{C_1}{C_2}$       & $\iverson{\guard} \cdot \wp{C_1}{f} + \iverson{\neg \guard} \cdot \wp{C_2}{f}$ \\
			$\PCHOICE{C_1}{p}{C_2}$        & $p \cdot \wp{C_1}{f} + (1 - p) \cdot \wp{C_2}{f}$                              \\
			$\COMPOSE{C_1}{C_2}$           & $\wp{C_1}{\wp{C_2}{f}}$                                                        \\
			$\WHILEDO{\guard}{\primed{C}}$ & $\lfp \charwp{\primed{C}}{\guard}{f}$                                          \\[.75em]
			\hline                                                                                                          \\[-1.25em]
			\multicolumn{2}{c}{$\charwp{\guard}{C}{f}(X) \eeq \iverson{\neg \guard}\cdot f + \iverson{\guard} \cdot \wp{C}{X}\qquad
					\begin{array}{c}
						\textnormal{\tiny characteristic} \\[-1.2em]
						\textnormal{\tiny function}
					\end{array}
			$}                                                                                                              \\
			\\[-1em]
		\end{tabular}
		\hrule \vspace{1.25ex}
		\caption{Rules for the $\wpsymbol$--transformer.}
		\label{table:wp}
		\vspace{-1.5ex}
		\hrule
	\end{table}

	We call the function $\charwp{\guard}{C}{f}$ the \emph{characteristic function} of $\WHILEDO{\guard}{C}$ with respect to $f$. Its least fixed point is understood in terms of $\preceq$. To increase readability, we omit $\wpsymbol$, $\guard$, $C$, or $f$ from $\Phi$ whenever they are clear from the context.
\end{definition}
\begin{example}[Applying the $\boldwpsymbol$ Calculus]
	\label{ex:wp-calc-a}
	Consider the probabilistic program $C$ given by
	\begin{align*}
		\PCHOICE{\ASSIGN{b}{b+5}}{\sfrac{4}{5}}{\ASSIGN{b}{10}}~.
	\end{align*}
	Suppose we want to know the expected value of $b$ after execution of $C$. For this, we determine $\wp{C}{b}$. Using the annotation style shown in \textnormal{\cref{fig:wp-annotations}}, we can annotate the program~$C$ as shown in \textnormal{\cref{fig:ex:wp-calc-a}}, using the rules from \textnormal{\cref{table:wp}}. At the top, we read off the weakest preexpectation of $C$ with respect to $b$, namely $\smash{\tfrac{4b}{5} + 6}$. This tells us that the expected value of $b$ after {termination of $C$ on $s_0$ is equal to $\smash{\tfrac{4\cdot s_0(b)}{5} + 6}$}.
\end{example}
\noindent
The $\wpsymbol$--transformer satisfies what is sometimes called \emph{healthiness conditions}~\cite{DBLP:conf/lics/HinoKH016,DBLP:journals/entcs/Keimel15,DBLP:series/mcs/McIverM05}
or \emph{homomorphism properties}~\cite{DBLP:books/daglib/0096285}:%
\begin{theorem}[Healthiness Conditions~\textnormal{\cite{thesis:kaminski,DBLP:series/mcs/McIverM05}}]
	\label{thm:healthiness}
	Let $C \in \pgcl$, $S = \{s_0 \preceq s_1 \preceq s_2 \preceq \cdots\} \subseteq \E$,\footnote{That is, $S$ is a chain.} $f, g \in \E$, and $r \in \PosReals$. Then:
	\begin{enumerate}
		\item \textbf{\textup{Continuity:}} \quad
		      $\wp{C}{\sup~ S} \eeq \sup~ \wp{C}{S}$~.

		\item \textbf{\textup{Strictness:}}\footnote{Here, we overload notation and denote by $0$ the constant expectation that maps every $\State \in \States$ to $0$.} \quad $\wpC{C}$ is strict, i.e.,
		      $\wp{C}{0} = 0$.

		\item \textbf{\textup{Monotonicity:}} \quad $f \ppreceq g \qimplies \wp{C}{f} \ppreceq \wp{C}{g}$.
		\item \textbf{\textup{Linearity:}} \quad
		      $\wp{C}{r \cdot f + g} \eeq r \cdot \wp{C}{f} + \wp{C}{g}$.
	\end{enumerate}
\end{theorem}

\noindent

\section{Bounds on Weakest Preexpectations}
\label{sec:bounds}
For \emph{loop--free} programs, it is generally straightforward to determine weakest preexpectations, simply by applying the rules in \cref{table:wp}, which guide us along the syntax of $C$, see \cref{ex:wp-calc-a}. Weakest preexpectations of loops, on the other hand, are generally non--computable least fixed points and we often have to content ourselves with some \emph{approximation} of those fixed points.

For us, a sound approximation is either a \emph{lower} or an \emph{upper} bound on the least fixed point. There are in principle two challenges: (1) finding a candidate bound and (2) verifying that the candidate is indeed an upper or lower bound. In this paper, we study the \emph{latter} problem.

\subsection{Upper Bounds}

\begin{figure}[t]
	\hrule \vspace{.05cm}
	\begin{subfigure}[t]{6.1cm}
		\begin{minipage}{6.1cm}
			\abovedisplayskip=-0em%
			\belowdisplayskip=0pt%
			\begin{align*}
				 & \bowtieannotate{g'} \tag{meaning $g' \bowtie g$, for ${\bowtie} \in \{{\preceq},{=},{\succeq}\}$} \\
				 & \wpannotate{g} \tag{meaning $g = \wp{C'}{f}$}                                                     \\
				 & C'                                                                                                \\
				 & \annotate{f} \tag{postexpectation is $f$}
			\end{align*}
			\normalsize%
		\end{minipage}
		\vspace{1ex}
		\hrule \vspace{-.5ex}
		\subcaption{Style for $\wpsymbol$ annotations.}
		\label{fig:wp-annotations}
	\end{subfigure}
	\qquad\qquad
	\begin{subfigure}[t]{4.2cm}
		\begin{minipage}{4.2cm}
			\abovedisplayskip=0em%
			\belowdisplayskip=0pt%
			\begin{align*}
				 & \eqannotate{\tfrac{4b}{5} + 6}                                  \\
				 & \wpannotate{\tfrac{4}{5} \cdot (b + 5) + \tfrac{1}{5} \cdot 10} \\
				 & \PCHOICE{\ASSIGN{b}{b+5}}{\sfrac{4}{5}}{\ASSIGN{b}{10}}         \\
				 & \annotate{b}\tag*{\phantom{postexpectation is $f$}}
			\end{align*}
			\normalsize%
		\end{minipage}
		\vspace{1ex}
		\hrule \vspace{-.5ex}
		\caption{$\wpsymbol$ annotations for \cref{ex:wp-calc-a}.}
		\label{fig:ex:wp-calc-a}
	\end{subfigure}
	\vspace{1ex}
	\hrule \vspace{-1.5ex}
	\caption{Annotations for weakest preexpectations. It is more intuitive to read these from the bottom to top.}
	\label{fig:annotationzz}
	\vspace{1ex}
	\hrule
\end{figure}
\noindent
The \emph{Park induction} principle provides us with a very convenient proof rule for verifying upper bounds. In general, this principle reads as follows:%
\begin{theorem}[Park Induction~\textnormal{\protect{\cite{park}}}]
	\label{thm:park-induction}
	Let $(D,\, {\sqsubseteq})$ be a complete lattice and let $\Phi\colon D \To D$ be continuous.\footnote{It would even suffice for $\Phi$ to be monotonic, but we consider continuous functions throughout this paper.}
	Then $\Phi$ has a least fixed point in $D$ and for any $I \in D$,
	\begin{align*}
		\Phi(I) \ssqsubseteq I \qqimplies \lfp\, \Phi \ssqsubseteq I~.
	\end{align*}
\end{theorem}
\noindent
In the realm of weakest preconditions, Park induction gives rise to the following induction principle:%
\begin{corollary}[Park Induction for $\boldwpsymbol$~\textnormal{\cite{DBLP:journals/jcss/Kozen85,thesis:kaminski}}]
	\label{Park Induction wp}
	Let $\Phi_f$ be the characteristic function of the while loop $\WHILEDO{\guard}{C}$ with respect to postexpectation $f$ and let $I \in \E$. Then
	\begin{align*}
		\Phi_f(I) \ppreceq I \qimplies \wp{\WHILEDO{\guard}{C}}{f} \ppreceq I~.
	\end{align*}
\end{corollary}
\noindent
We call an $I$ that satisfies $\Phi_f(I) \preceq I$ a \emph{superinvariant}. The striking power of Park induction is its simplicity: Once an appropriate candidate $I$ is found (even though this is usually not an easy task), all we have to do is push it through the characteristic function $\Phi_f$ \emph{once} and check whether we went down in our underlying partial order. If this is the case, we have verified that $I$ is indeed an upper bound on the least fixed point and thus on the sought--after weakest preexpectation.%
\begin{example}[Induction for Upper Bounds]
	\label{ex:wp-upper-bounds}
	Consider the program $C_\mathit{geo}$, given by
	\begin{align*}
		 & \WHILEDO{a \neq 0}{ \PCHOICE{\ASSIGN{a}{0}}{\sfrac{1}{2}}{\ASSIGN{b}{b + 1}} }~,
	\end{align*}
	where we assume $b \in \Nats$. Suppose we aim at an upper bound on the expected value of $b$ executing~$C_\mathit{geo}$. Using the annotation style of \textnormal{\cref{fig:loop-annotations}}, we can annotate the loop $C_\mathit{geo}$ as shown in \textnormal{\cref{fig:ex:wp-upper-bounds}}, using superinvariant $I = b + \iverson{a \neq 0}$, establishing $\Phi_b(I) \preceq I$, and by \textnormal{\cref{Park Induction wp}} establishing $\wp{C_\mathit{geo}}{b} \preceq b + \iverson{a \neq 0}$. So the expected value of $b$ after termination of $C_\mathit{geo}$ on $s_0$ is at most $\smash{s_0(b) + \iverson{s_0(a) \neq 0}}$.
\end{example}
\begin{figure}[t]
	\hrule
	\begin{subfigure}[t]{.475\linewidth}
		\begin{minipage}{1\linewidth}
			\begin{align*}
				 & \bowtieannotate{I} \tag{see \cref{fig:wp-annotations}}                                         \\
				 & \phiannotate{g} \tag{meaning $g = \iverson{\neg \guard} \cdot f + \iverson{\guard} \cdot I''$} \\
				 & \WHILE{\guard}                                                                                 \\
				 & \qquad \bowtieannotate{I''} \tag{see \cref{fig:wp-annotations}}                                \\
				 & \qquad \wpannotate{I'} \tag{see \cref{fig:wp-annotations}}                                     \\
				 & \qquad \mathit{Body}                                                                           \\
				 & \qquad \starannotate{I} ~\}\tag{meaning we employ invariant $I$}                               \\
				 & \annotate{f} \tag{postexpectation of loop is $f$}
			\end{align*}
		\end{minipage}
		\vspace{1ex}
		\hrule \vspace{-.5ex}
		\subcaption{Annotation style for loops using invariants.}
		\label{fig:loop-annotations}
	\end{subfigure}
	\quad\hfill
	\begin{subfigure}[t]{.475\linewidth}
		\begin{adjustbox}{max width=1\linewidth}
			\begin{minipage}{1\linewidth}
				\begin{align*}
					 & \preceqannotate{b + \iverson{a \neq 0}}                                                                                                \\
					 & \phiannotate{\iverson{a = 0} \cdot b \pplus \iverson{a \neq 0} \cdot \Bigl( b + \tfrac{1}{2} \bigl(1 + \iverson{a \neq 0}\bigr)\Bigr)} \\
					 & \WHILE{a \neq 0}                                                                                                                       \\
					 & \qquad \eqannotate{b \pplus \tfrac{1}{2} \bigl(1 + \iverson{a \neq 0}\bigr)}                                                           \\
					 & \qquad \wpannotate{\tfrac{1}{2} \bigl( b + \iverson{0 \neq 0} \pplus b + 1 + \iverson{a \neq 0} \bigr)}                                \\
					 & \qquad \!\!\PCHOICE{\ASSIGN{a}{0}}{\sfrac{1}{2}}{\ASSIGN{b}{b + 1}}                                                                    \\
					 & \qquad \starannotate{b + \iverson{a \neq 0}} ~\}                                                                                       \\
					 & \annotate{b}
				\end{align*}
			\end{minipage}
		\end{adjustbox}
		\vspace{1ex}
		\hrule \vspace{-.5ex}
		\subcaption{$\wpsymbol$ loop annotations for \cref{ex:wp-upper-bounds}.}
		\label{fig:ex:wp-upper-bounds}
	\end{subfigure}
	\vspace{1ex}
	\hrule \vspace{-1.5ex}
	\caption[]{Annotation style for loops using invariants and annotations for \cref{ex:wp-upper-bounds}. Inside the loop, we push an invariant $I$ (provided externally, denoted by $\boldsymbol{\annocolor{{\talloblong}\!{\talloblong}\:I}}$) through the loop body, thus obtaining $I''$ which is (possibly an over- or underapproximation of) $\wp{\mathit{Body}}{I}$. Above the loop head, we then annotate $g = \iverson{\neg\guard} \cdot f + \iverson{\guard} \cdot I''$. In the first line, we establish $g \bowtie I$, for~\mbox{${\bowtie} \in \{{\preceq},\, {\succeq}\}$}. Note that if ${\bowtie} = {\preceq}$, we have established the precondition of \cref{Park Induction wp}, since we have then overall established
	\begin{minipage}{\linewidth}
		\begin{align*}
			\Phi_f(I) \eeq \iverson{\neg\guard} \cdot f + \iverson{\guard} \cdot \wp{\mathit{Body}}{I}
			\ppreceq \iverson{\neg\guard} \cdot f + \iverson{\guard} \cdot I'' \eeq g \ppreceq I \\[-.9em]
		\end{align*}
	\end{minipage}
	For reasoning about lower bounds, we will later employ ${\bowtie} = {\succeq}$.%
		}
	\label{fig:loop-annotationzzz}
	\vspace{1ex}
	\hrule
\end{figure}
\noindent
Let us explain why Park induction is sound using the so-called \emph{Tarski-Knaster Fixed Point Theorem}:\footnote{
	In \cite{popl}, we explain Park induction via an extended version of \cref{thm:tarski-kantorovich}. However, this explanation was incorrect and is fixed here. This does not have any consequences regarding the results of \cite{popl}.}

\begin{theorem}[Tarski-Knaster Fixed Point Theorem \textnormal{\cite{tarski_fixpoint,knaster_fixpoint}}]
	\label{thm:tarski-knaster}
	Let $(D,\, {\sqsubseteq})$ be a complete lattice and $\Phi\colon D \to D$ be continuous.
	Then the set of fixed points of $\Phi$ is a complete lattice.
	Here, $\inf \{U \in D \mid \Phi(U) \sqsubseteq U\}$ is the least fixed point of $\Phi$ and $\sup \{I \in D \mid I \sqsubseteq \Phi(I)\}$ is the greatest fixed point of $\Phi$.
\end{theorem}
\noindent
Again, for this theorem it would even suffice for $\Phi$ to be only monotonic.

In our setting, the characterization of the least fixed point in \cref{thm:tarski-knaster} proves soundness of Park induction (recall that $\Phi$ is continuous, see \cref{thm:healthiness}).
For making a comparison to the lower bound case which we consider later, let us introduce the so-called \emph{Tarski-Kantorovich Principle}:

\begin{theorem}[Tarski-Kantorovich Principle\textnormal{, see \cite{jachymski2000tarski}}]
	\label{thm:tarski-kantorovich}
	Let $(D,\, {\sqsubseteq})$ be a complete lattice, $\Phi\colon D \to D$ be continuous, and let $I \in D$ such that $I \sqsubseteq \Phi(I)$.
	Then the sequence $I \sqsubseteq \Phi(I) \sqsubseteq \Phi^2(I) \sqsubseteq \Phi^3(I) \sqsubseteq {\cdots}$ is an \emph{ascending chain} that \emph{converges} to an element
	\[
		\Phi^\omega(I) = \lim_{n \to \omega} \Phi^n(I) = \sup \{\Phi^n(I) \mid n \in \Nats\} \in D,
	\]
	which is a fixed point of $\Phi$. In particular, $\Phi^\omega(I)$ is the least fixed point of $\Phi$ that is ${}\sqsupseteq{} I$.
\end{theorem}
\noindent
The well--known Kleene Fixed Point Theorem (cf.~\cite{DBLP:journals/ipl/LassezNS82}), which states that $\lfp \Phi = \Phi^\omega(\bot)$, where $\bot$ is the least element of~$D$, is a special case of the Tarski--Kantorovich Principle.

\subsection{Lower Bounds}
\label{sec:lower-bounds}
\lowerpenalties

For verifying lower bounds, we do not have a rule as simple as Park induction available. In particular, for a given complete lattice $(D,\, {\sqsubseteq})$ and monotonic function $\Phi\colon D \To D$, the rule
\begin{align*}
	I \ssqsubseteq \Phi(I) \qqimplies I \ssqsubseteq \lfp\, \Phi~, \tag*{\Large\lightning}
\end{align*}
is \emph{unsound} in general. We call an $I$ satisfying $I \sqsubseteq \Phi(I)$ a \emph{subinvariant} and the above rule \emph{simple lower induction}. Generally, we will call an $I$ that is a sub-- or a superinvariant an \emph{invariant}. $I$~being an invariant thus expresses mainly its \emph{inductive} nature, namely that $I$ is comparable with~$\Phi(I)$ with respect to the partial order $\sqsubseteq$.

An explanation why simple lower induction is unsound is as follows: By \cref{thm:tarski-kantorovich}, we know from $I \sqsubseteq \Phi(I)$ that $\Phi^\omega(I)$ is the \emph{least fixed point of $\Phi$ that is greater than or equal to $I$}. Since $\Phi^\omega(I)$ is a fixed point, $\Phi^\omega(I) \sqsubseteq \gfp \Phi$ holds, but we do not know how $I$ compares to $\lfp \Phi$. We only know that if indeed $I \sqsubseteq \lfp \Phi$ and $I \sqsubseteq \Phi(I)$, then iterating $\Phi$ on $I$ also \mbox{converges to $\lfp \Phi$, i.e.,}
\begin{align*}
	I \ssqsubseteq \lfp \Phi \qand I \ssqsubseteq \Phi(I) \qqimplies \Phi^\omega(I) \eeq \lfp \Phi~.
\end{align*}
If, however, $I \sqsubseteq \Phi(I)$ and $I$ is \emph{strictly greater} than $\lfp \Phi$, then iterating $\Phi$ on $I$ will yield a fixed point strictly greater than $\lfp \Phi$, contradicting soundness of simple lower induction.

While we just illustrated by means of the Tarski--Kan{\-}to{\-}ro{\-}vich principle why the simple lower induction rule is not sound in general, we should note that the rule is not per se absurd: So called \emph{metering functions} \cite{DBLP:conf/cade/FrohnNHBG16} basically employ simple lower induction to verify \emph{lower bounds} on runtimes of nonprobabilistic programs \cite[Thm.~7.18]{thesis:kaminski}. For weakest preexpectations, however, simple lower induction is unsound:%
\begin{counterexample}[Simple Induction for Lower Bounds]
	\label{ex:simple-induction}
	Consider the following loop $C_{\mathit{cex}}$, where $b,k \in \Nats$
	\begin{align*}
		 & \WHILE{a \neq 0}                                                            \\
		 & \qquad \COMPOSE{\PCHOICE{\ASSIGN{a}{0}}{\sfrac{1}{2}}{\ASSIGN{b}{b + 1}}}{} \\
		 & \qquad \ASSIGN{k}{k + 1}                                                    \\
		 & \}~.
	\end{align*}
	As in \textnormal{\cref{ex:wp-upper-bounds}}, $\wp{C_{\mathit{cex}}}{b} = b + \iverson{a \neq 0}$.
	In particular, this weakest preexpectation is independent of $k$.
	The corresponding characteristic function is
	\begin{align*}
		\Phi_b(X) \eeq & \iverson{a = 0} \cdot b \pplus \iverson{a \neq 0} \cdot \tfrac{1}{2} \cdot \bigl(X\subst{a}{0}
		\pplus X\subst{b}{b + 1}\bigr)\subst{k}{k+1}.
	\end{align*}
	Let us consider $\primed{I} = b + \iverson{a \neq 0} \cdot (1 + 2^k)$, which \emph{does} depend on $k$. 
	Indeed, one can check that $\primed{I} \preceq \Phi_b(\primed{I})$, i.e., $I'$ is a subinvariant. If the simple lower induction rule were sound, we would immediately conclude that $\primed{I}$ is a \emph{lower bound} on $\wp{C_{\mathit{cex}}}{b}$, but this is obviously false since
	\begin{align*}
		\wp{C_{\mathit{cex}}}{b} = b + \iverson{a \neq 0} ~{}\prec{}~ \primed{I}~.
	\end{align*}
\end{counterexample}

\subsection{Problem Statement}
\label{sec:purpose}

The purpose of this paper is to present a \emph{sound lower induction rule} of the following form: Let $\Phi_f$ be the characteristic function of $\WHILEDO{\guard}{C}$ with respect to $f$ and let $I \in \E$. Then
\begin{align*}
	I \ppreceq \Phi_f(I) ~{}\wedge{}~
	\begin{array}{c}
		\textnormal{\small some side} \\[-.5ex]
		\textnormal{\small conditions}
	\end{array}
	\qqimplies I \ppreceq \lfp\, \Phi_f~.
\end{align*}
We still want our lower induction rule to be simple in the sense that checking the side conditions should be conceptually as simple as checking $I \preceq \Phi_f(I)$. Intuitively, we want to apply the semantics of the loop body only \emph{finitely often}, not $\omega$ times, to avoid reasoning about limits of sequences or anything alike. We provide such side conditions in our main contribution, \cref{thm:optional_stopping_probabilistic_programs}, which transfers the Optional Stopping Theorem of probability theory to weakest preexpectation reasoning.

\renewcommand{\charwp}[3]{\Phi_{#3}}
\renewcommand{\charwpn}[4]{\Phi_{#3}^{#4}}

\subsection{Uniform Integrability}
\label{sec:series_perspective_on_loops}
\unlowerpenalties

\noindent
We now present a sufficient and necessary criterion to under--approximate the least fixed points that we seek for.
Let again $\charwp{\guard}{C}{f}$
be the characteristic function of $\WHILEDO{\guard}{C}$ with respect to $f$.
\cref{thm:healthiness} implies that $\charwp{\guard}{C}{f}$ is continuous and monotonic.

Let us consider a \emph{subinvariant} $I$, i.e., $I \preceq \charwp{\guard}{C}{f}(I)$.
If we iterate $\charwp{\guard}{C}{f}$ on $I$ ad infinitum, then the Tarski--Kantorovich principle (\cref{thm:tarski-kantorovich}) guarantees that we will converge to some fixed point~$\charwpn{\guard}{C}{f}\omega(I)$ that is ${\succeq}\,I$.
From continuity of $\charwp{\guard}{C}{f}$ and \cref{thm:tarski-kantorovich}, one can easily show that $\charwpn{\guard}{C}{f}\omega(I)$ coincides with $\lfp \charwp{\guard}{C}{f}$ \emph{if and only if} $I$ itself was already ${\preceq}\, \lfp \charwp{\guard}{C}{f}$, i.e.:
\begin{theorem}[Subinvariance and Lower Bounds]
	\label{thm:uniform_integrability_least_fixed_point}
	For any \emph{subinvariant} $I$, we have
	\begin{align*}
		\charwpn{\guard}{C}{f}\omega(I) \eeq \lfp \charwp{\guard}{C}{f} \qqiff I \ppreceq \lfp \charwp{\guard}{C}{f}~.
	\end{align*}
\end{theorem}
\noindent
More generally, for \emph{any} expectation $X$ (not necessarily a sub- or superinvariant), if iterating $\charwp{\guard}{C}{f}$ on $X$ converges to the least fixed point of $\charwp{\guard}{C}{f}$, then we call $X$ \emph{uniformly integrable for $f$}:%
\begin{definition}[Uniform Integrability of Expectations]
	\label{def:uniform_integrability_invariants}
	Given a loop $\WHILEDO{\guard}{C}$,
	an expectation $X \in \E$ is called \emph{uniformly integrable (u.i.) for $f \in \E$} if $\lim\limits_{n \to \omega} \charwpn{\guard}{C}{f}n(X)$ exists and
	\begin{align*}
		\lim\limits_{n \to \omega} \charwpn{\guard}{C}{f}n(X) \eeq \lfp \charwp{\guard}{C}{f}~.
	\end{align*}
\end{definition}
\noindent
So far, we have thus established the following diagram which we will gradually extend over the next two sections:
\begin{center}
	\begin{adjustbox}{max width=\linewidth}
		\begin{tikzcd}[column sep = huge, row sep=large]
			I~\textnormal{u.i.\ for}~f \quad \arrow[r, Leftrightarrow, "\textnormal{\raisebox{1ex}{\cref{def:uniform_integrability_invariants}}}"] \arrow[d, Leftrightarrow, "\textnormal{
				\begin{tabular}{l}
					\cref{thm:uniform_integrability_least_fixed_point} \\[-1ex]
					and \cref{def:uniform_integrability_invariants}
				\end{tabular}
			}" shift={(0,-0.1)}]
			&\quad\charwpn{\guard}{C}{f}n(I) \xrightarrow{~n \to \omega~} \lfp \charwp{\guard}{C}{f}
			\\
			I \preceq \charwp{\guard}{C}{f}(I) \Rightarrow I \preceq \lfp \charwp{\guard}{C}{f}
		\end{tikzcd}
	\end{adjustbox}
\end{center}
Uniform integrability~\cite{grimmett2001probability} --- a notion originally from probability theory --- will be essential for the Optional Stopping Theorem in \cref{sec:optional-stopping-wp}. While, so far, we have studied the function $\charwp{\guard}{C}{f}$ solely from an expectation transformer point of view and defined a purely expectation--theoretical notion of uniform integrability without any reference to probability theory, we will study in \cref{sec:expectations_processes} the function $\charwp{\guard}{C}{f}$ from a stochastic process point of view. Stochastic processes are not inductive per se, whereas expectation transformers make heavy use of induction. We will, however, rediscover the inductiveness also in the realm of stochastic processes. We will also see how our notion of uniform integrability corresponds to uniform integrability in its original sense.

\renewcommand{\charwp}[3]{\charfun{\wpsymbol}{#1}{#2}{#3}}
\renewcommand{\charwpn}[4]{\charfunn{\wpsymbol}{#1}{#2}{#3}{#4}}

\renewcommand{\charwp}[3]{\Phi_{#3}}
\renewcommand{\charwpn}[4]{\Phi_{#3}^{#4}}

\section{From Expectations to Stochastic Processes}
\label{sec:expectations_processes}
\allowdisplaybreaks
\noindent
In this section, we connect concepts from expectation transformers with notions from probability theory.
In \cref{sec:Canonical Probability Space}, we recapitulate standard constructions of probability spaces for probabilistic programs, instantiate them in our setting, and present our new results on connecting expectation transformers with stochastic processes (\cref{sec:Canonical Stochastic Process}) and uniform integrability (\cref{sec:Uniform Integrability}).
Proofs can be found in \techrep{\cref{app:proofs_expectations_processes}}\cameraready{\cite[App.\ C]{arxiv}}.
For further background on probability theory, we refer to \techrep{\cref{app:prelim_probability}}\cameraready{\cite[App.\ B]{arxiv}} and \cite{bauer71measure,grimmett2001probability}.

We fix for this section an arbitrary loop
$\WHILEDO{\guard}{C}$.
The loop body $C$ may contain loops but we require $C$ to be \emph{universally almost--surely terminating (AST)}, i.e., $C$ terminates on any input with probability 1.
The set of program states can be uniquely partitioned into
$\States = \States_{\guard} \uplus \States_{\neg \guard}$,
with $\State \in \States_{\guard}$ iff $\State \models \guard$.
The set $\States_{\neg \guard}$ thus contains the terminal states from which the loop is not executed further.

\subsection{Canonical Probability Space}
\label{sec:Canonical Probability Space}

\noindent
We begin with constructing a canonical probability measure and space corresponding to the execution of our loop.
As every \pgcl program is, operationally, a countable Markov chain, our construction is similar to the standard construction for Markov chains (cf.~\cite{DBLP:conf/focs/Vardi85}).

In general, a \emph{measurable space} is a pair $(\Omega, \F{F})$ consisting of a \emph{sample space} $\Omega$ and a \emph{$\sigma$--field} $\F{F}$ of $\Omega$, which is a collection of subsets of $\Omega$, closed under complement and countable union, such that~$\Omega \in \F{F}$.
In our setting, a loop $\WHILEDO{\guard}{C}$ induces the following canonical measurable space:%
\begin{definition}[Loop Space]
	The loop $\WHILEDO{\guard}{C}$ induces a unique measurable space~$(\Omega^\lp,\, \F{F}^\lp)$ with sample space
	$ \Omega^\lp ~{}\coloneqq{}~ \States^{\omega} \eeq \{\run \colon \Nats \to \States\}, $
	i.e., it is the set of all \emph{infinite} sequences of program states (so--called \emph{runs}).
	For $\run \in \Omega^\lp$, we denote by $\run[n]$ the $n$--th state in the sequence $\run$ (starting to count at 0).
	The $\sigma$--field $\F{F}^\lp$ is the smallest $\sigma$--field that contains all cylinder sets $Cyl(\pi) = \setcomp{ \pi \run }{ \run \in \States^\omega }$, for all finite prefixes $\pi \in \States^+$, denoted as
	\begin{align*}
		\F{F}^\lp \eeq \sigmagen{\setcomp{ Cyl(\pi) }{ \pi \in \States^+ }}~.
	\end{align*}
\end{definition}
\noindent
Intuitively, a run $\run \in \Omega$ is an infinite sequence of states
$ \run \eeq \State_0 \, \State_1 \, \State_2 \, \State_3 \, {{}\cdots{}}~, $
where $\State_0$ represents the initial state on which the loop is started and $\State_i$ is a state that could be reached after $i$ iterations of the loop.
Obviously, some sequences in $\Omega^\lp$ may not actually be admissible by our loop.

We next develop a canonical probability measure corresponding to the execution of the loop, which will assign the measure $0$ to inadmissible runs.
We start with considering a single loop iteration.
The loop body $C$ induces a \emph{family} of distributions\footnote{Since the loop body $C$ is AST, these are distributions and not subdistributions.}
\begin{align*}
	\phantom{}^{\bullet}\mu_{C}\colon \States \to \States \to [0,1]~,
\end{align*}
where $\phantom{}^{\State}\mu_{C}(\primed{\State})$ is the probability that after \emph{one} iteration of $C$ on~$\State$, the program is in state $\primed{\State}$.

The loop $\WHILEDO{\guard}{C}$ induces a \emph{family} of probability measures on $(\Omega^\lp,\, \F{F}^\lp)$.
This family is parameterized by the initial state of the loop.
Using the distributions $\phantom{}^{\bullet}\mu_{C}$ above, we first define the probability of a finite non--empty prefix of a run, i.e., for $\pi \in \Sigma^+$.
Here, ${}^{\State}p(\pi)$ is the probability that $\pi$ is the sequence of states reached after the first loop iterations, when starting the loop in state $\State$.
Hence, the family
\begin{align*}
	{}^{\bullet}{}p\colon \States \to \States^+ \to [0,1]
\end{align*}
of distributions on $\States^+$ is defined by
\begin{enumerate}
	\item $^{\State}p(\primed{\State})=\iverson{\State=\primed{\State}}$

	\item $^{\State}p(\pi \primed{\State}\ddashed{\State}) =
		      \begin{cases}
			      \phantom{}^{\State}p(\pi\primed{\State}) \cdot \iverson{\ddashed{\State}=\primed{\State}},             & \text{if } \primed{\State} \in \States_{\neg \guard} \\

			      \phantom{}^{\State}p(\pi\primed{\State}) \cdot \phantom{}^{\primed{\State}}\mu_{C}(\ddashed{\State}) , & \text{if } \primed{\State} \in \States_{\guard}
		      \end{cases}
	      $.
\end{enumerate}
Using the family ${}^{\bullet}{}p$, we now obtain a canonical probability measure on the loop space.%
\begin{lemma}[Loop Measure~\textnormal{\protect{\cite[Kolmogorov's Extension Theorem]{feller-vol-2}}}]
	There exists a unique family of probability measures ${}^{\bullet}{}\mathbb{P}\colon \States \to \F{F} \to [0,1]$ with
	\begin{align*}
		\IP{\State}{Cyl(\pi)} \eeq {}^{\State}p(\pi)~.
	\end{align*}
\end{lemma}
\noindent
We now turn to random variables and their expected values.
A mapping $X \colon \Omega \to \smash{\PosRealsInf}$ on a probability space $\smash{(\Omega, \F{F}, \mathbb{P})}$ is called \mbox{($\F{F}$--)}\emph{measurable} or \emph{random variable} if for any open set $U \subset \PosRealsInf$ its preimage lies in $\F{F}$, i.e., $X^{-1}(U) \in \F{F}$.
If $X(\Omega) \subset \NatsInf = \Nats \cup \{
	\omega \}$, then this is equivalent to checking $X^{-1}(\{n\}) \in \F{F}$ for any $n \in \Nats$.
The \emph{expected value} $\expec{}{X}$ of a random variable $X$ is defined as $\expec{}{X} \coloneqq \int_{\Omega} X d \mathbb{P}$.\footnote{Details on integrals for arbitrary measures can be found in \techrep{\cref{app:prelim_probability}}\cameraready{\cite[App.\ B]{arxiv}}.} If $X$ takes only countably many values we have
\begin{align*}
	\expec{}{X} \eeq \int_{\Omega} X \, d \mathbb{P} \eeq \sum\limits_{r \in X(\Omega)} \IP{}{X^{-1}(\{r\})} \cdot r~.
\end{align*}
We saw that $\WHILEDO{\guard}{C}$ gives rise to a unique canonical measurable space $(\Omega^\lp,\, \F{F}^\lp)$ and to a
family of probability measures ${}^{\State}\mathbb{P}$ parameterized by the initial state $\State$ on which our loop is started.
We now define a corresponding parameterized expected value~\mbox{operator~${}^{\bullet}{}\mathbb{E}$}.%
\begin{definition}[Expected Value for Loops $\boldsymbol{{}^{\bullet}{}\mathbb{E}}$]
	Let $\State \in \States$ and \mbox{$X\colon \Omega^\lp \to \PosRealsInf$} be a random variable.
	The \emph{expected value of $X$} with respect to the loop measure ${}^{\State}\mathbb{P}$, parameterized by state~$\State$, is defined by $\expec{\State}{X}\coloneqq\int_{\Omega} X \, d \left({}^{\State}\mathbb{P}\right)$.
\end{definition}
\noindent
Next, we define a random variable that corresponds to the number of iterations that our loop makes until it terminates.%
\begin{definition}[Looping Time]
	\label{def:looping-time}
	The mapping
	\begin{align*}
		\termtime{\guard} \colon\quad \Omega^\lp \to \overline{\Nats},~ \run \mapsto \inf \{n \in \Nats \mid \run[n] \in \States_{\neg \guard}\}~,
	\end{align*}
	is a random variable and called the \emph{looping time} of $\WHILEDO{\guard}{C}$.
	Here, $\overline{\Nats} = \Nats \cup \{\omega\}$ and $\inf \emptyset = \omega$.
\end{definition}
\noindent
The canonical $\sigma$--field $\F{F}^\lp$ contains infinite runs.
But after $n$~iterations of the loop we only know the first $n+1$ states $\State_0\cdots \State_n$ of a run.
Gaining knowledge in this successive fashion can be captured by a so--called \emph{filtration} of the $\sigma$--field $\F{F}^\lp$.
In general, a filtration is a sequence $(\F{F}_n)_{n\in \Nats}$ of subsets of $\F{F}$, such that $\F{F}_n \subset \F{F}_{n+1}$ and $\F{F}_n$ is a $\sigma$--field for any $n \in \Nats$, i.e., $\F{F}$ is approximated from below.%
\begin{definition}[Loop Filtration]
	\label{def:canonical-filtration}
	The sequence $(\F{F}_n^\lp)_{n \in \Nats}$ is a filtration of $\F{F}^\lp$, where
	\begin{align*}
		\F{F}_n^\lp \eeq \sigmagen{\{Cyl(\pi) \mid \pi \in \States^+,~|\pi| = n+1\}},
	\end{align*}
	i.e., $\F{F}_n^\lp$ is the smallest $\sigma$--field containing $\{Cyl(\pi) \mid \pi \in \States^+,~|\pi| = n+1\}$.\footnote{Note that here $\F{F}^\lp=\bigcup\limits_{n \in \Nats} \F{F}^\lp_n$ which is \emph{not} the case for general filtrations.}
\end{definition}
\noindent
Next, we recall the notion of stopping times from probability theory.

\begin{definition}[Stopping Time]
	For a probability space~\mbox{$(\Omega,\F{F},\mathbb{P})$} with filtration $(\F{F}_n)_{n \in \Nats}$, a random variable $T \colon \Omega \to \NatsInf$ is called a \emph{stopping time}
	with respect to~$(\F{F}_n)_{n \in \Nats}$ if for every $n \in \Nats$ we have $T^{-1}(\{n\})= \{ \run \in \Omega \mid T(\run) = n \} \in \F{F}_n$.
\end{definition}

Let us reconsider the looping time $\termtime{\guard}$ and the loop filtration $(\F{F}_n^\lp)_{n \in \Nats}$.
In order to decide for a run $\run=\State_0\State_1\cdots \in \Omega^\lp$ whether its looping time is $n$, we only need to consider the states $\State_0\cdots\State_n$.
Hence, $(\termtime{\guard})^{-1}(\{n\}) \in \F{F}_n^\lp$ for any $n \in \Nats$ and thus $\termtime{\guard}$ is a \emph{stopping time} with respect to $(\F{F}_n^\lp)_{n \in \Nats}$.

Note that $\termtime{\guard}$ does \emph{not} reflect the actual runtime of $\WHILEDO{\guard}{C}$, as it does not take the runtime of the loop body $C$ into account.
Instead, $\termtime{\guard}$ only counts the number of \emph{loop iterations} of the ``outer loop'' $\WHILEDO{\guard}{C}$.
This enriches the class of probabilistic programs our technique will be able to analyze, as we will not need to require that the whole program has finite expected runtime, but only that \emph{the outer loop is expected to be executed finitely often}.

\subsection{Canonical Stochastic Process}
\label{sec:Canonical Stochastic Process}

\noindent
Now we can present our novel results on the connection of weakest preexpectations and stochastic processes.
Henceforth, let $f,I \in \E$.
Intuitively, $f$ will play the role of the \emph{postexpectation} and $I$ the role of an \emph{invariant} (i.e., $I$ is a sub-- or superinvariant).
We now present a canonical stochastic process, i.e., a sequence of random variables that captures approximating $\wp{\WHILEDO{\guard}{C}}{f}$ using the invariant $I$.%
\begin{definition}[Induced Stochastic Process]
	\label{def:induced-stochastic-process}
	The \emph{stochastic process induced by $I$}, denoted $\mathbf{X}^{f,I} = (X_n^{f,I})_{n \in \Nats}$, is given by
	\begin{align*}
		X_n^{f,I}\colon \Omega^\lp \to \PosRealsInf, \run \mapsto
		\begin{cases}
			f\left(\run\left[\termtime{\guard}(\run)\right]\right), & \textnormal{if }\termtime{\guard}(\run) \leq n \\
			I(\run[n+1]),                                           & \textnormal{otherwise}
		\end{cases}
		.
	\end{align*}
\end{definition}
\noindent
Now, in what sense does the stochastic process $\mathbf{X}^{f,I}$ capture approximating the weakest preexpectation of our loop with respect to $f$ by invariant $I$?
$X_n^{f,I}$ takes as argument a run~$\run$ of the loop and assigns to~$\run$ a value as follows: If the loop has reached a terminal state within $n$ iterations, it returns the value of the postexpectation $f$ evaluated in that terminal state.
If no such terminal state is reached within $n$ steps, it simply \emph{approximates the remainder of the run}, i.e.,
\begin{align*}
	\textcolor{gray}{\run[0] \,{}\cdots{}\, \run[n]} \, \underbracket[.5pt][2.5pt]{\run[n{+}1]\, \run[n{+}2]\, \run[n{+}3] \, {}\cdots{}}~,
\end{align*}
by returning the value of the invariant $I$ evaluated in $\run[n{+}1]$.
We see that $X_n^{f,I}$ needs at most the first $n+2$ states of a run to determine its value.
Thus, $X_n^{f,I}$ is not $\F{F}_{n}$--measurable but $\F{F}_{n+1}$--measurable, as there exist runs that agree on the first $n+1$ states but yield different images under $X_n^{f,I}$.
Hence, we shift the loop filtration $(\F{F}_n^\lp)_{n \in \Nats}$ by one.
\begin{definition}[Shifted Loop Filtration]
	\label{def:filtration}
	The filtration $(\F{G}_n^\lp)_{n \in \Nats}$ of~$\F{F}^\lp$ is defined by
	\begin{align*}
		\F{G}_n^\lp \ccoloneqq \F{F}_{n+1}^\lp \eeq \sigmagen{\{Cyl(\pi) \mid \pi \in \States^+,~|\pi| = n+2\}}~.
	\end{align*}
\end{definition}
\noindent
Note that $(\termtime{\guard})^{-1}(\{n\})\in \F{F}_n^\lp \subset \F{F}_{n+1}^\lp=\F{G}_n^\lp$, so $\termtime{\guard}$ is a stopping time w.r.t.~$(\F{G}_n^\lp)_{n \in \Nats}$ as well.%
\begin{restatable}[Adaptedness of Induced Stochastic Process]{lemma}{lemmameasurability}
	\label{lemma:measurability}
	$\mathbf{X}^{f,I}$ is adapted to $(\F{G}_n^\lp)_{n \in \Nats}$, i.e., $X_n^{f,I}$ is $\F{G}_n^\lp$--measurable.
\end{restatable}
\noindent
The loop space, the loop measure, and the induced stochastic process~$\mathbf{X}^{f,I}$ are not defined by induction on the number of steps performed in the program.
The loop space, for instance, contains \emph{all} infinite sequences of states, whether they are admissible by the loop or not.
The loop measure filters out the inadmissible runs and gives them probability 0.

Reasoning by invariants and characteristic functions, on the other hand, \emph{is} inductive.
We will thus relate iterating a characteristic function on $I$ to the stochastic process~$\mathbf{X}^{f,I}$.
For this, let $\charwp{\guard}{C}{f}$ again be the characteristic function of $\WHILEDO{\guard}{C}$ with respect to $f$, i.e.,
\begin{align*}
	\charwp{\guard}{C}{f}(X) \eeq \iverson{\neg \guard} \cdot f \pplus \iverson{\guard} \cdot \wp{C}{X}~.
\end{align*}
We now develop a first connection between the stochastic process $\mathbf{X}^{f,I}$ and $\charwp{\guard}{C}{f}$, which involves the notion of conditional expected values with respect to a $\sigma$--field, for which we provide some preliminaries here.
In general, for $M\subset \Omega$, by slight abuse of notation, the \emph{Iverson bracket} $\indicator{M}:\Omega \to \PosRealsInf$ maps $\run \in \Omega$ to $1$ if $\run \in M$ and to $0$ otherwise.
$\indicator{M}$ is $\F{F}$--measurable iff $M\in \F{F}$.
If $X$ is a random variable on $(\Omega, \F{F}, \mathbb{P})$ and $\F{G}\subset \F{F}$ is a $\sigma$--field with respect to $\Omega$, then the \emph{conditional expected value} $\expec{}{X \mid \F{G}} \colon \Omega \to \PosRealsInf$ is a $\F{G}$--measurable \emph{mapping}
such that for every $G \in \F{G}$ the equality $\expec{}{X\cdot \indicator{G}}=\expec{}{\expec{}{X \mid \F{G}}\cdot \indicator{G}}$ holds, i.e., restricted to the set $G$ the conditional expected value $\expec{}{X \mid \F{G}}$ and $X$ have the same expected value.
Hence, $\expec{}{X \mid \F{G}}$ is a random variable that is like $X$, but for elements that are indistinguishable in the subfield $\F{G}$, i.e., they either are both contained or none of them is contained in a $\F{G}$--measurable set, it ``distributes the value of $X$ equally''.%
\begin{restatable}[Relating $\boldsymbol{\mathbf{X}^{f,I}}$ and $\boldsymbol{\charwp{\guard}{C}{f}}$]{theorem}{thmcondexpec}
	\label{thm:cond_expec}
	For any $n \in \Nats$ and any $\State \in \States$, we have
	\begin{align*}
		\expec{\State}{X_{n+1}^{f,I} \Big| \F{G}_n^\lp} \eeq X_n^{f, \charwp{\guard}{C}{f}(I)}~.
	\end{align*}
\end{restatable}
\noindent
Note that both sides in \cref{thm:cond_expec} are mappings of type $\Omega^\lp \to \PosRealsInf$.
\noindent
Intuitively, \cref{thm:cond_expec} expresses the following: Consider some cylinder $Cyl(\pi)\in \F{G}_n^\lp$, i.e., $\pi=\State_0\cdots\State_{n+1}\in \States^{n+2}$ is a sequence of states of length $n+2$.
Then, $X_n^{f, \charwp{\guard}{C}{f}(I)}$ and $X_{n+1}^{f,I}$ have the same expected value under ${}^{s}\mathbb{P}$ on the cylinder set $Cyl(\pi)$ independent of the initial state $\State$ of the loop.%

Using \cref{thm:cond_expec}, one can now explain in which way iterating $\charwp{\guard}{C}{f}$ on $I$ represents an expected value, thus revealing the inductive structure inside the induced stochastic process:%
\begin{restatable}[Relating Expected Values of $\boldsymbol{\mathbf{X}^{f,I}}$ and Iterations of $\boldsymbol{\charwp{\guard}{C}{f}}$]{corollary}{corconnectionexpectation}
	\label{cor:connection_expectation}
	For any $n \in \Nats$ and any $\State \in \States$, we have
	\begin{align*}
		\expec{\State}{X_n^{f,I}} \eeq \charwpn{\guard}{C}{f}{n+1}(I)(\State)~.
	\end{align*}
\end{restatable}
\noindent
Intuitively, $\charwpn{\guard}{C}{f}{n+1}$ represents allowing for at most $n+1$ evaluations of the loop guard.
For any state $\State \in \States$, the number $\charwpn{\guard}{C}{f}{n+1}(I)(s)$ is composed of
\begin{itemize}
	\item[(a)] $f$'s average value on the final states of those runs starting in $s$ that terminate within $n + 1$ guard evaluations, and
	\item[(b)] $I$'s average value on the $(n + 2)$--nd states of those runs starting in $s$ that do \emph{not} terminate within $n + 1$ guard evaluations.
\end{itemize}
We now want to take $n$ to the limit by considering all possible numbers of iterations of the loop body.
We will see that this corresponds to evaluating the stochastic process $\mathbf{X}^{f,I}$ at the time when our loop terminates, i.e., the looping time~$\termtime{\guard}$:%
\begin{definition}[Canonical Stopped Process]
	\label{def:stopped_process}
	The mapping
	\begin{align*}
		X_{\termtime{\guard}}^{f,I}\colon \Omega^\lp \to \PosRealsInf, \run \mapsto
		\begin{cases}
			X^{f,I}_{\termtime{\guard}(\run)}(\run) = f(\run[\termtime{\guard}(\run)]), & \textnormal{if }\termtime{\guard}(\run) <\infty \\
			0,                                                                          & \textnormal{otherwise}
		\end{cases}
	\end{align*}
	is the \emph{stopped process}, corresponding to $\mathbf{X}^{f,I}$ stopped at stopping time $\termtime{\guard}$.
	As this mapping is independent of $I$, we write \mbox{$X_{\termtime{\guard}}^{f}$ instead of $X_{\termtime{\guard}}^{f,I}$}.
\end{definition}
\noindent
The stopped process now corresponds exactly to the quantity we want to reason about --- the value of $f$ evaluated in the final state after termination of our loop.
For nonterminating runs we get $0$, as there exists no state in which to evaluate $f$.

We now show that the limit of the induced stochastic process $\mathbf{X}^{f,I}$ corresponds to the stopped process $X_{\termtime{\guard}}^{f}$.
For the following lemma, note that a statement over runs $\alpha$ holds \emph{almost--surely} in the probability space $(\Omega^{\lp}, \F{F}^{\lp}, {}^{\State}\mathbb{P})$, if \mbox{$\IP{\State}{\setcomp{\run \in \Omega}{ \run \text{ satisfies } \alpha}} =1$}, i.e., the set of all elements of the sample space satisfying $\alpha$ has probability $1$.%
\begin{restatable}[Convergence of $\boldsymbol{\mathbf{X}^{f,I}}$ to $\boldsymbol{X_{\termtime{\guard}}^f}$]{lemma}{lemmaaslimit}
	\label{lemma:as_limit}
	The stochastic process $\mathbf{X}^{f,I}\cdot \indicator{(\termtime{\guard})^{-1}(\Nats)}$ converges point--wise to $X_{\termtime{\guard}}^f$, i.e., for all $\run \in \Omega^\lp$,
	\begin{align*}
		\lim\limits_{n \to \omega}~ X_n^{f,I} (\run) \cdot \indicator{(\termtime{\guard})^{-1}(\Nats)}(\run) \eeq X_{\termtime{\guard}}^f(\run)~.
	\end{align*}
	So if $\WHILEDO{\guard}{C}$ is universally almost--surely terminating, then $\mathbf{X}^{f,I}$ converges to $X_{\termtime{\guard}}^f$ almost--surely with respect to the measure $\phantom{}^{\State}\mathbb{P}$ \mbox{for any $\State \in \States$}.
\end{restatable}
\noindent
Intuitively, the factor $\indicator{(\termtime{\guard})^{-1}(\Nats)}(\run)$ selects those runs $\run$ where the looping time $\termtime{\guard}$ is finite.
If the loop is AST, then this factor can be neglected, because then $\indicator{(\termtime{\guard})^{-1}(\Nats)}$ \emph{is} the constant function $1$ for the probability measures $\phantom{}^{\State}\mathbb{P}$.
In any case, (i.e., whether the looping time is almost--surely finite or not) the expected value of the \emph{stopped process} captures precisely the weakest preexpectation of our loop with respect to the postexpectation $f$, since only the terminating runs are taken into account by $X_{\termtime{\guard}}^f$ and by $\lfp \charwp{\guard}{C}{f}$ when computing the expected value of $f$ after termination of the loop.
So from \cref{cor:connection_expectation} and \cref{lemma:as_limit}
we get our first main result:$\!\!$%
\begin{restatable}[Weakest Preexpectation is Expected Value of Stopped Process]{theorem}{thmlfpexpectation}
	\label{thm:lfp_expectation}
	\begin{align*}
		\wp{\WHILEDO{\phi}{C}}{f} \eeq \lfp \Phi_{f} \eeq \lambda \State\mydot~ \expec{\State}{X_{\termtime{\guard}}^f}.
	\end{align*}
\end{restatable}
\noindent
\cref{thm:lfp_expectation} captures our sought--after least fixed point as an expected value of a canonical stopped process.
This is what will allow us to later apply the Optional Stopping Theorem.
Moreover, it is crucial for deriving our generalization of an existing rule for lower bounds (cf.\ \cref{sec:mciver_morgan}) and the connection of upper bounds to the Lemma of Fatou (cf.\ \cref{sec:fatou}).

\subsection{Uniform Integrability}
\label{sec:Uniform Integrability}

\noindent
As we will see in \cref{sec:optional-stopping-wp}, \emph{uniform integrability}
of a certain stochastic process is the central aspect of the Optional Stopping Theorem (\cref{thm:optional_stopping}).
In probability theory, uniform integrability means that taking the expected value and taking the limit of a stochastic process commutes.%
\begin{definition}[\protect{Uniform Integrability of Stochastic Processes, \textnormal{\cite[Lemma~7.10.(3)]{grimmett2001probability}}}]
	\label{def:uniform_integrability_stochastic_processes}
	Let $\mathbf{X}=(X_n)_{n \in \Nats}$ be a stochastic process on a probability space $(\Omega,\F{F},\mathbb{P})$ with almost--surely existing limit $\lim_{n \to \omega} X_n$.
	The process $\mathbf{X}$ is \emph{uniformly integrable} if
	\[
		\expec{}{\lim\limits_{n \to \omega} X_n} \eeq \lim\limits_{n \to \omega} \expec{}{X_n}~.
	\]
\end{definition}
\begin{counterexample}[\textnormal{\protect{\cite[Sect.~7.10]{grimmett2001probability}}}]
	Consider the stochastic process $\mathbf{X}=(X_n)_{n \in \Nats}$ on a probability space $(\Omega,\F{F},\mathbb{P})$ with $\IP{}{Y_n=n}=\tfrac{1}{n}=1-\IP{}{Y_n=0}$.
	Then $\expec{}{X_n}=1$.
	Moreover, $\mathbf{X}$ converges almost surely to $Y\equiv 0$, i.e., the constant function $0$.
	So, $\expec{}{Y}=0$.
	But
	\[
		\lim\limits_{n \to \omega} \expec{}{X_n}=\lim\limits_{n \to \omega} 1 = 1 \neq 0 = \expec{}{0},
	\]
	so $\mathbf{X}$ is not u.i.
\end{counterexample}
\noindent
Note that our notion of uniform integrability of expectations from \cref{def:uniform_integrability_invariants} coincides with uniform integrability of the corresponding induced stochastic process.%
\begin{restatable}[Uniform Integrability of Expectations and Stochastic Processes]{corollary}{corouniformintegrability}
	\label{coro:uniform_integrability}
	Let the loop $\WHILEDO{\guard}{C}$ be AST.\footnote{It suffices that $\IP{\State}{\termtime{\phi}<\infty} = 1$ for any $\State$.
		But this is equivalent to AST as we required the body of the loop to be AST.}
	Then $I$ is uniformly integrable for $f$ (in the sense of \textnormal{\cref{def:uniform_integrability_invariants}}) iff the induced stochastic process $\mathbf{X}^{f,I}$ is uniformly integrable (in the sense of \textnormal{\cref{def:uniform_integrability_stochastic_processes}}), i.e.,
	\[
		\lim_{n \to \omega} \charwpn{\guard}{C}{f}n(I) \eeq \lfp \charwp{\guard}{C}{f}
		\qqiff \forall \State \in \States\colon \quad \expec{\State}{\lim\limits_{n \to \omega} X_n^{f,I}} \eeq \lim\limits_{n \to \omega} \expec{\State}{X_n^{f,I}}~.
	\]
\end{restatable}
\noindent
\cref{coro:uniform_integrability} justifies the naming in \cref{def:uniform_integrability_invariants}: an expectation $I$ is uniformly integrable for $f$ iff its induced process $\mathbf{X}^{f,I}$ is uniformly integrable.
So, we can now extend the diagram from \cref{sec:series_perspective_on_loops} as follows:

\begin{center}
	\begin{adjustbox}{max width=\linewidth}
		\begin{tikzcd}[column sep = huge, row sep=large]
			\mathbf{X}^{f,I}~\textnormal{u.i.} \quad \arrow[d,Leftrightarrow, "\textnormal{~~\cref{coro:uniform_integrability}}"] \arrow[r, Leftrightarrow, "\textnormal{{\cref{lemma:as_limit,def:uniform_integrability_stochastic_processes}}}" shift={(0,0.2)}]
			&\quad\expec{\bullet}{X_n^{f,I}} \xrightarrow{~n \to \omega~} \expec{\bullet}{X_{\termtime{\guard}}^f}
			\arrow[d,Leftrightarrow, "\textnormal{
					\begin{tabular}{l}
						\cref{cor:connection_expectation} \\[-1ex]
						and \cref{thm:lfp_expectation}
					\end{tabular}
				}"] \\
			I~\textnormal{u.i.\ for}~f \quad \arrow[r, Leftrightarrow, "\textnormal{\raisebox{1ex}{\cref{def:uniform_integrability_invariants}}}"] \arrow[d, Leftrightarrow, "\textnormal{
				\begin{tabular}{l}
					\cref{thm:uniform_integrability_least_fixed_point} \\[-1ex]
					and \cref{def:uniform_integrability_invariants}
				\end{tabular}
			}" shift={(0,-0.1)}]
			&\quad\charwpn{\guard}{C}{f}n(I) \xrightarrow{~n \to \omega~} \lfp \charwp{\guard}{C}{f}
			\\
			I \preceq \charwp{\guard}{C}{f}(I) \Rightarrow I \preceq \lfp \charwp{\guard}{C}{f}
		\end{tikzcd}
	\end{adjustbox}
\end{center}
Uniform integrability is very hard to verify in general, both in the realm of stochastic processes as well as in the realm of expectation transformers. Thus, one usually tries to find \emph{sufficient} criteria for uniform integrability that are easier to verify. The very idea of the Optional Stopping Theorem is to provide such sufficient criteria for uniform integrability which then allow deriving a lower bound as we will discuss in the next section. \renewcommand{\charwp}[3]{\charfun{\wpsymbol}{#1}{#2}{#3}}
\renewcommand{\charwpn}[4]{\charfunn{\wpsymbol}{#1}{#2}{#3}{#4}}

\renewcommand{\charwp}[3]{\Phi_{#3}}
\renewcommand{\charwpn}[4]{\Phi_{#3}^{#4}}

\section{The Optional Stopping Theorem of Weakest Preexpectations}
\label{sec:optional-stopping-wp}
\noindent
In this section, we develop an inductive proof rule for lower bounds on preexpectations by using the results of \cref{sec:expectations_processes} and the Optional Stopping Theorem (\cref{thm:optional_stopping}).
The proofs of our results in this section can be found in \techrep{\cref{app:proofs_optional_stopping}}\cameraready{\cite[App.\ D]{arxiv}}.
Recall that we have fixed a loop $\WHILEDO{\guard}{C}$, a finite postexpectation $f$, a corresponding characteristic function $\charwp{\guard}{C}{f}$, and another finite expectation $I$ which plays the role of an invariant.

We first introduce the Optional Stopping Theorem from probability theory.
It builds upon the concept of submartingales.
A \emph{submartingale} is a stochastic process that induces a monotonically increasing sequence of its expected values.
\begin{definition}[Submartingale]
	\label{defQQQsubmartingale_filtration}
	Let $(X_n)_{n\in \Nats}$ be a stochastic process on a probability space $(\Omega,\F{F},\mathbb{P})$ adapted to a filtration $(\F{F}_n)_{n\in \Nats}$ of $\F{F}$, i.e., a sequence of random variables $X_n\colon\Omega \to \PosRealsInf$ such that $X_n$ is $\F{F}_n$--measurable.
	Then $(X_n)_{n \in \Nats}$ is called a \emph{submartingale} with respect to $(\F{F}_n)_{n \in \Nats}$ if
	\begin{align*}
		 & \mathrm{(a)} \qquad \expec{}{X_n}< \infty \text{ for all } n\in \Nats, \text{ and} \\
		 & \mathrm{(b)} \qquad \expec{}{X_{n+1}\mid \F{F}_n} \geq X_n.
	\end{align*}
\end{definition}
\noindent
It turns out that submartingales are closely related to subinvariants.
In fact, $I$ being a \emph{subinvariant} (plus some side conditions) gives us that the stochastic process induced by $I$ is a \emph{submartingale}.
\begin{restatable}[Subinvariant Induces Submartingale]{lemma}{lemsubmartingale}
	\label{lemma:submartingale}
	Let $I$ be a \emph{subinvariant}, i.e., $I \preceq \charwp{\guard}{C}{f}(I)$, such that $\charwpn{\guard}{C}{f}{n}(I) \pprec \infty$ for every $n \in \Nats$, that is, $\charwpn{\guard}{C}{f}{n}(I)$ only takes finite values.
	Then the induced stochastic process $\mathbf{X}^{f,I}$ is a submartingale with respect to $(\F{G}_n^\lp)_{n\in \Nats}$.
\end{restatable}
\noindent
Given a submartingale $(X_n)_{n \in \Nats}$ and a stopping time $T$, the goal of the Optional Stopping Theorem is to prove a lower bound for the expected value of $X_n$ at the stopping time $T$.
To this end, we define a stochastic process $(X_{n \wedge T})_{n \in \Nats}$ where for any $\run \in \Omega$, $X_{n \wedge T}(\run) = X_n(\run)$ if $n$ is smaller than the stopping time $T(\run)$ and otherwise, $X_{n \wedge T}(\run) = X_{T(\run)}(\run)$.
Hence, $\expec{}{\lim_{n \to \omega}X_{n\wedge T}}$ is the expected value of $X_n$ at the stopping time $T$.
The Optional Stopping Theorem asserts that the first component $X_0$ of the stochastic process $(X_n)_{n \in \Nats}$ is a lower bound for $\expec{}{\lim_{n \to \omega}X_{n\wedge T}}$ provided that $(X_{n \wedge T})_{n \in \Nats}$ is uniformly integrable.
Moreover, the Optional Stopping Theorem provides a collection of criteria that are sufficient for uniform integrability of $(X_{n \wedge T})_{n \in \Nats}$.
\begin{theorem}[Optional Stopping Theorem \textnormal{\protect{\cite[Theorems~12.3.(1), 12.4.(11), 12.5.(1), 12.5.(2), 12.5.(9)]{grimmett2001probability}}}]
	\label{thm:optional_stopping}
	Let $(X_n)_{n\in \Nats}$ be a submartingale and $T$ be a stopping time on a probability space $(\Omega,\F{F},\mathbb{P})$ with respect to a filtration $(\F{F}_n)_{n \in \Nats}$.
	Then 
	$\mathbf{X}_{\wedge T} = (X_{n \wedge T})_{n \in \Nats}$ defined by
	\[
		X_{n \wedge T}: \Omega \to \PosRealsInf, \run \mapsto X_{\min(n,T(\run))}(\run),
	\]
	is also a submartingale w.r.t.~$(\F{F}_n)_{n \in \Nats}$.
	If $\mathbf{X}_{\wedge T}$ converges almost--surely and is uniformly integrable,
	\[
		\expec{}{X_0}=\expec{}{X_{0 \wedge T}}\leq \lim\limits_{n \to \omega}\expec{}{X_{n\wedge T}} = \expec{}{\lim\limits_{n \to \omega}X_{n\wedge T}}.
	\]
	If \emph{one} of the following conditions holds, then $\mathbf{X}_{\wedge T}$ converges almost--surely and is uniformly integrable:
	\begin{enumerate}
		\item[(a)] $T$ is almost--surely bounded, i.e., there is a constant $N\in \Nats$ such that~$\IP{}{T \leq N} = 1$.
		\item[(b)] $\expec{}{T}<\infty$ and there is a constant $c\in \Reals_{\geq 0}$, such that for each $n\in \Nats$
			\[
				\expec{}{\abs{X_{n+1}-X_n}\mid \F{F}_n}\leq c \qquad \text{holds almost--surely~.}
			\]
		\item[(c)] There exists a constant $c\in \Reals_{\geq 0}$ such that $X_{n \wedge T}\leq c$ holds almost--surely for every $n \in \Nats$.
	\end{enumerate}
\end{theorem}
\noindent
Our goal now is to transfer the Optional Stopping Theorem from probability theory to the realm of weakest preexpectations in order to obtain inductive proof rules for lower bounds on weakest preexpectations.
So far, we have introduced the looping time $\termtime{\guard}$ (which is a stopping time w.r.t.\ $(\F{F}_n^\lp)_{n \in \Nats}$), presented the connection of subinvariants and submartingales, and defined the concept of uniform integrability also for expectations.
Hence, the only missing ingredient is a proper connection of expectations to the condition ``$\expec{}{\abs{X_{n+1}-X_n}\mid \F{F}_n}\leq c$'' in \cref{thm:optional_stopping}
(b).
To translate this concept to expectations, we require that the expectation $I$ has a certain shape depending on the postexpectation $f$.
\begin{definition}[Harmonization]
	\label{def:harmony}
	An expectation $I$ \emph{harmonizes with $f \in \E$} if it is of the form
	\begin{align*}
		I \eeq \iverson{\neg \guard} \cdot f \pplus \iverson{\guard} \cdot \primed{I}~,
	\end{align*}
	for some expectation $\primed{I} \in \E$.
\end{definition}
\noindent
\cref{def:harmony} reflects that in terminal states $t$ of the loop the invariant $I$ evaluates to $f(t)$.
For an invariant $I$ to harmonize with postexpectation $f$ is a minor restriction on the shape of~$I$.
It is usually easy to choose an $I$ that takes the value of $f$ for states in which the loop is not executed at all.
Moreover, performing one iteration of $\charwp{\guard}{C}{f}$ obviously brings \emph{any} expectation~\mbox{``into shape''}:%
\begin{corollary}[Harmonizing Expectations]
	For any $f, J \in \E$, $\charwp{\guard}{C}{f}(J)$ harmonizes with $f$.
\end{corollary}
\noindent
The actual criterion that connects ``$\expec{}{\abs{X_{n+1}-X_n}\mid \F{F}_n}\leq c$'' with the invariant $I$ is called \emph{conditional difference boundedness} (see also \cite{DBLP:conf/vmcai/FuC19,DBLP:conf/popl/FioritiH15}):%
\begin{definition}[Conditional Difference Boundedness]
	\label{conditionally difference boundedness}
	Let $I \in \E$.
	We define the function $H\colon \E \to \E$ and the expectation $\Diff{I} \in \E$ as\footnote{Recall that we have fixed a loop $\WHILEDO{\guard}{C}$.}
	\[
		H(X) \eeq \iverson{\guard} \cdot \wp{C}{X} \qqand \Diff{I} \eeq \lambda s \mydot \left(H\bigl(\abs{I-I(s)}\bigr)(s)\right)~.
	\]
	The expectation $I$ is called \emph{conditionally difference bounded} (c.d.b.) if for some constant $c \in \Reals_{\geq 0}$
	\begin{align*}
		\Diff{I}(\State) \lleq c \text{ holds for all } \State \in \States.
	\end{align*}
\end{definition}
\noindent
The expectation $\Diff{I}$ expresses the \emph{expected change of $I$} within one loop iteration.
So, if $I$ is c.d.b.~ it is expected to change at most by a constant in one loop iteration.%
\begin{example}
	\label{ex:running_example_cdb}
	Reconsider the program $C_{\mathit{cex}}$ 
	from \textnormal{\cref{ex:simple-induction}} and expectation $I = b + \iverson{a \neq 0}$.
	We will check \emph{conditional difference boundedness} of $I$, using the function $H$ given by
	\begin{align*}
		H(X) \eeq \iverson{a \neq 0} \cdot \tfrac{1}{2} \cdot \bigl(X\subst{a}{0} + X\subst{b}{b + 1}\bigr)\subst{k}{k+1}~.
	\end{align*}
	We then check the following:
	\begin{align*}
		\Diff{I}
		\eeq & \bigl(\lambda s\mydot \iverson{a \neq 0} \cdot \tfrac{1}{2} \cdot \bigl(\abs{I -I(s)}\subst{a}{0} \pplus \abs{I -I(s)}\subst{b}{b + 1}\bigr)\subst{k}{k+1}\bigr) (\State)                                                      \\
		\eeq & \lambda s\mydot \bigl(\iverson{a \neq 0} \cdot \tfrac{1}{2} \cdot \bigl(\abs{b + \iverson{0 \neq 0} - s(b) - \iverson{a \neq 0}(s)} \pplus \abs{b + 1 + \iverson{a \neq 0} - s(b) - \iverson{a \neq 0}(s)}\bigr)\bigr)(\State) \\
		\eeq & \iverson{a \neq 0}(\State) \cdot \tfrac{1}{2} \cdot \bigl(\abs{s(b) + \iverson{0 \neq 0} - s(b) - \iverson{a \neq 0}(s)} \pplus \abs{s(b) + 1 + \iverson{a \neq 0}(s) - s(b) - \iverson{a \neq 0}(s)}\bigr)                    \\
		\eeq & \iverson{a \neq 0}(s) \cdot \tfrac{1}{2} \cdot \bigl(\abs{{-} 1} + \abs{1}\bigr) \lleq 1~.
	\end{align*}
	Thus, $I$ is c.d.b.~by the constant $1$.
	In contrast, the subinvariant $\primed{I} = b + \iverson{a \neq 0} \cdot (1 + 2^k)$ from \textnormal{\cref{ex:simple-induction}} is not conditionally difference bounded.
	Indeed, we would get (cf.\ \textnormal{\techrep{\cref{app:proofs_optional_stopping}}\cameraready{\cite[App.\ D]{arxiv}}} for details)
	\begin{align*}
		 & \Diff{I'} \eeq \bigl(\iverson{a \neq 0} \cdot (1 + 2^k)\bigr) (\State)~,
	\end{align*}
	which cannot be bounded by a constant.
\end{example}
\noindent
Finally, we can connect the expected change of $I$ to a property of the stochastic process $\mathbf{X}^{f,I}$.
This is our second major result.%
\begin{restatable}[Expected Change of $\boldsymbol{I}$]{theorem}{thmconditionaldifferenceboundedness}
	\label{thm:conditional_difference_boundedness}
	Let $I \pprec \infty$ harmonize with $f$.
	Then
	\begin{align*}
		\expec{\State}{\big|X_{n+1}^{f,I}-X_n^{f,I}\big| ~\Big|~ \F{G}_n^\lp} \eeq X_n^{0,\Diff{I}}~.
	\end{align*}
\end{restatable}
\noindent
The stochastic process $\mathbf{X}^{f,I}$ induced by $I$ exhibits \mbox{an interesting correspondence}:
If $\Diff{I}$ is bounded by a constant $c$ (i.e., if $I$ is c.d.b.), then so is $X_n^{0,\Diff{I}}$ and thus \cref{thm:conditional_difference_boundedness} ensures that precondition (b) of the Optional Stopping Theorem (\cref{thm:optional_stopping}) is fulfilled.
Note that \cref{thm:conditional_difference_boundedness} depends crucially on the fact that $I \pprec \infty$ as otherwise the well--definedness of the expectation $\Diff{I}$ cannot be ensured.

Now \cref{lemma:submartingale} allows us to use the Optional Stopping Theorem from probability theory (\cref{thm:optional_stopping}) to prove a \emph{novel Optional Stopping Theorem for weakest preexpectations}, which collects sufficient conditions for uniform integrability.
In particular, due to \cref{thm:conditional_difference_boundedness}, \emph{our} Optional Stopping Theorem shows that our notion of conditional difference boundedness is an (easy--to--check) sufficient criterion for uniform integrability and hence, for ensuring that a subinvariant is indeed a lower bound for the weakest preexpectation under consideration.
After stating the theorem, we will discuss the intuition of its parts in more detail.%
\begin{restatable}[Optional Stopping Theorem for Weakest Preexpectation Reasoning]{theorem}{thmoptionalstoppingprobabilisticprograms}
	\label{thm:optional_stopping_probabilistic_programs}
	Consider a loop $\WHILEDO{\guard}{C}$ where $C$ is AST.
	Let $I \pprec \infty$ be a subinvariant w.r.t.~the postexpectation $f \pprec \infty$ (i.e., $I \preceq \charwp{\guard}{C}{f}(I)$).
	$I$ is uniformly integrable~for $f$ iff $I$ is a \emph{lower bound}, i.e.,
	\[
		I \ppreceq \lfp \charwp{\guard}{C}{f} \eeq \wp{\WHILEDO{\guard}{C}}{f}.
	\]
	$I$ is \emph{uniformly integrable} for $f$ if one of the following three conditions holds:

	\begin{enumerate}
		\item[(a)] The looping time $\termtime{\guard}$ of $\WHILEDO{\guard}{C}$ is almost--surely bounded, i.e., for every state $\State\in \States$ there exists a constant $N(\State) \in \Nats$ with $\IP{\State}{\termtime{\guard} \leq N(\State)}=1$ and $\charwpn{\guard}{C}{f}{n}(I) \pprec \infty$ for every $n \in \Nats$.
		\item[(b)] The expected looping time of $\WHILEDO{\guard}{C}$ is finite for every initial state $\State \in \States$, $I$ harmonizes with $f$, $\charwp{\guard}{C}{f}(I) \pprec \infty$, and $I$ is conditionally difference bounded.
		\item[(c)] Both $f$ and $I$ are bounded and $\WHILEDO{\guard}{C}$ is AST.
	\end{enumerate}
\end{restatable}
\noindent
We can now extend the diagram from \cref{sec:expectations_processes} connecting the realm of stochastic processes (on the right) and the realm of expectation transformers (on the left) for a universally almost--surely terminating program.
The respective Optional Stopping Theorems provide the sufficient criteria for uniform integrability, which is marked by the dashed implications.%

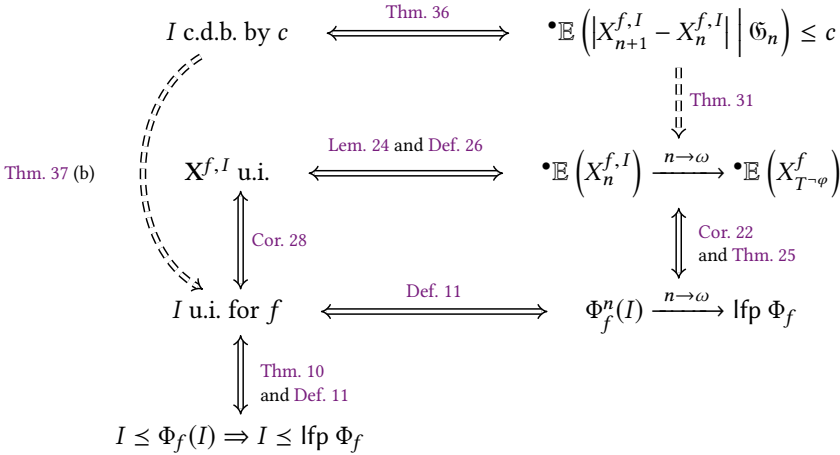
\begin{center}
	\begin{adjustbox}{max width=\linewidth}
		\begin{tikzcd}[column sep = huge, row sep=large]
			I~\textnormal{c.d.b.~by}~c \quad \arrow[r, Leftrightarrow, "\textnormal{\raisebox{1ex}{\cref{thm:conditional_difference_boundedness}}}" shift={(0,0)}] \arrow[dd, bend right = 60, Rightarrow, dashed, "\textnormal{~~\cref{thm:optional_stopping_probabilistic_programs}~(b)}" shift={(-2,0)}]
			&\quad\expec{\bullet}{\big|X_{n+1}^{f,I}-X_n^{f,I} \big| ~\Big|~ \F{G}_n}\leq c\arrow[d, Rightarrow, dashed, "\textnormal{\raisebox{1ex}{~~\cref{thm:optional_stopping}}}"] \\
			\mathbf{X}^{f,I}~\textnormal{u.i.} \quad \arrow[d,Leftrightarrow, "\textnormal{~~\cref{coro:uniform_integrability}}"] \arrow[r, Leftrightarrow, "\textnormal{{\cref{lemma:as_limit,def:uniform_integrability_stochastic_processes}}}" shift={(0,0.2)}]
			&\quad \expec{\bullet}{X_n^{f,I}} \xrightarrow{~n \to \omega~} \expec{\bullet}{X_{\termtime{\guard}}^f}
			\arrow[d,Leftrightarrow, "\textnormal{
					\begin{tabular}{l}
						\cref{cor:connection_expectation} \\[-1ex]
						and \cref{thm:lfp_expectation}
					\end{tabular}
				}"] \\
			I~\textnormal{u.i.\ for}~f \quad \arrow[r, Leftrightarrow, "\textnormal{\raisebox{1ex}{\cref{def:uniform_integrability_invariants}}}"] \arrow[d, Leftrightarrow, "\textnormal{
				\begin{tabular}{l}
					\cref{thm:uniform_integrability_least_fixed_point} \\[-1ex]
					and \cref{def:uniform_integrability_invariants}
				\end{tabular}
			}" shift={(0,-0.1)}]
			&\quad \charwpn{\guard}{C}{f}n(I) \xrightarrow{~n \to \omega~} \lfp \charwp{\guard}{C}{f}
			\\
			I \preceq \charwp{\guard}{C}{f}(I) \Rightarrow I \preceq \lfp \charwp{\guard}{C}{f}
		\end{tikzcd}
	\end{adjustbox}
\end{center}
Let us elaborate on the different cases of \emph{our} Optional Stopping Theorem (\cref{thm:optional_stopping_probabilistic_programs}): Case (a) yields an alternative proof for the technique of so--called metering functions by~\cite{DBLP:conf/cade/FrohnNHBG16} for \emph{deterministic terminating} loops. As for the severity of the finiteness condition ``$\charwpn{\guard}{C}{f}{n}(I) \pprec \infty$ for every $n \in \Nats$'', note that if the body $C$ is loop--free, this condition is vacuously satisfied as $I$ itself is finite and cannot become infinite by finitely iterations of $\charwp{\guard}{C}{f}$. If $C$ contains loops, then we can establish the finiteness condition by finding a finite superinvariant $U$ with \mbox{$I \preceq U \pprec \infty$}. \mbox{In this case, we can also guarantee $\charwpn{\guard}{C}{f}{n}(I) \pprec \infty$}.\footnote{The reason is that we have $U \succeq \charwpnnoindex{\cond}{\PP}{f}{n}(U)$ for all $n \in \Nats$ by monotonicity of $\charwpnoindex{\cond}{\PP}{f}$ (\cref{thm:healthiness}).
	Then, $U \succeq I$ implies $\charwpnnoindex{\cond}{\PP}{f}{n}(U) \succeq \charwpnnoindex{\cond}{\PP}{f}{n}(I)$ also by monotonicity of $\charwpnoindex{\cond}{\PP}{f}$, which gives us $\infty \mathrel{{\succ}\hspace{-.5ex}{\succ}} U \succeq \charwpnnoindex{\cond}{\PP}{f}{n}(U) \succeq \charwpnnoindex{\cond}{\PP}{f}{n}(I)$.}

Case (b) applies whenever the outer loop is expected to be executed finitely often. In particular, this holds if the entire loop terminates positively almost--surely (i.e., within finite expected \emph{runtime}).

To the best of our knowledge, Cases (a) and (b) are the first sufficiently simple induction rules for lower bounds that do not require restricting to \mbox{bounded postexpectations $f$}. While the requirements on the loop's termination behavior gradually weaken along $\textnormal{(a)} \rightarrow \textnormal{(b)} \rightarrow \textnormal{(c)}$, the requirements on the subinvariant $I$ become stricter.

Finally, Case (c) yields an alternative proof of the result of \cite{DBLP:series/mcs/McIverM05} on inductive lower bounds for bounded expectations in case of AST, which we will generalize in \cref{sec:mciver_morgan}.

When comparing the cases (c) of \cref{thm:optional_stopping} and \cref{thm:optional_stopping_probabilistic_programs}, we notice that \cref{thm:optional_stopping} (c) has no restrictions on the stopping time, whereas \cref{thm:optional_stopping_probabilistic_programs}~(c) requires almost--sure termination. This might spark some hope that AST is not needed in \cref{thm:optional_stopping_probabilistic_programs} (c), but the following counterexample shows that this is not the case:%
\begin{counterexample}
	\label{non-termination}
	Consider the program
	\begin{align*}
		 & \WHILEDO{\TRUE}{\SKIP}~,
	\end{align*}
	together with the bounded postexpectation $f = 1$, i.e., we are interested in the termination probability which is obviously~$0$. The corresponding characteristic function is given by
	\begin{align*}
		\Phi_1(X) \eeq \iverson{\neg \TRUE}\cdot 1 + \iverson{\TRUE}\cdot\wp{\SKIP}{X} \eeq X~,
	\end{align*}
	i.e., $\Phi_1$ is the identity map. Trivially, the bounded expectation $I = 1$ is a fixed point of $\Phi_1$, thus in particular $I$ is a subinvariant. Clearly, $I$ is not a lower bound on the actual termination probability, i.e., on $\lfp \Phi_1$. If the condition of almost--sure termination in \textnormal{\cref{thm:optional_stopping_probabilistic_programs} (c)}
	could be weakened, it has to be ensured that for any program $\WHILEDO{\guard}{C}$ with universally almost--surely terminating body $C$\footnote{Note that in this case $1$ is always a subinvariant.} and postexpectation $f = 1$, $1$ is a lower bound only if the program terminates universally almost--surely. But this means that this property has to be at least as strong as almost--sure termination.
\end{counterexample}
\noindent
We reconsider \cref{ex:simple-induction} illustrating unsoundness of simple lower induction and do \emph{sound} lower induction instead.%
\begin{example}
	\label{ex:running_example_our_rule}
	Let us continue \textnormal{\cref{ex:running_example_cdb}}, where we have checked that for the program $C_{\mathit{cex}}$ the expectation $I = b + \iverson{a \neq 0}$ is conditionally difference bounded by $1$. It is easy to check that $I$ is a fixed point of the characteristic function $\Phi_b$ with respect to the postexpectation $b$, which by Park induction gives us a finite upper bound on the least fixed point of $\Phi_b$. But up to now we could not prove that $I$ is indeed equal to the least fixed point. Using \textnormal{\cref{thm:optional_stopping_probabilistic_programs}}, we can now do this.

	First of all, we already have $\charwp{\guard}{C}{b}(I)=I \pprec \infty$ and since $I$ is a fixed point, it is also a subinvariant. Secondly, the loop is expected to be executed twice.\footnote{Positive almost--sure termination itself can also be verified by Park induction, see~\cite{DBLP:conf/esop/KaminskiKMO16,DBLP:journals/jacm/KaminskiKMO18}.}
	Finally, $I = b + \iverson{a \neq 0} = \iverson{\neg (a\neq 0)} \cdot b + \iverson{a \neq 0}\cdot (b+1)$ harmonizes with $b$ and is conditionally difference bounded. Hence, the preconditions of \textnormal{\cref{thm:optional_stopping_probabilistic_programs}}~(b) are satisfied and $I$ is indeed a lower bound on $\lfp \Phi_b$. Since $I$ is a fixed point, \emph{it is the least fixed point}, i.e., we have proved $\wp{C_\mathit{cex}}{b} = I$.
\end{example}
\noindent
Further case studies demonstrating the effectiveness of our proof rule, as well as an example that cannot be treated by \mbox{\cref{thm:optional_stopping_probabilistic_programs}, are provided in \techrep{\cref{app:examples}}\cameraready{\cite[App.\ A]{arxiv}}}.

\renewcommand{\charwp}[3]{\charfun{\wpsymbol}{#1}{#2}{#3}}
\renewcommand{\charwpn}[4]{\charfunn{\wpsymbol}{#1}{#2}{#3}{#4}}

\renewcommand{\charwp}[3]{\Phi_{#3}}
\renewcommand{\charwpn}[4]{\Phi_{#3}^{#4}}
\section{Lower Bound Rules by McIver and Morgan}
\label{sec:mciver_morgan}
In \cref{sec:optional-stopping-wp}, we briefly mentioned the rules for lower bounds for \emph{bounded} expectations by \techrep{\linebreak}\cite{DBLP:series/mcs/McIverM05} which are restated in \cref{thm:lower_bounds_mcivermorgan} below.
To the best of our knowledge, before our new \cref{thm:optional_stopping_probabilistic_programs} these were the only existing \emph{inductive} proof rules for weakest preexpectations.

\begin{restatable}[\textnormal{\cite{DBLP:series/mcs/McIverM05}}]{theorem}{thmlowerboundsmcivermorgan}
	\label{thm:lower_bounds_mcivermorgan}
	Let $f \in \E$ be a \textbf{bounded} postexpectation.
	Furthermore, let $I' \in \E$ be a \textbf{bounded} expectation such that the harmonized expectation $I \in \E$ given by $ I \eeq \iverson{\neg \guard} \cdot f + \iverson{\guard} \cdot I' $ is a sub{\-}in{\-}va{\-}ri{\-}ant of $\WHILEDO{\guard}{C}$ with respect to $f$.
	Finally, let $ T = \wp{\WHILEDO{\guard}{C}}{1}
	$ be the termination probability of $\WHILEDO{\guard}{C}$. Then:
	\begin{enumerate}
		\item If $I = \iverson{G}$ for some predicate $G$, then $ T \cdot I \ppreceq \wp{\WHILEDO{\guard}{C}}{f}~. $
		\item If $\iverson{G} \preceq T$ for some predicate $G$, then $ \iverson{G} \cdot I \ppreceq \wp{\WHILEDO{\guard}{C}}{f}~. $
		\item If $\epsilon \cdot I \preceq T$ for some $\epsilon > 0$, then $ I \ppreceq \wp{\WHILEDO{\guard}{C}}{f}~. $
	\end{enumerate}
\end{restatable}
\noindent
\cref{thm:lower_bounds_mcivermorgan} does not make any assumptions on the termination behavior of the loop, so, it is also possible to analyze programs with termination probability $<1$.
It turns out that \cref{thm:lower_bounds_mcivermorgan} (1) -- (3) can be proved easily from our results from \cref{sec:expectations_processes} in the case where $C$ is AST where we do \emph{not} need the restriction that $I$ harmonizes with $f$.
In particular, we can show that in \cref{thm:lower_bounds_mcivermorgan} (3) the fact that $T$ is the probability of termination is \emph{insignificant} (see \techrep{\cref{app:mciver_morgan}}\cameraready{\cite[App.\ E]{arxiv}}).
In fact, it suffices if $T$ is the weakest preexpectation for some arbitrary bounded postexpectation, i.e., a \emph{least fixed point} (see \techrep{\cref{app:mciver_morgan}}\cameraready{\cite[App.\ E]{arxiv}}
for details and proofs).
So, we obtain the following generalized version of \cref{thm:lower_bounds_mcivermorgan} (3) in the case where $C$ is AST which is substantially more powerful: it states a sufficient condition for a subinvariant to be a lower bound but also a \emph{necessary} condition.
This is the main new contribution of this section.

\begin{restatable}[Generalization of \cref{thm:lower_bounds_mcivermorgan} (3)]{theorem}{thmgeneralizationmcivermorgan}
	\label{thm:generalization_mciver_morgan}
	Let $f \in \E$ be a \textbf{bounded} postexpectation.
	Furthermore, let $I \in \E$ be a \textbf{bounded} expectation such that $I$ is a sub{\-}in{\-}va{\-}ri{\-}ant of $\WHILEDO{\guard}{C}$ with respect to $f$ where $C$ is AST.
	There exist $\epsilon > 0$ and $g \in \E$ \textbf{bounded} s.t.~
	\[
		\epsilon \cdot I \preceq \wp{\WHILEDO{\guard}{C}}{g} \qquad \text{if and only if} \qquad I \ppreceq \wp{\WHILEDO{\guard}{C}}{f}~.
	\]
\end{restatable}

\begin{example}
	\label{ex:mciver_morgan}
	Let us consider the program $C_{\mathit{rdw}}$ for an asymmetric random walk
	\begin{align*}
		 & \WHILE{x > 0}                                                                 \\
		 & \qquad \COMPOSE{\PCHOICE{\ASSIGN{x}{x-1}}{\sfrac{1}{3}}{\ASSIGN{x}{x + 1}}}{} \\
		 & \qquad \ASSIGN{y}{\max(y-1,0)}                                                \\
		 & \}~,
	\end{align*}
	with $x,y\in \Nats$ and $y \leq 100$.
	This program is \emph{not} AST but the body of the loop is indeed AST.
	Furthermore, the postexpectation $y$ is bounded.
	If $y \leq x$ initially then $y$ is $0$ after termination of the program.
	So, $\wp{C_{\mathit{rdw}}}{y}\geq \iverson{y > x}\cdot \left(\tfrac{1}{3}\right)^x\cdot (y-x)\coloneqq I$.

	Now consider $f=\iverson{y \text{ even}}\cdot 200\cdot y^2 + \iverson{y \text{
				odd}}\cdot (y+5)^4$. We have $I' \preceq \Phi_f(I')$, where $I'=400\cdot I$ (see \textnormal{\techrep{\cref{app:mciver_morgan}}\cameraready{\cite[App.\ E]{arxiv}}}). As we have $\tfrac{1}{400}\cdot I' \preceq \wp{C_{\mathit{rdw}}}{y}$ we can conclude from \textnormal{\cref{thm:generalization_mciver_morgan}} that $I' \preceq \wp{C_{\mathit{rdw}}}{f}$. Note that this is easier than relating $I'$ and the termination probability as required by \textnormal{\cref{thm:lower_bounds_mcivermorgan}} since the probability of termination of the loop is independent of $y$.
\end{example}
Of course, \cref{ex:mciver_morgan} is an artificial example.
Nevertheless, it shows a strength of our generalization: it makes it easier to reason about bounded expectations which are independent of the probability of termination.
However, a drawback of \cref{thm:lower_bounds_mcivermorgan} remains: one already needs a lower bound, i.e., one has to be able to read off a lower bound directly from the program.
\renewcommand{\charwp}[3]{\charfun{\wpsymbol}{#1}{#2}{#3}}
\renewcommand{\charwpn}[4]{\charfunn{\wpsymbol}{#1}{#2}{#3}{#4}}

\renewcommand{\charwp}[3]{\Phi_{#3}}
\renewcommand{\charwpn}[4]{\Phi_{#3}^{#4}}
\section{Upper Bounds and Fatou's Lemma}
\label{sec:fatou}

\noindent
We saw that Park induction for proving upper bounds does not require additional conditions such as conditional difference boundedness or even bound{\-}ed{\-}ness of $f$ or $I$, respectively.
The question arises whether this fact is also explainable using our canonical stochastic process.
Indeed, the well--known \emph{Lemma of Fatou} provides such an explanation.
We will present a specialized variant of it which is sufficient for our purpose.
\begin{lemma}[Fatou's Lemma (cf.\ \protect{\cite[Lemma~2.7.1]{bauer71measure}})]
	\label{fatou}
	Let $(X_n)_{n \in \Nats}$ be a stochastic process on a probability space $(\Omega, \F{F}, \mathbb{P})$.
	Then
	\begin{align*}
		\expec{}{\lim_{n \to \omega} X_n} \lleq \lim_{n \to \omega} \expec{}{X_n}~,
	\end{align*}
	where the $\lim$ on the left--hand--side is point--wise.
\end{lemma}
\noindent
We can now reprove Park induction for $\wpsymbol$ using Fatou's Lemma: Let $I$ be a superinvariant, i.e., $\charwp{\guard}{C}{f}(I) \preceq I$.
By \cref{thm:cond_expec}, the canonical stochastic process $\mathbf{X}^{f,I}$ satisfies
\begin{align*}
	\expec{\State}{X^{f,I}_{n+1} \mid \F{G}_n} \eeq X_n^{f, \charwp{\guard}{C}{f}(I)} \lleq X_n^{f,I}~.
\end{align*}
By applying $\phantom{}^{\State}\mathbb{E}$ on both sides, we obtain $\expec{\State}{X_0^{f,I}}\geq \expec{\State}{X_1^{f,I}}\geq \dots$.
This implies
\begin{align}
	\label{Fatou proof}
	\expec{\State}{X_0^{f,I}} \geq \expec{\State}{X^{f,I}_n} \geq \expec{\State}{X^{f,I}_n\cdot \indicator{(\termtime{\guard})^{-1}(\Nats)}}~,
\end{align}
as $X^{f,I}_n \geq X^{f,I}_n\cdot \indicator{(\termtime{\guard})^{-1}(\Nats)}$.
We conclude
\begin{align*}
	      & \Bigl(\lfp \charwp{\guard}{C}{f}\Bigr)(\State)                                                                                       \\
	\eeq  & \expec{\State}{X_{\termtime{\guard}}^f} \tag{by \cref{thm:lfp_expectation}}                                                          \\
	\eeq  & \expec{\State}{\lim\limits_{n \to \omega} X^{f,I}_n\cdot \indicator{(\termtime{\guard})^{-1}(\Nats)}} \tag{by \cref{lemma:as_limit}} \\
	\lleq & \lim\limits_{n \to \omega}\expec{\State}{X^{f,I}_n\cdot \indicator{(\termtime{\guard})^{-1}(\Nats)}} \tag{by Fatou's Lemma}          \\
	\lleq & \lim\limits_{n \to \omega}\expec{\State}{X^{f,I}_0} \tag{by \eqref{Fatou proof}}                                                     \\
	\eeq  & \expec{\State}{X^{f,I}_0}                                                                                                            \\
	\eeq  & \charwp{\guard}{C}{f}(I)(\State) \tag{by \cref{cor:connection_expectation}}                                                          \\
	\lleq & I(\State)~,\tag{since $\charwp{\guard}{C}{f}(I) \preceq I$}
\end{align*}
so $I$ is indeed an upper bound on the least fixed point.

Note that here we handle arbitrary loops, i.e., they are \emph{not} necessarily AST.
While $I$ being a superinvariant (plus some side conditions) still implies that $\mathbf{X}^{f,I}$ is a supermartingale, the second part of \cref{lemma:as_limit} is not applicable, i.e., in general we have $X_{\termtime{\guard}}^f \neq \lim_{n \to \omega} X_{n}^{f,I}$ if the loop is not AST.
So in this case we cannot use classic results from martingale theory.
Nevertheless, Fatou's Lemma combined with \cref{thm:lfp_expectation} and the first part of \cref{lemma:as_limit} provide a connection of Park induction for upper bounds to stochastic processes.
\renewcommand{\charwp}[3]{\charfun{\wpsymbol}{#1}{#2}{#3}}
\renewcommand{\charwpn}[4]{\charfunn{\wpsymbol}{#1}{#2}{#3}{#4}}

\section{Lower bounds on the expected runtime}
\label{sec:runtime}
So far, we have developed techniques for verifying lower bounds on weakest preexpectations, i.e., expected values of random variables upon program termination.
In this section, we transfer those techniques to verify lower bounds on \emph{expected runtimes} of probabilistic programs.
For this, we employ the $\ertsymbol$--transformer \cite{DBLP:journals/jacm/KaminskiKMO18,DBLP:conf/esop/KaminskiKMO16}, which is very similar to the $\wpsymbol$-transformer: Given program $C$ and post\emph{runtime} $t \in \E$, we are interested in the \emph{expected time} it takes to first execute $C$ and then let time $t$ pass (where $t$ is evaluated in the final states reached after termination of $C$).
Again, the behavior (and the runtime) of $C$ depends on its input, so we are actually interested in a \mbox{\emph{function}~$g \in \E$} mapping \emph{initial} states~$\State_0$ to the respective expected time.
For more details, see also \cite[Chapter 7]{thesis:kaminski}.
Similarly to weakest preexpectations, expected runtimes can be determined in a systematic and \emph{compositional} manner by means of the $\ertsymbol$ calculus:
\begin{definition}[The \textnormal{$\boldertsymbol$}--Transformer~\textnormal{\cite{DBLP:journals/jacm/KaminskiKMO18,DBLP:conf/esop/KaminskiKMO16}}]
	\label{defQQQert}
	Let $\pgcl$ be again the set of programs in the \emph{probabilistic guarded command language}.
	Then the expected~\mbox{runtime~transformer}
	\begin{align*}
		\ertsymbol\colon \pgcl \To \E \To \E
	\end{align*}
	is defined according to the rules given in \textnormal{\cref{table:ert}}.
	We call the function $\charert{\guard}{C}{t}$ the $\ertsymbol$--\emph{character{\-}is{\-}tic function} of the loop $\WHILEDO{\guard}{C}$ with respect to $t$.
	Its least fixed point is understood in terms of the partial order $\preceq$.
	To increase readability, we will again usually omit $\ertsymbol$, $\guard$, $C$, or $t$ from $\Phi$ whenever they are clear from the context.
	\begin{table}[t]
		\renewcommand{\arraystretch}{1.25}
		\begin{tabular}{@{\hspace{.1em}}l@{\hspace{1em}}l@{\hspace{.1em}}}
			\hline $\boldsymbol{C}$        & $\boldert{C}{t}$                                                                               \\[.25em]
			\hline $\SKIP$                 & $1 \pplus t$                                                                                   \\
			$\ASSIGN{b}{e}$                & $1 \pplus t\subst{b}{e}$                                                                       \\
			$\ITE{\guard}{C_1}{C_2}$       & $1 \pplus \iverson{\guard} \cdot \ert{C_1}{t} \pplus \iverson{\neg \guard} \cdot \ert{C_2}{t}$ \\
			$\PCHOICE{C_1}{p}{C_2}$        & $1 \pplus p \cdot \ert{C_1}{t} \pplus (1 - p) \cdot \ert{C_2}{t}$                              \\
			$\COMPOSE{C_1}{C_2}$           & $\ert{C_1}{\ert{C_2}{t}}$                                                                      \\
			$\WHILEDO{\guard}{\primed{C}}$ & $\lfp \charert{\primed{C}}{\guard}{t}$                                                         \\[.75em]
			\hline                                                                                                                          \\[-1.25em]
			\multicolumn{2}{c}{$\charert{\guard}{C}{t}(X) \eeq 1 \pplus \iverson{\neg \guard}\cdot t \pplus \iverson{\guard} \cdot \ert{C}{X} \qquad
					\begin{array}{c}
						\textnormal{\tiny characteristic} \\[-1em]
						\textnormal{\tiny function}
					\end{array}
			$}                                                                                                                              \\
			\\[-1em]
		\end{tabular}
		\hrule \vspace{1.25ex}
		\caption{Rules for the $\ertsymbol$--transformer.}
		\label{table:ert}
		\vspace{-1.5ex}
		\hrule
	\end{table}
	\noindent
\end{definition}
\begin{example}[Applying the $\boldertsymbol$ Calculus]
	\label{ex:ert-calc-a}
	Consider the probabilistic program $C$ given by
	\begin{align*}
		 & \COMPOSE{\PCHOICE{\ASSIGN{b}{b+5}}{\sfrac{4}{5}}{\ASSIGN{b}{10}}}{} \\
		 & \ITE{b = 10}{\SKIP}{\COMPOSE{\SKIP}{\SKIP}}
	\end{align*}
	Suppose we want to know the expected runtime of $C$. Then we need to determine $\ert{C}{0}$. Reusing the annotation styles of \textnormal{\cref{fig:wp-annotations}} for $\wpsymbol$, we make the following $\ertsymbol$ annotations:
	\allowdisplaybreaks%
	\begin{align*}
		 & \eqannotate{4 \pplus \iverson{b \neq5} \cdot \tfrac{4}{5} }                                                                                                 \\
		 & \ertannotate{1 \pplus \tfrac{4}{5} \cdot \bigl(1 + 2 + \iverson{b + 5 \neq 10} \bigr) \pplus \tfrac{1}{5} \cdot \bigl(1 + 2 + \iverson{10 \neq 10} \bigr) } \\
		 & \COMPOSE{\PCHOICE{\ASSIGN{b}{b+5}}{\sfrac{4}{5}}{\ASSIGN{b}{10}}}{}                                                                                         \\
		 & \eqannotate{2 \pplus \iverson{b \neq 10} }                                                                                                                  \\
		 & \ertannotate{1 \pplus \iverson{b = 10} \cdot (1 + 0) \pplus \iverson{b \neq 10} \cdot (1 + 1 + 0) }                                                         \\
		 & \ITE{b = 10}{\SKIP}{\COMPOSE{\SKIP}{\SKIP}}                                                                                                                 \\
		 & \annotate{0}
	\end{align*}
	\allowdisplaybreaks%
	At the top, we read off the expected runtime of $C$, namely $4 + \iverson{b \neq5} \cdot \tfrac{4}{5} $. This tells us that the expected runtime of $C$ is $4$ if started in an initial state where $b$ is $5$, and $4 + \tfrac{4}{5} = \tfrac{24}{5}$ otherwise.
\end{example}
\noindent
The $\ertsymbol$-- and the $\wpsymbol$--transformers are not only similar in definition, but they are closely connected by the following equality~\cite{DBLP:conf/lics/OlmedoKKM16}:
\begin{align*}
	\ert{C}{t} \eeq \ert{C}{0} + \wp{C}{t}~.
\end{align*}
In addition, reasoning about upper bounds by Park induction works exactly the same way. For reasoning about lower bounds using subinvariants, notice above that $\ert{C}{0}$ is \emph{independent of~$t$}. So, we can combine our derivation of \cref{thm:optional_stopping_probabilistic_programs} for lower bounds on \wpsymbol{} in \cref{sec:expectations_processes,sec:optional-stopping-wp} with the equation above to establish the first inductive rule for verifying lower bounds on expected runtimes:
\begin{restatable}[Inductive Lower Bounds on Expected Runtimes]{theorem}{thmlowerboundsert}
	\label{thm:lower_bounds_ert}
	Let $t, I \in \E$ with $t,I \pprec \infty$ and let $I$ harmonize with $t$. Furthermore, let 
	$\charertnoindex{\guard}{C}{t}$ be the $\ertsymbol$--characteristic function of the loop $\WHILEDO{\guard}{C}$ with respect to $t$. If $I$ is conditionally difference bounded and $\charwpnoindex{\guard}{C}{t}(I) \pprec \infty$, then
	\begin{align*}
		I \ppreceq \charertnoindex{\guard}{C}{t}(I) \qqimplies I \ppreceq \ert{\WHILEDO{\guard}{C}}{t}~.
	\end{align*}
	We call an $I$ that satisfies $I \preceq \charertnoindex{\guard}{C}{t}(I)$ a \emph{runtime subinvariant}.
\end{restatable}
\noindent
The proof of \cref{thm:lower_bounds_ert} can be found in \techrep{\cref{app:runtime_proofs}}\cameraready{\cite[App.\ F.1]{arxiv}}. We now illustrate the applicability of \cref{thm:lower_bounds_ert}:
\begin{example}[Coupon Collector~\textnormal{\protect{\cite{doi:10.1002/zamm.19300100113}}}]
	\label{ex:coupon-collector}
	Consider the well--known coupon collector's problem: There are $N$ different types coupons. A collector wants to collect at least one of each type. Each time she buys a new coupon, its type is drawn uniformly at random. How many coupons does she (expectedly) need to buy in order to have collected at least one coupon of each type?

	We can model this problem by the program $C_{\mathit{cc}}$ for some non--zero
        natural number $N \in \Nats$:\footnote{In \cite{popl}, the guard of the inner loop is just ``$x < i$'' and thus, the body of the outer loop (as a standalone program) is not AST. Here, we fix this by changing the guard to ``$0 < x < i$''. While we adapted the calculations accordingly, this does not affect our overall result since the outer loop is only entered if $x$ is positive.}
	\begin{align*}
		 & \COMPOSE{\ASSIGN{x}{N}}{}                     \\
		 & \WHILE{0 < x}                                 \\
		 & \qquad \COMPOSE{\ASSIGN{i}{N+1}}{}            \\
		 & \qquad \WHILE{0 < x < i}                      \\
		 & \qquad \qquad \ASSIGN{i}{\mathrm{Unif}[1..N]} \\
		 & \qquad \COMPOSE{\}}{}                         \\
		 & \qquad \ASSIGN{x}{x-1}                        \\
		 & \}~,
	\end{align*}
	Variable $x$ represents the number of uncollected coupon types. The inner loop models the buying of new coupons until an uncollected type is drawn.\footnote{The random assignment $\ASSIGN{i}{\mathrm{Unif}[1..N]}$ does --- strictly speaking --- not adhere to our $\pgcl$ syntax, but it can be modeled in $\pgcl$. For the sake of readability, we opted for $\ASSIGN{i}{\mathrm{Unif}[1..N]}$.}

	The expected runtime of $C_{\mathit{cc}}$ is proportional to the expected number
        of coupons the collector needs to buy. We want to prove that $N \cdot \harm{N}$ is a lower bound on that expected runtime, where $\harm{m}$ is the $m$-th harmonic number, i.e., $\harm{0} = 0$ and $\harm{m} = \sum_{k = 1}^{m} \tfrac{1}{k}$. For this, we make the following annotations, reusing the annotation style of \textnormal{\cref{fig:loop-annotations}} (for more detailed annotations, see \textnormal{\techrep{\cref{app:coupon-collector}}\cameraready{\cite[App.\ F.2]{arxiv}}}):

	\allowdisplaybreaks
	\begin{align*}
		 & \eqannotate{1 \pplus N \cdot \harm{N} }                                                                                                       \\
		 & \ertannotate{1 \pplus \iverson{0 < N \leq N} \cdot N \cdot \harm{N} \pplus \iverson{N < N} \cdot (N \cdot \harm{N} + N - N)}                        \\
		 & \COMPOSE{\ASSIGN{x}{N}}{}                                                                                                               \\
		 & \preceqannotate{\iverson{0 < x \leq N} \cdot N \cdot \harm{x} \pplus \iverson{N < x} \cdot (N \cdot \harm{N} + x - N)}                        \\
		 & \eqannotate{
			1 \pplus \iverson{0 < x \leq N} \cdot \left( N \cdot \harm{x} + 3 + \tfrac{N}{x}\right) + \iverson{N < x} \cdot (2 + N \cdot \harm{N} + x - N)
		}
		\\
		 & \phiannotate{
			1 \pplus \iverson{x \leq 0} \cdot 0 \pplus \iverson{0 < x} \cdot \Bigl( 2 + t + \iverson{0 < x \leq N} \cdot \tfrac{2 \cdot N}{x}
			\Bigr)
		}                                                                                                                                          \\
		 & \WHILE{0 < x}                                                                                                                           \\
		 & \qquad \eqannotate{
			2 \pplus t \pplus \iverson{0 < x \leq N} \cdot \tfrac{2 \cdot N}{x}
		}                                                                                                                                          \\
		 & \qquad \ertannotate{
			1 \pplus 1 \pplus t \pplus \iverson{0 < x < N + 1} \cdot 2 \cdot \Max{\tfrac{N}{x}}{1}
		}                                                                                                                                          \\
		 & \qquad \COMPOSE{\ASSIGN{i}{N+1}}{}                                                                                                      \\
		 & \qquad \ertannotate{
			1 \pplus t \pplus \iverson{0 < x < i} \cdot 2 \cdot \Max{\tfrac{N}{x}}{1}
		} \tag{by \techrep{\cref{app:coupon-collector}, \cref{lem:batz-special-case}}\cameraready{\cite[App.\ F.2, Lemma 86]{arxiv}}}              \\
		 & \qquad \COMPOSE{\WHILEDO{0 < x < i}{ \ASSIGN{i}{\mathrm{Unif}[1..N]} }}{}                                                               \\
		 & \qquad \eqannotate{\underbrace{
			1 \pplus \iverson{1 < x \leq N + 1} \cdot N \cdot \left(\harm{x} - \tfrac{1}{x}\right) \pplus \iverson{N + 1 < x} \cdot (N \cdot \harm{N} + x - 1 - N)
		}_{~{}~{}~{}~{}~{}\eqqcolon t}}                                                                                                            \\
		 & \qquad \ertannotate{1 \pplus \iverson{0 < x - 1 \leq N} \cdot N \cdot \harm{x - 1} \pplus \iverson{N < x - 1} \cdot (N \cdot \harm{N} + (x-1) - N)} \\
		 & \qquad \ASSIGN{x}{x-1}                                                                                                                  \\
		 & \qquad \starannotate{\iverson{0 < x \leq N} \cdot N \cdot \harm{x} \pplus \iverson{N < x} \cdot (N \cdot \harm{N} + x - N)} ~\}                     \\
		 & \annotate{0}
	\end{align*}
	\allowdisplaybreaks
	By our above annotations, we have shown that\footnote{In \cite{popl}, there was a typo in the invariant $I$. In the second case, it said ``$N - x$'' which leads to negative values. Here, we correct this mistake and adapt the calculations accordingly. Again, this does not change the overall result.}
	\begin{align*}
		I \eeq \iverson{0 < x \leq N} \cdot N \cdot \harm{x} \pplus \iverson{N < x} \cdot (N \cdot \harm{N} + x - N)
	\end{align*}
	is indeed a runtime subinvariant of the outer loop. Before we finish proving that $I$ is indeed a lower bound on the expected runtime of the outer loop, let us take a closer look at the meaning of $I$: If $I$ is a lower bound, the outer loop takes at least expected runtime $N \cdot \harm{x}$ if $x$ is between $1$ and $N$, and expected runtime $N \cdot \harm{N} + x - N$ if $x$ is larger than $N$. In the second case, the too--large $x$ value suggests that we have to collect more coupons than there are different coupons. So we first collect $x - N$ arbitrary ``excess coupons" before we enter the ``normal coupon collector mode" and collect the remaining $N$ coupons in expected time $N \cdot \harm{N}$. Indeed, $N \cdot \harm{x}$ (without case analysis) is \emph{not} a lower bound on the expected runtime and we would in fact fail to prove its subinvariance.

	For the inner loop, we have used the fact that this loop is a so--called \emph{independent and identically distributed loop}, for which \emph{exact} expected runtimes can be determined~\mbox{\textnormal{\cite[Theorem 4]{DBLP:conf/esop/BatzKKM18}}}. For more details, see \textnormal{\techrep{\cref{app:coupon-collector}, \cref{lem:batz-special-case}}\cameraready{\cite[App.\ F.2, Lemma 86]{arxiv}}}. We stress that while in this case we had an \emph{exact} expected runtime for the inner loop available by external techniques, a suitable \emph{underapproximation} of the expected runtime of the inner loop using the technique presented in this paper~(\textnormal{\cref{thm:lower_bounds_ert}}) would have worked as well. Hence, our technique \emph{is} generally applicable to nested loops.

	At the very top of the above annotations, we push $I$ over the initial assignment,
        thus verifying $1 + N \cdot \harm{N}$ (and hence also~$N \cdot \harm{N}$) as lower bound for the entire expected runtime of $C_{\mathit{cc}}$.

	In order to establish that the subinvariant $I$ is in fact a lower bound, we are still left to prove conditional difference boundedness of $I$. For this, we first make the following annotations:
	\begin{align*}
		 & \annotate{\bigl| I\subst{x}{x-1} - I(s) \bigr|} \tag{$I$ does not depend on $i$}   \\
		 & \ASSIGN{i}{N+1}                                                                    \\
		 & \annotate{\bigl| I\subst{x}{x-1} - I(s) \bigr|}
		\tag{by almost-sure term.~of inner loop and \cite[Lemma 1]{DBLP:conf/esop/BatzKKM18}} \\
		 & \WHILEDO{0 < x < i}{ \ASSIGN{i}{\mathrm{Unif}[1..N]} }                             \\
		 & \annotate{\bigl| I\subst{x}{x-1} - I(s)\bigr|}                                     \\
		 & \ASSIGN{x}{x-1}                                                                    \\
		 & \annotate{\bigl| I - I(s) \bigr|}
	\end{align*}
	Now that we have determined $\smash{\wp{\mathit{outer~loop~body}}{\bigl| I - I(s) \bigr|}}$, we finally bound $\Diff{I}$:
	\begin{align*}
		\Diff{I}
		\eeq     & \lambda s \mydot \wp{\mathit{outer~loop~body}}{\bigl| I - I(s) \bigr|} \\
		\eeq     & \lambda s \mydot \bigl| I\subst{x}{x-1} - I(s) \bigr|                  \\
		\eeq     & \iverson{x = 1} \cdot N \pplus \iverson{1 < x < N} \cdot \tfrac{N}{x}
		\pplus \iverson{N + 1 \leq x}
		\tag{by case analysis}                                                            \\
		\ppreceq & N
	\end{align*}
	Hence, $\Diff{I}$ is bounded by a constant, as $N$ is constant \emph{within} the program $C_{\mathit{cc}}$. Finally, we would still have to show $\charwpnoindex{\guard}{C}{t}(I) \pprec \infty$, which is easily checked and thus omitted here. This concludes our lower bound proof for the coupon collector's problem.
\end{example}
\noindent
In the example above, we have verified that $N \cdot \harm{N}$ is a lower bound on the expected runtime of the coupon collector program. This lower bound enjoys several nice properties: For one, our lower bound is an \emph{exact} asymptotic lower bound. Another fact is that our lower bound is a \emph{strict} lower bound. The actual runtime is a bit higher, as we have omitted some constants. This is, however, a desirable fact, as often we are only interested in the asymptotic runtime and do not wish to bother with the constants. Notice further, that we never had to find the limit of any sequence. Loop semantics (be it $\wpsymbol$ or $\ertsymbol$) were all applied only finitely many times in order to verify a tight asymptotic lower bound.\footnote{This is also true for the technique we used for the inner loop.}
All in all, the above example demonstrates the effectiveness of our inductive lower bound rule.

\section{Related Work}
\label{sec:related}

\paragraph{Weakest preexpectation reasoning.}
The weakest preexpectation calculus goes back to the predicate transformer calculus by~\cite{DBLP:journals/cacm/Dijkstra75,DBLP:books/ph/Dijkstra76}, which provides an important tool for qualitative formal reasoning about nonprobabilistic programs.
The probabilistic and quantitative analog to predicate transformers for nonprobabilistic programs are \emph{expectation transformers} for probabilistic programs.
Weakest--preexpectation--style reasoning was first studied in seminal work on probabilistic propositional dynamic logic (PPDL) by~\cite{DBLP:conf/stoc/Kozen83,DBLP:journals/jcss/Kozen85}.
Its box-- and diamond--modalities provide probabilistic versions of Dijkstra's weakest (liberal) preconditions.
Amongst others, \cite{DBLP:phd/ethos/Jones90}, \cite{DBLP:journals/toplas/MorganMS96}, \cite{DBLP:series/mcs/McIverM05}, and \cite{Hehner:FAC:2011} have furthered this line of research, e.g., by considering nondeterminism and proof rules for bounding preexpectations in the presence of loops.
Work towards automation of weakest preexpectation reasoning was carried out, amongst others, by \cite{ChenCAV2015}, \cite{DBLP:journals/afp/Cock14}, \cite{DBLP:conf/sas/KatoenMMM10}, and \cite{DBLP:conf/atva/FengZJZX17}.
Abstract interpretation of probabilistic programs was studied in this setting by \cite{DBLP:journals/scp/Monniaux05}.

\paragraph{Bounds on weakest preexpectations.}

Rules for bounding weakest preexpectations were considered from very early on.
Already \cite{DBLP:conf/stoc/Kozen83} provides an induction rule for verifying upper bounds.
Pioneering work on lower bounds by means of limits of sequences was carried out by \cite{DBLP:phd/ethos/Jones90} and later reconsidered by \cite{DBLP:journals/scp/AudebaudP09}.
Proof rules that do not make use of limits were studied by \cite{morgan1996proof} and later more extensively in \techrep{\linebreak}\cite{DBLP:series/mcs/McIverM05}.
An orthogonal approach to lower bounds by means of bounded model checking was explored by~\cite{DBLP:conf/atva/0001DKKW16}.

\paragraph{Advanced weakest preexpectation calculi.}

Apart from reasoning about expected values of random variables at termination of simple $\pgcl$ programs, more advanced expectation--based calculi were invented.
For instance, \cite{DBLP:journals/igpl/MorganM99} use expectation transformers to reason about temporal logic.
More recently, \cite{DBLP:journals/toplas/OlmedoGJKKM18} studies expectation transformers for probabilistic programs with \emph{conditioning}.
\cite{DBLP:conf/esop/KaminskiKMO16,DBLP:conf/lics/OlmedoKKM16,DBLP:journals/jacm/KaminskiKMO18} introduce expectation based calculi to reason about expected runtimes of probabilistic programs.
\cite{QSLpopl} present a quantitative separation logic together with a weakest preexpectation calculus for verifying probabilistic programs with pointer--access to dynamic memory.

In all of the above works, the rules for lower bounds rely throughout on finding limits of sequences as well as the sequences themselves.
In particular, the proof of the (exact) expected runtime of the coupon collector by \cite{DBLP:conf/esop/KaminskiKMO16} requires a fairly complicated sequence, whereas our invariant in \cref{ex:coupon-collector} was conceptually fairly easy and thus more informative for a human.

\paragraph{Martingale--based reasoning.}

Probabilistic program analysis using martingales was pioneered by \cite{DBLP:conf/cav/ChakarovS13}.
Our rules rely on the notions of \emph{uniform integrability} and \emph{conditional difference boundedness} as well as the \emph{Optional Stopping Theorem}.
Previous works have also used these notions.
\cite{DBLP:conf/cav/BartheEFH16} focus on synthesizing \emph{exact} martingale expressions.
\cite{DBLP:conf/popl/FioritiH15}
develop a type system for uniform integrability in order to prove (positive) almost--sure termination\footnote{Termination with probability 1 (within finite expected time).} of probabilistic programs and give upper bounds on the expected runtime.
\cite{DBLP:conf/vmcai/FuC19}
give lower bounds on expected runtimes.
\cite{DBLP:journals/corr/abs-1811-02133} provide a semi--decision procedure for lower bounding termination probabilities of probabilistic higher--order recursive programs.
\cite{DBLP:conf/pldi/NgoC018} perform automated template--driven resource analysis, but infer \mbox{upper bounds only}.

The latter four works analyze the termination behavior of a probabilistic program, whereas we focus on \emph{general} expected values, e.g., of program variables.
Furthermore, we do not only \emph{make use} of uniform integrability and/or conditional difference boundedness of some auxiliary stochastic process in order to prove soundness of our proof rules but establish tight connections between expectation--based reasoning via induction and martingale--based reasoning.

Other work on probabilistic program analysis by specialized kinds of martingales includes \techrep{\linebreak}\cite{DBLP:conf/sas/ChakarovS14}, \cite{DBLP:conf/popl/ChatterjeeFNH16}, \cite{DBLP:conf/popl/ChatterjeeNZ17}, \cite{DBLP:journals/corr/abs-1709-04037}, \cite{DBLP:conf/aplas/HuangFC18}, \cite{DBLP:conf/vmcai/FuC19}, and \cite{DBLP:conf/pldi/Wang0GCQS19}.
For instance, regarding expected runtimes of probabilistic (and possibly nondeterministic) programs, \cite{DBLP:conf/vmcai/FuC19} construct \emph{difference bounded} (as opposed to \emph{conditionally} difference bounded, which is a strictly weaker requirement) supermartingales which have to correspond to the \emph{exact} asymptotic expected runtime.
In contrast, our rule allows for reasoning about \emph{strict} lower bounds.

\section{Conclusion}
In this paper, we have studied proof rules for lower bounds in probabilistic program verification.
Our rules are \emph{simple} in the sense that the invariants need to be ``pushed through the loop semantics" only a \emph{finite} number of times, much like invariants in Hoare logic.
In contrast, existing rules for lower bounds of unbounded weakest preexpectations required coming up with an infinite \emph{sequence} of invariants, performing induction to prove relative inductiveness of two subsequent invariants, and then --- most unpleasantly --- finding the limit of this sequence.
The main results of this paper are the following:
\begin{enumerate}
	\item We have presented the first \emph{inductive proof rules}
	      (\cref{thm:optional_stopping_probabilistic_programs}~(a) and (b)) for verifying \emph{lower bounds on (possibly unbounded) weakest preexpectations of probabilistic while loops} using quantitative \emph{invariants}.
	      Our inductive rules are given as an \emph{Optional Stopping Theorem (OST) for weakest preexpectations}.
	      They provide sufficient conditions for the requirement of \emph{uniform integrability} which are much easier to check than uniform integrability in general.
	      Case studies demonstrating the effectiveness but also the limitations of these rules are found in~\techrep{\cref{app:examples}}\cameraready{\cite[App.\ A]{arxiv}}.

	\item For proving our OST, we resort to the classical OST from probability theory.
	      However, for most notions that appear in the classical OST, like \emph{uniform integrability} and \emph{conditional difference boundedness}, we were able to find purely expectation--transformer--based counterparts (see \cref{sec:expectations_processes,sec:optional-stopping-wp}).
	      We thus conjecture that our OST can be proven in purely expectation--theoretic terms, which would most likely simplify the proof of our OST significantly as no probability theory would be required anymore.

	\item We studied the inductive proof rules for lower bounds on \emph{bounded} weakest preexpectations from \cite{DBLP:series/mcs/McIverM05}.
	      Our results gave rise to a generalization of their proof rule to a \emph{sufficient and necessary} criterion for lower bounds.
	      (\cref{thm:generalization_mciver_morgan}).

	\item We have investigated a measure theoretical explanation for why verifying upper bounds using domain theoretical Park induction is conceptually simpler (\cref{sec:fatou}).
	      The underlying reason is the well--known \emph{Lemma of Fatou}.
	      This leads us to speculate that Fatou's Lemma could be proved in purely domain theoretical terms, perhaps as an instance of Park induction.
	      A successful attempt at a similar idea is due to \cite{DBLP:journals/dm/Baranga91} who proved that the well--known \emph{Banach Contraction Principle} is a particular instance of the Kleene Fixed Point Theorem.

	\item We used the close connection between $\wpsymbol$ and $\ertsymbol$ to present the first inductive proof rule for lower bounding expected runtimes (\cref{thm:lower_bounds_ert}).
	      As an example to demonstrate the power of this rule, we inferred a nontrivial lower bound on the expected runtime of the famous coupon collector's problem (\cref{ex:coupon-collector}).
\end{enumerate}

\noindent{}
Future work includes extending our proof rules for weakest preexpectation reasoning to recursive programs \cite{DBLP:conf/lics/OlmedoKKM16}, to probabilistic programs with nondeterminism~\cite{DBLP:journals/tcs/McIverM01a,DBLP:series/mcs/McIverM05}, and to \emph{mixed--sign} postexpectations.
For the latter, this will likely yield more appealing proof rules for loops than those provided in~\cite{DBLP:conf/lics/KaminskiK17} which currently involve reasoning about sequences.
Moreover, we are interested in (partially) automating the synthesis of the quantitative invariants needed in our proof rules.

\begin{acks}                            
The authors gratefully acknowledge the support of the German Research Council (DFG)
Research Training Group 2236 UnRAVeL and ERC Advanced Grant 787914 FRAPPANT. Furthermore,
we would like to thank Florian Frohn and Christoph Matheja for many fruitful discussions on examples and counterexamples.
\end{acks}

\bibliography{inductive_lower_bounds}

\techrep{
\clearpage
\appendix
\section*{Appendix}
This appendix contains additional material for our paper.
\cref{app:examples} presents a collection of case studies to demonstrate the strengths and the limitations of our
rule. In \cref{app:prelim_probability}, we give a more detailed introduction into
the required preliminaries from probability theory. Afterwards, in
\cref{app:proofs_expectations_processes}, we present the proofs for
\cref{sec:expectations_processes}. \cref{app:proofs_optional_stopping} then contains
the proofs of our main results in \cref{sec:optional-stopping-wp}. In \cref{app:mciver_morgan} we give the proofs for \cref{sec:mciver_morgan}. Finally, \cref{app:runtime} contains the proofs of our result for lower bounding the expected runtime from \cref{sec:runtime}.

\noindent
\renewcommand{\charwp}[3]{\Phi_{#3}}
\renewcommand{\charwpn}[4]{\Phi_{#3}^{#4}}
\section{Case Studies}
\label{app:examples}
\begin{example}
	[Negative Binomial Loop \protect{(cf.
			\cite{POPL18})}] \label{ex:negative_binomial_loop}
	Let us consider the program $C_{neg}$
	\begin{align*}
		 & \WHILE{x > 0}                                                     \\
		 & \qquad \PCHOICE{\ASSIGN{x}{x-1}}{\sfrac{1}{2}}{\ASSIGN{k}{k + 1}} \\
		 & \}~,
	\end{align*}
	with $x,k\in \Nats$.
	The characteristic function for $f=k$ of the program is given by
	\begin{align*}
		\charwp{\guard}{C}{f}(X) \eeq \iverson{x=0}\cdot k + \iverson{x>0}\cdot \tfrac{1}{2}\cdot \Bigl( X\subst{x}{x-1} + X\subst{k}{k+1}\Bigr)~.
	\end{align*}
	The loop is expected to be executed $2 \cdot \iverson{x>0}\cdot x$ times, so its expected looping time is finite.
	Intuitively, the value of $k$ after termination of the program is $I=\iverson{x=0}\cdot k + \iverson{x>0} \cdot (k+x)$, i.e., the initial value of $k$ increases by the initial value of $x$ if the loop can be executed at all.
	Note that $I$ harmonizes with $f$.
	We will prove that our intuition for $I$ is correct, i.e., that $I$ is indeed the least fixed point of $\charwp{\guard}{C}{f}$.

	First of all, it is a fixed point of $\charwp{\guard}{C}{f}$:

	\begin{align*}
		\charwp{\guard}{C}{f}(I) {} = {} & \iverson{x=0}\cdot k + \iverson{x>0}\cdot \tfrac{1}{2}\cdot\Bigl( I\subst{x}{x-1} + I\subst{k}{k+1}\Bigr)                                                                        \\
		{} = {}                          & \iverson{x=0}\cdot k + \iverson{x>0}\cdot \tfrac{1}{2}\cdot \Bigl( \iverson{x-1=0}\cdot k + \iverson{x-1>0} \cdot (k+x-1)                                                        \\
		                                 & \quad+ \iverson{x=0}\cdot k + \iverson{x>0} \cdot (k+1+x)\Bigr)                                                                                                                  \\
		{} = {}                          & \iverson{x=0}\cdot k + \iverson{x>0}\cdot \tfrac{1}{2}\cdot \Bigl( \iverson{x=1}\cdot k + \iverson{x>1} \cdot (k+x-1) + \iverson{x=0}\cdot k + \iverson{x>0} \cdot (k+1+x)\Bigr) \\
		{} = {}                          & \iverson{x=0}\cdot k + \iverson{x>0}\cdot \tfrac{1}{2}\cdot \Bigl( \iverson{x>0} \cdot (k+x-1) + \iverson{x>0} \cdot (k+1+x)\Bigr)                                               \\
		{} = {}                          & \iverson{x=0}\cdot k + \iverson{x>0} \cdot (k+x)=I
	\end{align*}
	Since $I$ is indeed a fixed point of $\charwp{\guard}{C}{f}$, it is also a subinvariant and furthermore finite.
	To apply \textnormal{\cref{thm:optional_stopping_probabilistic_programs} (b)} we need to check that $I$ is conditionally difference bounded.
	We derive
	\begin{align*}
		        & \Diff{I}(\State)                                                                                                                                                                                             \\
		{} = {} & \iverson{x>0}(\State)\cdot \tfrac{1}{2}\cdot\Bigl(\abs{I-I(\State)}\subst{x}{x-1}+\abs{I-I(\State)}\subst{k}{k+1}\Bigr)                                                                                      \\
		{} = {} & \iverson{x>0}(\State)\cdot \tfrac{1}{2}                                                                                                                                                                      \\
		        & \quad\cdot\Bigl(|\iverson{x-1=0}(\State)\cdot \State(k) +\iverson{x-1>0}(\State)\cdot(\State(k) + \State(x) - 1)                                                                                             \\
		        & \qquad-\iverson{x=0}(\State)\cdot \State(k) -\iverson{x>0}(\State)\cdot (\State(k)+\State(x))|                                                                                                               \\
		        & \qquad +|\iverson{x=0}(\State)\cdot (\State(k)+1)+\iverson{x>0}(\State)\cdot(\State(k)+1+\State(x))                                                                                                          \\
		        & \qquad-\iverson{x=0}(\State)\cdot \State(k) - \iverson{x>0}(\State)\cdot(\State(k)+\State(x))|\Bigr)                                                                                                         \\
		{} = {} & \iverson{x>0}(\State)\cdot \tfrac{1}{2} \cdot                                                                                                                                                                \\
		        & \quad\Bigl(\abs{\iverson{x=1}(\State)\cdot \State(k) +\iverson{x>1}(\State)\cdot(\State(k) + \State(x) - 1)-\iverson{x=0}(\State)\cdot \State(k) -\iverson{x>0}(\State)\cdot (\State(k)+\State(x))}          \\
		        & \quad+\abs{\iverson{x=0}(\State)+\iverson{x>0}(\State)}\Bigr)                                                                                                                                                \\
		{} = {} & \iverson{x>0}(\State)\cdot \tfrac{1}{2}\cdot\Bigl(\abs{\iverson{x>0}(\State)\cdot(\State(k) + \State(x) - 1)-\iverson{x=0}(\State)\cdot \State(k) -\iverson{x>0}(\State)\cdot (\State(k)+\State(x))}+1\Bigr) \\
		{} = {} & \iverson{x>0}(\State)\cdot \tfrac{1}{2}\cdot\Bigl(\abs{-\iverson{x>0}(\State)-\iverson{x=0}(\State)\cdot \State(k) }+1\Bigr)                                                                                 \\
		{} = {} & \iverson{x>0}(\State)\cdot \tfrac{1}{2}\cdot\Bigl(\iverson{x>0}(\State) +1\Bigr)                                                                                                                             \\
		{} = {} & \iverson{x>0}(\State)\leq 1~.
	\end{align*}
	So $I$ is indeed conditionally difference bounded by $1$.
	Hence, we can apply our new Optional Stopping Theorem (\textnormal{\cref{thm:optional_stopping_probabilistic_programs} (b)}) to obtain $I \preceq \lfp \charwp{\guard}{C}{f}$.
	As $I$ is also fixed point itself, it is the least fixed point of $\charwp{\guard}{C}{f}$, i.e.,
	\[
		I=\wp{C_{neg}}{f}~.
	\]
\end{example}

\begin{example}[Fair in the Limit Negative Binomial Loop]
	Consider the slight adaption of $C_{neg}$ from \textnormal{\cref{ex:negative_binomial_loop}} to $C_{filneg}$:
	\begin{align*}
		 & \WHILE{x > 0}                                                                 \\
		 & \qquad \PCHOICE{\ASSIGN{x}{x-1}}{\frac{1}{2+\sfrac{1}{x}}}{\ASSIGN{k}{k + 1}} \\
		 & \}~,
	\end{align*}
	with $x,k\in \Nats$.
	Note that for any $x>0$ we have $\tfrac{1}{3} \leq \frac{1}{2+\sfrac{1}{x}} \leq \tfrac{1}{2}$ and $\frac{1}{2+\sfrac{1}{x}}$ is monotonically increasing in the value of $x$.
	Therefore one can show using the \ertsymbol--transformer (cf.\ \textnormal{\cite{DBLP:journals/jacm/KaminskiKMO18}}) that the expected runtime of $C_{filneg}$ is at most $3 \cdot \iverson{x>0} \cdot x$.
	So we have positive almost--sure termination and therefore finite expected looping time.
	Again, we would like to reason about the expected value of $k$ after termination of $C_{filneg}$.
	The characteristic function for $f=k$ of the program is given by $$\charwp{\guard}{C}{f}(X) = \iverson{x=0}\cdot k + \iverson{x>0}\cdot \Bigl(\frac{1}{2+\sfrac{1}{x}}\cdot X\subst{x}{x-1} + \frac{1+\sfrac{1}{x}}{2+\sfrac{1}{x}}\cdot X\subst{k}{k+1}\Bigr)~.$$

	\noindent
	Intuitively, the value of $k$ after termination of the program should again be at least $I=\iverson{x=0}\cdot k + \iverson{x>0} \cdot (k+x)$, i.e., the initial value of $k$ again increases at least by the initial value of $x$ if the loop can be executed at all.
	We will prove that this intuition is correct, so that $I$ is indeed a lower bound on the least fixed point of $\charwp{\guard}{C}{f}$.
	Again, $I$ harmonizes with $f$.

	We first show that $I$ is a subinvariant:
	\begin{align*}
		\charwp{\guard}{C}{f}(I) {} = {} & \iverson{x=0}\cdot k + \iverson{x>0}\cdot \Bigl(\frac{1}{2+\sfrac{1}{x}}\cdot I\subst{x}{x-1} + \frac{1+\sfrac{1}{x}}{2+\sfrac{1}{x}}\cdot I\subst{k}{k+1}\Bigr)                 \\
		{} = {}                          & \iverson{x=0}\cdot k + \iverson{x>0}\cdot \Bigl(\frac{1}{2+\sfrac{1}{x}}\cdot \left(\iverson{x-1=0}\cdot k + \iverson{x-1>0} \cdot (k+x-1)\right)                                \\
		                                 & \quad+ \frac{1+\sfrac{1}{x}}{2+\sfrac{1}{x}}\cdot \left(\iverson{x=0}\cdot (k+1) + \iverson{x>0} \cdot (k+1+x)\right)\Bigr)                                                      \\
		{} = {}                          & \iverson{x=0}\cdot k + \iverson{x>0}\cdot \Bigl(\frac{1}{2+\sfrac{1}{x}}\cdot \left(\iverson{x>0} \cdot (k+x-1)\right) + \frac{1+\sfrac{1}{x}}{2+\sfrac{1}{x}}\cdot              \\
		                                 & \quad \left(\iverson{x=0}\cdot (k+1) + \iverson{x>0} \cdot (k+1+x)\right)\Bigr)                                                                                                  \\
		{} = {}                          & \iverson{x=0}\cdot k + \Bigl(\frac{1}{2+\sfrac{1}{x}}\cdot \left(\iverson{x>0}\cdot\iverson{x>0} \cdot (k+x-1)\right)                                                            \\
		                                 & \quad+ \frac{1+\sfrac{1}{x}}{2+\sfrac{1}{x}}\cdot \left(\iverson{x>0}\cdot\iverson{x=0}\cdot (k+1) + \iverson{x>0}\cdot\iverson{x>0} \cdot (k+1+x)\right)\Bigr)                  \\
		{} = {}                          & \iverson{x=0}\cdot k + \iverson{x>0}\cdot \Bigl(\frac{1}{2+\sfrac{1}{x}}\cdot \left(k+x-1\right) + \frac{1+\sfrac{1}{x}}{2+\sfrac{1}{x}}\cdot \left(k+1+x\right)\Bigr)           \\
		{} = {}                          & \iverson{x=0}\cdot k + \iverson{x>0}\cdot \Bigl(\frac{k+x-1 +k+1+x +\sfrac{k}{x}+\sfrac{1}{x}+1}{2+\sfrac{1}{x}}\Bigr)                                                           \\
		{} = {}                          & \iverson{x=0}\cdot k + \iverson{x>0}\cdot \Bigl(\frac{2k+2x +\sfrac{k}{x}+\sfrac{1}{x}+1}{2+\sfrac{1}{x}}\Bigr)                                                                  \\
		{} = {}                          & \iverson{x=0}\cdot k + \iverson{x>0}\cdot \Bigl(\frac{\left(2+\sfrac{1}{x}\right)(k+x)+\sfrac{1}{x}}{2+\sfrac{1}{x}}\Bigr)                                                       \\
		{} = {}                          & \iverson{x=0}\cdot k + \iverson{x>0}\cdot (k+x)+ \iverson{x>0} \cdot \frac{\sfrac{1}{x}}{2+\sfrac{1}{x}} = I + \iverson{x>0} \cdot \frac{\sfrac{1}{x}}{2+\sfrac{1}{x}} \succeq I
	\end{align*}
	So $I$ is indeed a subinvariant of $\charwp{\guard}{C}{f}$ and $\charwp{\guard}{C}{f}(I)\pprec \infty$.
	To apply \textnormal{\cref{thm:optional_stopping_probabilistic_programs} (b)} we need to check that $I$ is conditionally difference bounded.
	\begin{align*}
		        & \Diff{I}(\State)                                                                                                                                                                                                                                                           \\
		{} = {} & \left(\iverson{x>0}(\State)\cdot \Bigl(\frac{1}{2+\sfrac{1}{x}}\cdot \abs{I-I(\State)}\subst{x}{x-1} + \frac{1+\sfrac{1}{x}}{2+\sfrac{1}{x}}\cdot \abs{I-I(\State)}\subst{k}{k+1}\Bigr)\right)(\State)                                                                     \\
		{} = {} & \iverson{x>0}(\State)\cdot \Bigl(\frac{1}{2+\sfrac{1}{\State(x)}} \cdot |\iverson{x-1=0}(\State)\cdot \State(k) + \iverson{x-1>0}(\State) \cdot (\State(k)+\State(x)-1)                                                                                                    \\
		        & \qquad\qquad\qquad-\iverson{x=0}(\State)\cdot \State(k) - \iverson{x>0}(\State) \cdot (\State(k)+\State(x))|                                                                                                                                                               \\
		        & \qquad\qquad\qquad + \frac{1+\sfrac{1}{\State(x)}}{2+\sfrac{1}{\State(x)}}\cdot \abs{I-I(\State)}(\State \subst{k}{k+1})\Bigr)                                                                                                                                             \\
		{} = {} & \iverson{x>0}(\State)\cdot \Bigl(\frac{1}{2+\sfrac{1}{\State(x)}}\cdot \abs{\iverson{x>0}(\State) \cdot (\State(k)+\State(x)-1)-\iverson{x=0}(\State)\cdot \State(k) - \iverson{x>0}(\State) \cdot (\State(k)+\State(x))}                                                  \\
		        & \quad+ \frac{1+\sfrac{1}{\State(x)}}{2+\sfrac{1}{\State(x)}}\cdot \abs{I-I(\State)}(\State \subst{k}{k+1})\Bigr)                                                                                                                                                           \\
		{} = {} & \iverson{x>0}(\State)\cdot \Bigl(\frac{1}{2+\sfrac{1}{\State(x)}}\cdot \abs{\iverson{x>0}(\State) \cdot (-1)-\iverson{x=0}(\State)\cdot \State(k) }                                                                                                                        \\
		        & \qquad\qquad\qquad+ \frac{1+\sfrac{1}{\State(x)}}{2+\sfrac{1}{\State(x)}}\cdot |\iverson{x=0}(\State)\cdot (\State(k)+1) + \iverson{x>0}(\State) \cdot (\State(k)+1+\State(x))                                                                                             \\
		        & \qquad\qquad\qquad\qquad\qquad\qquad-\iverson{x=0}(\State)\cdot \State(k) - \iverson{x>0}(\State) \cdot (\State(k)+\State(x))|\Bigr)                                                                                                                                       \\
		{} = {} & \iverson{x>0}(\State)\cdot \Bigl(\frac{1}{2+\sfrac{1}{\State(x)}}\cdot \abs{\iverson{x>0}(\State) \cdot (-1)-\iverson{x=0}(\State)\cdot \State(k) } + \frac{1+\sfrac{1}{\State(x)}}{2+\sfrac{1}{\State(x)}}\cdot \abs{\iverson{x=0}(\State)+ \iverson{x>0}(\State)} \Bigr) \\
		{} = {} & \iverson{x>0}(\State)\cdot \Bigl( \frac{1}{2+\sfrac{1}{\State(x)}} + \frac{1 + \sfrac{1}{\State(x)}}{2+\sfrac{1}{\State(x)}}\Bigr) = \iverson{x>0}(\State) \leq 1
	\end{align*}
	So $I$ is indeed conditionally difference bounded by $1$.
	Hence, by \textnormal{\cref{thm:optional_stopping_probabilistic_programs} (b)}
	$I$ is a lower bound on the least fixed point of $\charwp{\guard}{C}{f}$, i.e.,
	\[
		I\preceq\wp{C_{filneg}}{f}.
	\]
\end{example}

\begin{example}[Negative Binomial Loop with Non--Constant Updates]
	Let us consider another adaption of $C_{neg}$ from \textnormal{\cref{ex:negative_binomial_loop}} to the program $C_{negncu}$:
	\begin{align*}
		 & \WHILE{x > 0}                                                                    \\
		 & \qquad \PCHOICE{\COMPOSE{\ASSIGN{x}{x-1}}{\ASSIGN{y}{y+x}}}{\sfrac{1}{2}}{\SKIP} \\
		 & \}~,
	\end{align*}
	with $x,y\in \Nats$.
	This program is positively almost--sure terminating and the expected number of loop iterations is $\iverson{x>0}\cdot 2 \cdot x$.
	The characteristic function for $f=y$ of the program is given by $\charwp{\guard}{C}{f}(X) = \iverson{x=0}\cdot y + \iverson{x>0}\cdot \frac{1}{2}\cdot \Bigl(X\subst{y}{y+x}\subst{x}{x-1} + X\Bigr)$.

	Intuitively, the value of $y$ after termination of the program should be $I=\iverson{x=0}\cdot y + \iverson{x>0} \cdot (y+\frac{x\cdot (x-1)}{2})$, i.e., the initial value of $y$ increases by the sum $\sum \limits_{j=0}^{x-1}j=\frac{x\cdot (x-1)}{2}$ if the loop can be executed at all.
	We will prove that this intuition is correct, so that $I$ is indeed the least fixed point of $\charwp{\guard}{C}{f}$.
	Again, $I$ harmonizes with $f$.

	First of all, we show that $I$ is a fixed point of $\charwp{\guard}{C}{f}$:

	\begin{align*}
		        & \charwp{\guard}{C}{f}(I)                                                                                                                                                                                                                  \\
		{} = {} & \iverson{x=0}\cdot y + \iverson{x > 0}\cdot \tfrac{1}{2} \cdot \Bigl(I\subst{y}{y+x}\subst{x}{x-1} +I\Bigr)                                                                                                                               \\
		{} = {} & \iverson{x=0}\cdot y + \iverson{x > 0}\cdot \tfrac{1}{2} \cdot \Bigl(\iverson{x-1=0}\cdot (y+x-1)                                                                                                                                         \\
		        & \qquad\qquad\qquad\qquad\qquad\qquad+ \iverson{x-1>0} \cdot (y+(x-1)+\frac{(x-1)\cdot (x-1-1)}{2}) +I\Bigr)                                                                                                                               \\
		{} = {} & \iverson{x=0}\cdot y + \iverson{x > 0}\cdot \tfrac{1}{2} \cdot \Bigl(\iverson{x=1}\cdot (y+x-1) + \iverson{x>1} \cdot (y+x-1+\frac{(x-1)\cdot (x-2)}{2}) +I\Bigr)                                                                         \\
		{} = {} & \iverson{x=0}\cdot y + \iverson{x > 0}\cdot \tfrac{1}{2} \cdot \Bigl(\iverson{x>0} \cdot \left(y+\frac{(x-1)\cdot x}{2}\right) +I\Bigr)                                                                                                   \\
		{} = {} & \iverson{x=0}\cdot y + \iverson{x > 0}\cdot \tfrac{1}{2} \cdot \Bigl(\iverson{x>0} \cdot \left(y+\frac{(x-1)\cdot x}{2}\right) +\iverson{x=0}\cdot y + \iverson{x>0} \cdot \left(y+\frac{(x-1)\cdot x}{2}\right)\Bigr)                    \\
		{} = {} & \iverson{x=0}\cdot y + \tfrac{1}{2} \cdot \Bigl(\iverson{x > 0}\cdot\iverson{x>0} \cdot \left(y+\frac{(x-1)\cdot x}{2}\right) +\iverson{x > 0}\cdot\iverson{x=0}\cdot y + \iverson{x>0} \cdot \left(y+\frac{(x-1)\cdot x}{2}\right)\Bigr) \\
		{} = {} & \iverson{x=0}\cdot y + \tfrac{1}{2} \cdot \Bigl(\iverson{x>0} \cdot \left(y+\frac{(x-1)\cdot x}{2}\right) +\iverson{x>0} \cdot \left(y+\frac{(x-1)\cdot x}{2}\right)\Bigr)                                                                \\
		{} = {} & \iverson{x=0}\cdot y + \iverson{x > 0}\cdot \iverson{x>0} \cdot \left(y+\frac{(x-1)\cdot x}{2}\right) = I
	\end{align*}
	So $I$ is indeed a fixed point of $\charwp{\guard}{C}{f}$ and $\charwp{\guard}{C}{f}(I)\pprec \infty$.
	To apply \textnormal{\cref{thm:optional_stopping_probabilistic_programs} (b)} we need to check that $I$ is conditionally difference bounded.
	\begin{align*}
		\Diff{I}(\State) {} = {} & \iverson{x>0}(\State)\cdot \tfrac{1}{2} \cdot \Bigl( \abs{I\subst{y}{y+x}\subst{x}{x-1}-I(\State)}(\State) + \abs{I-I(\State)}(\State) \Bigr)                                                                     \\
		{} = {}                  & \iverson{x>0}(\State)\cdot \tfrac{1}{2} \cdot \Bigl( \abs{I(\State\subst{y}{y+x}\subst{x}{x-1})-I(\State)} + \abs{I(\State)-I(\State)} \Bigr)                                                                     \\
		{} = {}                  & \iverson{x>0}(\State)\cdot \tfrac{1}{2} \cdot \Bigl( \abs{I(\State\subst{y}{y+x}\subst{x}{x-1})-I(\State)} \Bigr)                                                                                                 \\
		{} = {}                  & \iverson{x>0}(\State)\cdot \tfrac{1}{2} \cdot \Bigl( \Bigl \vert \iverson{x-1=0}(\State)\cdot (\State(y)+\State(x)-1)                                                                                             \\
		                         & \quad + \iverson{x-1>0}(\State) \cdot\left(\State(y)+(\State(x)-1)+\frac{(\State(x)-1)\cdot (\State(x)-1-1)}{2}\right)                                                                                            \\
		                         & \quad -\iverson{x=0}(\State)\cdot \State(y) - \iverson{x>0}(\State) \cdot \left(\State(y)+\frac{\State(x)\cdot (\State(x)-1)}{2}\right)\Bigr \vert \Bigr)                                                         \\
		{} = {}                  & \iverson{x>0}(\State)\cdot \tfrac{1}{2} \cdot \Bigl( \Bigl \vert \iverson{x=1}(\State)\cdot (\State(y)+\State(x)-1) + \iverson{x>1}(\State) \cdot \left(\State(y)+\frac{(\State(x)-1)\cdot (\State(x))}{2}\right) \\
		                         & \quad -\iverson{x=0}(\State)\cdot \State(y) - \iverson{x>0}(\State) \cdot \left(\State(y)+\frac{\State(x)\cdot (\State(x)-1)}{2}\right)\Bigr \vert \Bigr)                                                         \\
		{} = {}                  & \iverson{x>0}(\State)\cdot \tfrac{1}{2} \cdot \Bigl( \Bigl \vert \iverson{x>0}(\State) \cdot \left(\State(y)+\frac{(\State(x)-1)\cdot (\State(x))}{2}\right)                                                      \\
		                         & \quad -\iverson{x=0}(\State)\cdot \State(y) - \iverson{x>0}(\State) \cdot \left(\State(y)+\frac{\State(x)\cdot (\State(x)-1)}{2}\right)\Bigr \vert \Bigr)                                                         \\
		{} = {}                  & \iverson{x>0}(\State)\cdot \tfrac{1}{2} \cdot \Bigl( \abs{-\iverson{x=0}(\State)\cdot \State(y)} \Bigr)                                                                                                           \\
		{} = {}                  & \iverson{x>0}(\State)\cdot \tfrac{1}{2} \cdot \Bigl( \iverson{x=0}(\State)\cdot \State(y) \Bigr)                                                                                                                  \\
		{} = {}                  & \tfrac{1}{2} \cdot \Bigl( \iverson{x>0}(\State)\cdot \iverson{x=0}(\State)\cdot \State(y) \Bigr)                                                                                                                  \\
		{} = {}                  & 0
	\end{align*}
	So $I$ is indeed conditionally difference bounded.
	Hence, by \textnormal{\cref{thm:optional_stopping_probabilistic_programs} (b)} $I$ is a lower bound on the least fixed point of $\charwp{\guard}{C}{f}$.
	Since it is also a fixed point itself, we obtain
	\[
		I\eeq\wp{C_{negncu}}{f}.
	\]
\end{example}

\begin{example}[Probabilistic Doubling with Bounded Looping Time]
	Let us consider the program $C_{double}$
	\begin{align*}
		 & \WHILE{x > 0}                                                                          \\
		 & \qquad \COMPOSE{\ASSIGN{x}{x-1}}{\PCHOICE{\ASSIGN{y}{2 \cdot y}}{\sfrac{1}{2}}{\SKIP}} \\
		 & \}~,
	\end{align*}
	with $x,y\in \Nats$.
	The characteristic function for $f=y$ of the program is given by
	\[
		\charwp{\guard}{C}{f}(X) = \iverson{x \leq 0}\cdot y + \iverson{x > 0}\cdot \tfrac{1}{2}\cdot \Bigl( X\subst{x,y}{x-1,2\cdot y} + X\subst{x}{x-1}\Bigr).
	\]

	\noindent
	The looping time of this program is $\iverson{x>0}\cdot x$ which is bounded by $\max(\State(x),0) \in \Nats$ for any initial state $\State \in \States$.
	We will prove that $I= \iverson{x \leq 0}\cdot y + \iverson{x > 0}\cdot 2^{\frac{x}{2}}\cdot y$ is a lower bound on the expected value of $y$ after termination of the program.

	\begin{align*}
		\charwp{\guard}{C}{f}(I) {} = {} & \iverson{x \leq 0}\cdot y + \iverson{x > 0}\cdot \tfrac{1}{2}\cdot \Bigl( \iverson{x \leq 1}\cdot 2 \cdot y + \iverson{x > 1}\cdot 2^{\frac{x}{2} - \frac{1}{2}}\cdot 2 \cdot y + \iverson{x \leq 1}\cdot y + \iverson{x > 1}\cdot 2^{\frac{x}{2} - \frac{1}{2}}\cdot y\Bigr) \\
		{} = {}                          & \iverson{x \leq 0}\cdot y + \iverson{x > 0}\cdot \tfrac{1}{2}\cdot \Bigl( \iverson{x \leq 1}\cdot (2 \cdot y + y) + \iverson{x > 1}\cdot (2^{\frac{x}{2} + \frac{1}{2}}\cdot y + 2^{\frac{x}{2} - \frac{1}{2}}\cdot y)\Bigr)                                                  \\
		                                 & \succeq \iverson{x \leq 0}\cdot y + \iverson{x > 0}\cdot \Bigl( \iverson{x \leq 1}\cdot 2^{\frac{x}{2}} \cdot y + \iverson{x > 1}\cdot 2^{\frac{x}{2}}
		\cdot \tfrac{1}{2} \cdot \left(\sqrt{2} + \frac{1}{\sqrt{2}}\right) \cdot y\Bigr)                                                                                                                                                                                                                                \\
		                                 & \succeq \iverson{x \leq 0}\cdot y + \iverson{x > 0}\cdot \Bigl( \iverson{x \leq 1}\cdot 2^{\frac{x}{2}} \cdot y + \iverson{x > 1}\cdot 2^{\frac{x}{2}} \cdot y\Bigr) \tag{since $\sqrt{2} + \frac{1}{\sqrt{2}} \geq 2$}                                                       \\
		{} = {}                          & \iverson{x \leq 0}\cdot y + \iverson{x > 0} \cdot 2^{\frac{x}{2}} \cdot y                                                                                                                                                                                                     \\
		{} = {}                          & I
	\end{align*}

	\noindent
	So $I$ is indeed a subinvariant.
	Furthermore, we have $I \preceq \charwpn{\guard}{C}{f}n(I) \pprec \infty$ for any $n \in \Nats$, as the loop body is loop--free.
	So by using \textnormal{\cref{thm:optional_stopping_probabilistic_programs} (a)}
	we can deduce that $I$ is indeed a lower bound on the least fixed point of $\charwp{\guard}{C}{f}$, i.e.,

	\[
		I \ppreceq \wp{C_{double}}{f}.
	\]

	\noindent
	Note that in this case \textnormal{\cref{thm:optional_stopping_probabilistic_programs} (b)}
	is not applicable.
	$I$ harmonizes with $f$, but $I$ is not conditionally difference bounded:
	\begin{align*}
		         & \Diff{I}(\State)                                                                                                                                                                                                                                                               \\
		{} = {}  & \iverson{x > 0}(\State) \cdot \tfrac{1}{2} \cdot \left( \Bigl \vert I\subst{x,y}{x-1,2 \cdot y}-I(\State) \Bigr \vert(\State) + \Bigl \vert I\subst{x}{x-1}-I(\State) \Bigr \vert(\State)\right)                                                                               \\
		{}\geq{} & \iverson{x > 0}(\State) \cdot \tfrac{1}{2} \cdot \left( \Bigl \vert I\subst{x}{x-1}-I(\State) \Bigr \vert(\State)\right)                                                                                                                                                       \\
		{} = {}  & \iverson{x > 0}(\State) \cdot \tfrac{1}{2} \cdot \Bigl( \Bigl \vert \iverson{x \leq 1}(\State)\cdot \State(y) + \iverson{x > 1}(\State)\cdot 2^{\frac{\State(x)-1}{2}}\cdot \State(y)                                                                                          \\
		         & \qquad\qquad\qquad\qquad-\iverson{x = 0}(\State)\cdot \State(y) - \iverson{x > 0}(\State)\cdot 2^{\frac{\State(x)}{2}}\cdot \State(y) \Bigr \vert\Bigr)                                                                                                                        \\
		%
		%
		%
		{} = {}  & \iverson{x > 0}(\State) \cdot \tfrac{1}{2} \cdot \left( \Bigl \vert \iverson{x \leq 1}(\State)\cdot \State(y) + \iverson{x > 1}(\State)\cdot 2^{\frac{\State(x)-1}{2}}\cdot \State(y) - \iverson{x > 0}(\State)\cdot 2^{\frac{\State(x)}{2}}\cdot \State(y) \Bigr \vert\right) \\
		{} = {}  & \iverson{x > 0}(\State) \cdot \tfrac{1}{2} \cdot \Bigl( \Bigl \vert \iverson{x = 1}(\State)\cdot \State(y) + \iverson{x > 1}(\State)\cdot 2^{\frac{\State(x)-1}{2}}\cdot \State(y) - \iverson{x = 1}(\State)\cdot 2^{\frac{\State(x)}{2}}\cdot \State(y)                       \\
		         & \qquad\qquad\qquad\qquad - \iverson{x > 1}\cdot 2^{\frac{\State(x)}{2}}\cdot \State(y) \Bigr \vert\Bigr)                                                                                                                                                                       \\
		{} = {}  & \iverson{x > 0}(\State) \cdot \tfrac{1}{2} \cdot \left( \Bigl \vert \iverson{x = 1}(\State)\cdot \State(y) (1 - \sqrt{2}) + \iverson{x > 1}(\State) \cdot 2^{\frac{\State(x)}{2}}\cdot \State(y) \left( \frac{1}{\sqrt{2}} - 1 \right)\Bigr\vert\right)                        \\
		{} = {}  & \iverson{x > 0}(\State) \cdot \tfrac{1}{2} \cdot \left( \iverson{x = 1}(\State)\cdot \State(y) (\sqrt{2} - 1) + \iverson{x > 1}(\State) \cdot 2^{\frac{\State(x)}{2}}\cdot \State(y) \left( 1 - \frac{1}{\sqrt{2}} \right)\right)
	\end{align*}
	which is unbounded.
\end{example}


\noindent
Nevertheless, even in the case of finite expected looping time, \cref{thm:optional_stopping_probabilistic_programs} just provides \emph{sufficient}
conditions for lower bounds.
This is not surprising as conditional difference boundedness is a sufficient condition for uniform integrability but far from being necessary.
The following example presents a limitation of our proof rule.

\begin{example}[Probabilistic Doubling with Unbounded Looping Time]
	Let us consider the program
	\begin{align*}
		 & \WHILE{a = 1}                                                                   \\
		 & \qquad \COMPOSE{\PCHOICE{\ASSIGN{a}{0}}{\sfrac{1}{2}}{\ASSIGN{b}{2 \cdot b}}}{} \\
		 & \}~,
	\end{align*}
	with $b\in \Nats$ and $b>0$.
	The characteristic function for $f=b$ of the program is given by $\charwp{\guard}{C}{f}(X) = \iverson{a\neq 1}\cdot b + \iverson{a=1}\cdot \tfrac{1}{2}\cdot \Bigl( X\subst{a}{0} + X\subst{b}{2\cdot b}\Bigr)$.

	Now consider the sequence of invariants $I_n \coloneqq \iverson{a\neq 1}\cdot b + \iverson{a=1}\cdot \frac{n \cdot b}{2}$.
	Then $I_n = \charwpn{\guard}{C}{f}{n+1}(0)$: $I_0 = \iverson{a \neq 1} \cdot b = \charwp{\guard}{C}{f}(0)$.
	Furthermore, $\charwp{\guard}{C}{f}(I_n)=\iverson{a\neq 1}\cdot b + \iverson{a=1}\cdot \tfrac{1}{2}\cdot \Bigl(I_n\subst{a}{0} + I_n\subst{b}{2\cdot b}\Bigr) = \iverson{a\neq 1}\cdot b + \iverson{a=1}\cdot \tfrac{1}{2}\cdot \Bigl(b + n \cdot b\Bigr) = I_{n+1}$.
	Hence, $I_0 \preceq I_1 \preceq \cdots $ and $\iverson{a=0}\cdot b + \iverson{a=1}\cdot \infty = \lim_{n \to \omega} I_n=\lfp \charwp{\guard}{C}{f}$ by the \textnormal{Tarski--Kantorovich Principle \cref{thm:tarski-kantorovich}}.
	Moreover, each of the $I_n$ is a lower bound on the least fixed point, i.e., $I_n \preceq \lfp \charwp{\guard}{C}{f}$.

	Furthermore, the loop is expected to be executed twice.
	Clearly, \textnormal{\cref{thm:optional_stopping_probabilistic_programs}
		(a)} and \textnormal{(c)} are not applicable.
	Let $n \geq 2$.
	Then \textnormal{\cref{thm:optional_stopping_probabilistic_programs}
		(b)} is not applicable either, because $\Diff{I_n}$ is unbounded, i.e., $I_n$ is not conditionally difference bounded.
	To see this, let $\State \in \States$.
	\begin{align*}
		  & \Diff{I_n}(\State)                                                                                                                                                                                                                                                \\
		= & \left(\iverson{a = 1}\cdot \tfrac{1}{2} \cdot \Bigl(\abs{I_n -I_n(\State)}\subst{a}{0} + \abs{I_n -I_n(\State)}\subst{b}{2\cdot b}\Bigr)\right)(\State)                                                                                                           \\
		= & \iverson{a = 1}(\State)\cdot \tfrac{1}{2} \cdot \Bigl(\abs{\iverson{0\neq 1}(\State)\cdot \State(b) + \iverson{0=1}(\State)\cdot \frac{n \cdot \State(b)}{2}-(\iverson{a\neq 1}(\State)\cdot \State(b) + \iverson{a=1}(\State)\cdot \frac{n \cdot \State(b)}{2})} \\
		+ & \abs{\iverson{a\neq 1}(\State)\cdot 2\cdot \State(b) + \iverson{a=1}(\State)\cdot \frac{n \cdot 2 \cdot \State(b)}{2}-(\iverson{a\neq 1}(\State)\cdot \State(b) + \iverson{a=1}(\State)\cdot \frac{n \cdot \State(b)}{2})}\Bigr)                                  \\
		= & \iverson{a = 1}(\State)\cdot \tfrac{1}{2} \cdot \Bigl(\abs{\State(b)-(\iverson{a\neq 1}(\State)\cdot \State(b) + \iverson{a=1}(\State)\cdot \frac{n \cdot \State(b)}{2})}                                                                                         \\
		+ & \abs{\iverson{a\neq 1}(\State)\cdot 2\cdot \State(b) + \iverson{a=1}(\State)\cdot \frac{n \cdot 2 \cdot \State(b)}{2}-(\iverson{a\neq 1}(\State)\cdot \State(b) + \iverson{a=1}(\State)\cdot \frac{n \cdot \State(b)}{2})}\Bigr)                                  \\
		= & \iverson{a = 1}(\State)\cdot \tfrac{1}{2} \cdot \Bigl(\abs{\State(b) - \frac{n \cdot \State(b)}{2}} + \abs{n \cdot \State(b) - \frac{n \cdot \State(b)}{2}} \Bigr)                                                                                                \\
		= & \iverson{a = 1}(\State)\cdot \tfrac{1}{2} \cdot \Bigl(\abs{\frac{(2-n) \cdot \State(b)}{2}} + \abs{\frac{n \cdot \State(b)}{2}} \Bigr)                                                                                                                            \\
		= & \iverson{a = 1}(\State)\cdot \tfrac{1}{2} \cdot \frac{(2\cdot n -2) \cdot \State(b)}{2}                                                                                                                                                                           \\
		= & \iverson{a = 1}(\State)\cdot \tfrac{1}{2} \cdot \frac{(n -1) \cdot \State(b)}{2}~,
	\end{align*}
	which can take an arbitrary large value as there is no bound on the value of $b>0$.
	Hence \textnormal{\cref{thm:optional_stopping_probabilistic_programs} (b)} \emph{cannot}
	be applied although $I_n$ is a lower bound.
	However, in \cite{DBLP:conf/esop/KaminskiKMO16} it is proved that the $I_n$ form an $\omega$--subinvariant, hence, they are all lower bounds.
\end{example}

\renewcommand{\charwp}[3]{\charfun{\wpsymbol}{#1}{#2}{#3}}
\renewcommand{\charwpn}[4]{\charfunn{\wpsymbol}{#1}{#2}{#3}{#4}}

\section{Details on Probability Theory}
\label{app:prelim_probability}
\noindent
This section is devoted to a more detailed introduction of the concepts from probability theory that we use in our work.
\subsection{$\sigma$--Fields}
\noindent
When setting up a probability space over some sample space $\Omega$, which can be any set, we have to distinguish the sets whose probabilities we want to be able to measure.
The collection of these measurable sets is called a $\sigma$--field.
\begin{definition}[$\sigma$--Field]
	Let $\Omega$ be an arbitrary set and $\F{F}\subset Pot(\Omega)$.
	$\F{F}$ is called a \emph{$\sigma$--field over $\Omega$} if the following three conditions are satisfied.
	\begin{enumerate}
		\item $\Omega \in \F{F}$,
		\item $A \in \F{F} \Rightarrow \Omega \setminus A \in \F{F}$ i.e., $\F{F}$ is closed under taking the complement,
		\item $A_i \in \F{F} \Rightarrow \bigcup\limits_{i \in \Nats}A_i \in \F{F}$ i.e., $\F{F}$ is closed under countable union.
	\end{enumerate}
	The pair $(\Omega, \F{F})$ is called a \emph{measurable space}.
	The elements of $\F{F}$ are called \emph{measurable} sets.
\end{definition}

\noindent
In the setting of program verification, we have seen that $\Omega$ is the set of all program runs.
A \emph{program run} is an infinite sequence of states, i.e., variable assignments.
We regard $\sigma$--fields $\F{F} \subseteq Pot(\Omega)$ of the form $\sigmagen{\F{E}}$, where $\F{E}$ is a collection of cylinder sets.
(More precisely, we regard fields $\F{F}_n \subseteq Pot(\Omega)$, where $\F{F}_n$ is the smallest $\sigma$--field containing all cylinder sets of order $n$ or smaller.) In our setting, a set of runs $\F{E}$ is a \emph{cylinder set} of order $n$ if all runs in $\F{E}$ have the same $n+1$ first configurations, and all the following configurations can be arbitrary.

Let $\F{F}$ and $\F{B}$ be $\sigma$--fields over $\Omega$.
Then $\F{F} \cap \F{B}$ is a $\sigma$--field over $\Omega$.
Furthermore, if $(\F{F}_i)_{i \in I}$ is a family of $\sigma$--fields over $\Omega$ then so is $\bigcap\limits_{i \in I}\F{F}_i$.

For any set $\F{E}$ of subsets of $\Omega$, let $\sigmagen{\F{E}} \subseteq \F{F}$ consist of all elements that are contained in all $\sigma$--fields that are supersets of $\F{E}$.
The mapping from $\F{E}$ to $\sigmagen{\F{E}}$ is also called \emph{$\sigma$--operator}.

\begin{definition}[Generating $\sigma$--Fields]
	Let $\F{E}\subset Pot(\Omega)$.
	Then the smallest $\sigma$--field over $\Omega$ containing $\F{E}$ is
	\[
		\sigmagen{\F{E}} \coloneqq \bigcap\limits_{\F{E}\subset\F{F},~\F{F}~\sigma\textnormal{--field over } \Omega} \F{F}~.
	\]
\end{definition}

\noindent
It turns out that there is a special case in which the generated $\sigma$--field is easy to describe, namely in the case where a countable covering of the space $\Omega$ is given.

\begin{lemma}[Generating $\sigma$--Fields for Covering of $\Omega$]
	\label{lemma:sigma_field_countable_cover}
	If $\Omega = \biguplus\limits_{i=1}^{\infty} A_i$ for a sequence $A_i \in Pot(\Omega)$ and $\F{H} \coloneqq \sigmagen{\{A_i \mid i \in \Nats\}}$ then
	\[
		\F{H}=\underbrace{\left\{\biguplus\limits_{i \in J} A_i \mid J \subset \Nats\right\}}_{ \eqqcolon \F{E}}~.
	\]
\end{lemma}

\begin{proof}
	Showing that $\F{E}$ is a $\sigma$--field is enough to prove the desired result: it contains all the sets $A_i$ and every $\sigma$--algebra containing all the $A_i$ has to contain all their countable unions, i.e., $\F{H}$:
	\begin{item}
	      \item $\Omega \in \F{H}$ by choosing $J=\Nats$.
	      \item Let $J\subset \Nats$.
	      Then $\Omega\setminus\left(\biguplus\limits_{i \in J}A_i\right) = \biguplus\limits_{i \in \Nats \setminus J}A_i \in \F{H}.$
	      \item Let $J_n \subset \Nats$.
	      Then $\bigcup\limits_{n \in \Nats} \biguplus\limits_{i \in J_n} A_i = \biguplus\limits_{i \in \bigcup\limits_{n \in \Nats}J_n} A_i\in \F{H}.$
	\end{item}
\end{proof}

\noindent
This special type of a $\sigma$--field can be used to describe the elements of the $\sigma$--fields $\F{F}_n$ (cf.
\cref{def:canonical-filtration}) and $\F{G}_n$ (cf.
\cref{def:filtration}).
We will discuss it in more detail in Appendix \ref{app:proofs_expectations_processes} where we use it in the proofs.
\begin{definition}[Borel--Field]
	If $\Omega = \PosRealsInf$ we use its $\sigma$--field $\F{B}=\F{B}\left(\PosRealsInf\right)$, the \emph{Borel--field} with
	\[
		\F{B}\left(\PosRealsInf\right) \coloneqq \sigmagen{\setcomp{(a,b)\subset \PosRealsInf}{a < b \in \PosRealsInf}}~.
	\]
\end{definition}

\noindent
In this work, we use the concept of measurable maps (or ``measurable mappings'').
Measurable maps are the structure--preserving maps between measurable spaces.
They are defined as follows.

\begin{definition}[Measurable Map]
	\label{defQQQmeasurable_map}
	Let $(\Omega_1,\F{F}_1)$ and $(\Omega_2,\F{F}_2)$ be measurable spaces.
	A function $f\colon \Omega_1 \to \Omega_2$ is called an \emph{$\F{F}_1$--$\F{F}_2$ measurable map} or just \emph{measurable} if for all $A \in \F{F}_2$
	\[
		f^{-1}(A) = \{\run \in \Omega_1 \mid f(\run) \in A\} \in \F{F}_1~.
	\]
	$\sigmagen{f} \coloneqq f^{-1}(\F{F}_2) \coloneqq \{f^{-1}(A) \mid A \in \F{F}_2\}$ is the smallest $\sigma$--field $\F{F}$ such that $f$ is $\F{F}_1$--$\F{F}_2$ measurable.
	Similarly, $\sigmagen{f_0,\ldots,f_n} = \sigmagen{f_0^{-1}(\F{F}_2) \cup \ldots \cup f_n^{-1}(\F{F}_2)}$.
	This will become important when talking about random variables and conditional expected value.
\end{definition}

\bigskip

\subsection{Probability Spaces}

\noindent
So far we have only introduced the concept of measurable spaces.
Intuitively, a measurable space provides the structure for defining a \emph{measure}.
Probability spaces are measurable space attached with a certain measure where the measure of the sample space is $1$.

\begin{definition}[Probability Measure, Probability Space]
	Let $(\Omega, \F{F})$ be a measurable space.
	A map $\mu \colon \F{F} \to \IR_{\geq 0}$ is called a \emph{measure} if
	\begin{enumerate}
		\item $\mu(\emptyset)=0$
		\item $\mu(\biguplus\limits_{i\geq 0} A_i) = \sum\limits_{i\geq 0}\mu(A_i)$~.
	\end{enumerate}
	A \emph{probability measure} is a measure $\mathbb{P} \colon \F{F} \to \IR_{\geq 0}$ with $\mathbb{P}(\Omega)=1$.
	This implies that $\mathbb{P}(A)\in [0,1]$ for every $A \in \F{F}$.
	If $\mathbb{P}$ is a probability measure, then $(\Omega, \F{F}, \mathbb{P})$ is called a \emph{probability space}.
	In this setting a set in $\F{F}$ is called an \emph{event}.
\end{definition}

\noindent
Here, the intuition for a probability measure $\mathbb{P}$ is that for any set $A \in \F{F}$, $\mathbb{P}(A)$ is the probability that an element chosen from $\Omega$ is contained in $A$.

In this work we will consider properties that hold almost--surely, for example almost--sure termination or almost--sure convergence.

\begin{definition}[Almost--Sure Properties]
	Let $(\Omega, \F{F}, \mathbb{P})$ be a probability space and $\alpha$ some property, e.g., a logical formula.
	If $A_\alpha \coloneqq \{\run \in \Omega \mid \run \vDash \alpha\}\in \F{F}$ (i.e., it is measurable) and $\IP{}{A_\alpha}=1$ then $\alpha$ is said to hold \emph{almost--surely}.
\end{definition}

\noindent
If a property $\alpha$ holds almost--surely it does not need to hold for all $\run \in \Omega$.
However, the measure $\mathbb{P}$ cannot distinguish $A_\alpha=\Omega$ and $A_\alpha\neq \Omega$ if $\IP{}{A_\alpha}=1$, so in the sense of $\mathbb{P}$, almost--surely holding properties can be considered as holding globally.
\bigskip

\subsection{Integrals of Arbitrary Measures}
\label{sec:integral_arbitrary_measure}
\noindent
We now introduce a notion of an integral with respect to an arbitrary measure.
Therefore, we fix a measurable space $(\Omega, \F{F})$ and a measure $\mu$.
The objective is to define a ``mean'' of a measurable function $f$.
The basic idea is to partition the image of $f$ into sets $A_i$ on which $f$ has a constant value $\alpha_i$.
Then we compute the weighted average of the $\alpha_i$, where the weights are the measures of the $A_i$.
This definition is fine if $f$ takes only finitely many values (these functions are called \emph{elementary}).
If $f$ takes infinitely (countable or even uncountable) many values, we have to approximate $f$ step by step by such functions with finite image.
This yields a limit process.
In this work we will only consider the cases where $\mu=\mathbb{P}$ is a probability measure or a probability submeasure (i.e., $\mathbb{P}(\Omega)\leq 1$).

\begin{definition}[Elementary Function, \protect{\cite[Def.~2.2.1]{bauer71measure}}]
	An \emph{elementary function} is a nonnegative measurable function $f\colon \Omega \to \PosRealsInf$ that takes only finitely many finite values, i.e., there exist $A_1, \dots, A_n \in \F{F}$ and $\alpha_1, \dots, \alpha_n \geq 0$ such that
	\[
		f= \sum \limits_{i=1}^n \alpha_i \cdot \indicator{A_i}~.
	\]
\end{definition}

\begin{lemma}[Decomposition of Elementary Functions \protect{\cite[Lemma 2.2.2]{bauer71measure}}]
	\label{lemma:elementary_functions_decompositions}
	Let $f=\sum \limits_{i=1}^n \alpha_i \cdot \indicator{A_i}=\sum \limits_{j=1}^n \beta_j \cdot \indicator{B_j}$ be an elementary function, where $\alpha_i, \beta_j \geq 0$, $A_i, B_j \in \F{F}$ for all $i$ and $j$.
	Then
	\[
		\sum \limits_{i=1}^n \alpha_i \cdot \mu\left(A_i\right)=\sum \limits_{j=1}^n \beta_j \cdot \mu\left(B_j\right)~.
	\]
\end{lemma}

\noindent
Given an elementary function, the decomposition into a linear combination of indicator functions is not unique.
But to define an integral we have to guarantee that its value does not depend on the chosen decomposition.
Fortunately, this can be proved:

\begin{definition}[Integral of Elementary Function, \protect{\cite[Def.~2.2.3]{bauer71measure}}]
	\label{def:integral_elementary}
	Let $f= \sum \limits_{i=1}^n \alpha_i \cdot \indicator{A_i}$ be an elementary function.
	Then we define its integral w.r.t.~$\mu$ as
	\[
		\int_{\Omega} f\, d\mu \coloneqq \int f \, d\mu \coloneqq \sum \limits_{i=1}^n \alpha_i \cdot \mu\left(A_i\right).
	\]
	The well--definedness is justified by \textnormal{\cref{lemma:elementary_functions_decompositions}} as it shows the independence of the chosen decomposition of the elementary function.
\end{definition}

\noindent
However, the measurable functions we use are not elementary.
They take arbitrary (countably or even uncountably) many values.
So we have to generalize \cref{def:integral_elementary}.
It can be shown that any nonnegative measurable function is the limit of a monotonic increasing sequence of elementary functions.

\begin{theorem}[Representation by Elementary Functions]
	Let $f\colon \Omega \to \PosRealsInf$ be a nonnegative measurable function.
	Then there exists a monotonic sequence $f_0 \leq f_1 \leq \dots$ of elementary functions such that
	\[
		f = \sup_{n \in \Nats} f_n = \lim \limits_{n \to \omega} f_n~.
	\]
	Furthermore, for any two such sequences $(f_n)_{n \in \Nats}$ and $(g_n)_{n \in \Nats}$ we have $\sup\limits_{n \in \Nats} \int f_n \, d \mu=\sup\limits_{n \in \Nats} \int g_n \, d \mu$.
\end{theorem}

\begin{proof}
	See \cite[Cor.
		2.3.2., Thm.
		2.3.6]{bauer71measure}.
\end{proof}

\noindent
This theorem justifies the following definition of an integral for an arbitrary nonnegative function.

\begin{definition}[Integral of Arbitrary Functions]
	Let $f\colon \Omega \to \PosRealsInf$ be a nonnegative measurable function and $(f_n)_{n \in \Nats}$ a monotonic sequence of elementary functions such that~$f= \sup_{n \in \Nats} f_n$.
	Then we define the integral of $f$ w.r.t.~$\mu$ by
	\[
		\int_{\Omega} f \, d\mu \coloneqq \int f \, d\mu \coloneqq \sup_{n \in \Nats} \int f_n \, d\mu~.
	\]
\end{definition}

\noindent
Before we state the properties of the integral used in this work we will define the integral on a measurable subset of $\Omega$.

\begin{definition}[Integral on Measurable Subset]
	Let $f\colon \Omega \to \PosRealsInf$ be a nonnegative measurable function and $A\in \F{F}$.
	Then $f \cdot \indicator{A}$ is nonnegative and measurable, and we define
	\[
		\int_{A} f \, d\mu \coloneqq \int f \cdot \indicator{A} \, d\mu~.
	\]
\end{definition}

\noindent
With this definition of an integral, a very special property holds for monotonically increasing sequences of nonnegative functions: taking the limit (it always exists due to monotonicity) and the integral can be intertwined.
We will focus on this when discussing uniform integrability.

\begin{theorem}[Monotone Convergence Theorem, \protect{\cite[Thm.~2.3.4.]{bauer71measure}}]
	\label{app_thm:monotone_convergence_general_integral}
	Let $(f_n)_{n \in \Nats}$ be a monotonic sequence of nonnegative measurable functions, i.e., $f_n \colon \Omega \to \PosRealsInf$ is measurable and $f_0 \leq f_1 \leq \dots$.
	Then $\sup_{n\in \Nats} f_n = \lim\limits_{n \to \omega} f_n \colon \Omega \to \PosRealsInf$ is measurable and
	\[
		\lim_{n \to \omega} \int f_n \, d \mu = \int\lim_{n \to \omega} f_n \, d \mu~.
	\]
\end{theorem}

\begin{lemma}[Properties of the Integral]
	\label{lemma:properties_of_integral}
	Let $a,b \geq 0$, $f,g \colon \Omega \to \PosRealsInf$ be measurable functions, and $A,A_i \in \F{F}$.
	Then
	\begin{align*}
		\int_{A} a\cdot f + b \cdot g \, d \mu               & = a\cdot\int_{A} f \, d \mu + b \cdot \int_{A} g \, d \mu~. & \text{(Linearity)}  \\
		\int_{A} 1 \, d \mu                                  & = \mu(A)~.                                                  & \text{(Measure)}    \\
		\int_{\biguplus\limits_{i \in \Nats} A_i} f \, d \mu & = \sum\limits_{i \in \Nats}\int_{ A_i} f \, d \mu~.         & \text{(Additivity)} \\
	\end{align*}
\end{lemma}
\begin{proof}
	See \cite[(2.2.4.), (2.3.6.), (2.3.7.), Cor.~2.3.5.]{bauer71measure}.
\end{proof}

\subsection{Random Variables}

\noindent
A \emph{random variable} $X$ maps elements of one set $\Omega$ to another set $\Omega'$.
If $\mathbb{P}$ is a probability measure for $\Omega$ (i.e., for $A\subseteq \Omega$, $\mathbb{P}(A)$ is the probability that an element chosen from $\Omega$ is contained in $A$), then one obtains a corresponding probability measure $\mathbb{P}^X$ for $\Omega'$.
For $A' \subseteq \Omega'$, $\mathbb{P}^X(A')$ is the probability that an element chosen from $\Omega$ is mapped by $X$ to an element contained in $A'$.
In other words, instead of regarding the probabilities for choosing elements from $A$, one now regards the probabilities for the values of the random variable $X$.

\begin{definition}[Random Variable]
	\label{defQQQrandom_variable}
	Let $(\Omega, \F{F}, \mathbb{P})$ be a probability space.
	An $\F{F}$--$\F{B}(\PosRealsInf)$ measurable map $X \colon \Omega \to \PosRealsInf$ is a random variable.
	Instead of saying ``$\F{F}$--$\F{B}\left(\PosRealsInf\right)$ measurable'' we simply use the notion ``$\F{F}$--measurable''.
	It is called \emph{discrete random variable}, if its image is a \emph{countable} set.

	$\mathbb{P}^X \colon \primed{\F{F}} \to [0,1], A' \mapsto \mathbb{P}(X^{-1}(A'))$ is the induced probability measure by $X$ on $(\PosRealsInf,\F{B}(\PosRealsInf))$.
	Instead of $\mathbb{P}^X(A)$ the notation $\mathbb{P}(X \in A)$ is common.
	If $A = \{i\}$ is a singleton set, we also write $\mathbb{P}(X = i)$ instead of $\mathbb{P}^X(\{i\})$.
\end{definition}

\begin{definition}[Expected Value]
	Let $X \colon \Omega \to \PosRealsInf$ be a random variable.
	Then $\expec{}{X} \coloneqq \int X d\mathbb{P}$.
\end{definition}

\begin{lemma}[Expected Value as Sum]
	\label{alternative expected value}
	If $X$ is a discrete random variable we have $\expec{}{X}=\sum\limits_{r \in \PosRealsInf} r \cdot \IP{}{X=r}$.
	Note that this series has only countably many nonzero nonnegative summands.
	Hence, it either converges or it diverges to infinity.
\end{lemma}

\begin{proof}
	Let $X(\Omega)=\{r_1,r_2,\dots\}$.
	Then $\Omega = \biguplus\limits_{i \in \Nats} X^{-1}(\{r_i\})$.
	Hence
	\begin{align*}
		\IE{}{X} & =\int X d\mathbb{P}                                                                                                                     \\
		         & =\int_{\biguplus\limits_{i \in \Nats} X^{-1}(\{r_i\})} X \, d\mathbb{P}                                                                 \\
		         & =\sum\limits_{i \in \Nats}\int_{X^{-1}(\{r_i\})} X \, d\mathbb{P}          & \text{by \cref{lemma:properties_of_integral} (Additivity)} \\
		         & = \sum\limits_{i \in \Nats}\int_{ X^{-1}(\{r_i\})} r_i \, d\mathbb{P}      & X \text{ is constant on $X^{-1}(\{r_i\})$}                 \\
		         & =\sum\limits_{i \in \Nats}r_i\cdot \int_{X^{-1}(\{r_i\})} 1 \, d\mathbb{P} & \text{by \cref{lemma:properties_of_integral} (Linearity)}  \\
		         & =\sum\limits_{i \in \Nats}r_i\cdot \IP{}{X^{-1}(\{r_i\})}                  & \text{by \cref{lemma:properties_of_integral} (Measure)}    \\
		         & =\sum\limits_{i \in \Nats}r_i\cdot \IP{}{X = r_i}                                                                                       \\
		         & =\sum\limits_{r \in \PosRealsInf} r \cdot \IP{}{X = r}~.
	\end{align*}
\end{proof}

\subsection{Uniform Integrability}

\noindent
Now given any stochastic process (i.e., a sequence of random variables, $X_n:\Omega \to \PosRealsInf$ on a probability space $(\Omega,\F{F},\mathbb{P})$) that has an almost--surely existing limit the question arises whether we can construct the expectation of the limit as the limit of the expectations of the $X_n$.
However, this is false in general, a counterexample is given in \cite[Introduction of 7.10]{grimmett2001probability}.
Therefore, we distinguish stochastic processes with this special property.

\begin{definition}
	[Uniform Integrability, \protect{\cite[Thm.
				7.10.(3)]{grimmett2001probability}}] Let $(X_n)_{n \in \Nats}$ be a sequence of random variables converging almost--surely to a random variable $X$.
	Then $(X_n)_{n \in \Nats}$ is uniformly integrable if and only if $\lim\limits_{n \to \omega} \expec{}{X_n}=\expec{}{X}$.
\end{definition}

\noindent
To check for uniform integrability is one of the main purposes of this work.
There are two sufficient criteria which we will list below.
The first one is the monotonic convergence theorem for random variables, a corollary of \cref{app_thm:monotone_convergence_general_integral}.
It states that a sequence of monotonically increasing variables is always uniformly integrable.

\begin{corollary}[Monotone Convergence Theorem for Random Variables, \protect{\cite[Thm.~2.3.4.]{bauer71measure}}]
	\label{app_thm:monotone_convergence}
	Let $(X_n)_{n \in \Nats}$ be a monotonic sequence of nonnegative random variables, i.e., $X_n \colon \Omega \to \PosRealsInf$ is measurable and $X_0 \leq X_1 \leq \dots$.
	Then $\lim_{n \to \omega} X_n \colon \Omega \to \PosRealsInf$ is measurable and
	\[
		\lim_{n \to \omega} \expec{}{X_n} = \expec{}{\lim_{n \to \omega} X_n }.
	\]
\end{corollary}

\noindent
The second sufficient criterion states that if the sequence is bounded by an integrable random variable $M$, then uniform integrability is given as well.

\begin{lemma}
	[Bounded Stochastic Processes are Uniformly Integrable \protect{\cite[Thm.
				7.10.(4)]{grimmett2001probability}}] \label{lem:upper_bound_uniform_integrability}
	Let $(X_n)_{n \in \Nats}$ be a sequence of random variables and $M$ a nonnegative random variable on a probability space $(\Omega, \F{F}, \mathbb{P})$ with $X_n\leq M$ for all $n \in \Nats$.
	If $\expec{}{M}<\infty$ then $(X_n)_{n \in \Nats}$ is uniformly integrable.
\end{lemma}

\subsection{Conditional Expected Values}

\noindent
We introduce the notion of conditional expected value w.r.t.~a sub--$\sigma$--field on a fixed probability space $(\Omega, \F{F},\mathbb{P})$.
The idea is that given a random variable $X$ and a subfield $\F{G} \subset \F{F}$ we would like to approximate $X$ by another $\F{G}$--measurable random variable w.r.t.~expectation.
Intuitively, this means that we want to construct a (possibly infinite) nonnegative linear combination of the functions $\indicator{G}, G \in \F{G}$, in such a way that restricted to a set $G \in \F{G}$, the random variable $X$ and this linear combination have the same average value w.r.t.~$\mathbb{P}$.

\begin{definition}[Conditional Expected Value w.r.t.\ a $\sigma$--Field, \protect{\cite[Def.~10.1.2]{bauer71measure}}]
	\label{defQQQconditional_expectation}
	Let $X \colon \Omega \to \PosRealsInf$ be a random variable and $\F{G}$ be a sub--$\sigma$--field of $\F{F}$.
	A random variable $Y \colon \Omega \to \PosRealsInf$ is called a \emph{conditional expected value} of $X$ w.r.t.~$\F{G}$ if
	\begin{enumerate}
		\item $Y$ is $\F{G}$--measurable
		\item $\int_G Y \, d \mathbb{P}=\expec{}{Y \cdot \indicator{G}}= \expec{}{X \cdot \indicator{G}}=\int_G X \, d \mathbb{P}$ for every $G \in \F{G}$.
	\end{enumerate}
	If two such random variables $Y$ and $\primed{Y}$ exist, the just stated properties already ensure that $\mathbb{P}(Y=\primed{Y})=1$.
	Therefore a conditional expected value is almost--surely unique which justifies the notation $\expec{}{X \mid \F{G}} \coloneqq Y$.
\end{definition}

\noindent
If the sub--$\sigma$--field has the special structure described in \cref{lemma:sigma_field_countable_cover} then the just stated property just needs to be checked on the generators.

\begin{lemma}[Conditional Expected Value]
	\label{app_lemma:conditional_expectation_countable_cover}
	Let $X$ be a random variable.
	If $\Omega = \biguplus\limits_{i=1}^{\infty} A_i$ for a sequence $A_i \in \F{F}$ and $\F{G} \coloneqq \sigmagen{\{A_i \mid i \in \Nats\}}$ then a $\F{G}$--measurable function $Y$ is a conditional expected value of $X$ iff
	\[
		\expec{}{X\cdot\indicator{A_i}}=\expec{}{Y\cdot\indicator{A_i}}, \text{ for all } i \in \Nats.
	\]
\end{lemma}

\begin{proof}
	By definition it is left to show that $\expec{}{X\cdot\indicator{A_i}}=\expec{}{Y\cdot\indicator{A_i}}$ for all $i \in \Nats$ implies $\expec{}{X\cdot\indicator{G}}=\expec{}{Y\cdot\indicator{G}}$ for any $G \in \F{G}$.
	But due to \cref{lemma:sigma_field_countable_cover} it is enough to show this for any disjoint union $\biguplus\limits_{i \in J}A_i$.
	We have $\indicator{\biguplus\limits_{i \in J}A_i}=\sum\limits_{i \in J}\indicator{A_i}$ and this series always converges point--wise as at most one of the summands is nonzero for any $\alpha \in \Omega$.
	Hence we have
	\begin{align*}
		\expec{}{X\cdot\indicator{\biguplus\limits_{i \in J}A_i}} & =\expec{}{X\cdot\sum\limits_{i \in J}\indicator{A_i}}                                                             \\
		                                                          & \overset{\textnormal{\cref{lemma:properties_of_integral}}}{=}\sum\limits_{i \in J}\expec{}{X\cdot\indicator{A_i}} \\
		                                                          & =\sum\limits_{i \in J}\expec{}{Y\cdot\indicator{A_i}}=\expec{}{Y\cdot\indicator{\biguplus\limits_{i \in J}A_i}}.
	\end{align*}
\end{proof}

It turns out that in our setting a conditional expectation always exists as it is almost surely finite.
\begin{theorem}
	[Existence of Conditional Expected Values \protect{\cite[Prop.
				3.1.]{DBLP:journals/corr/abs-1709-04037}}] \label{thm:existence_cond_exp_nonnegative}
	Let $X \colon \Omega \to \PosRealsInf$ be a random variable such that~$\IP{}{X=\infty}=0$ and let $\F{G}$ be a sub--$\sigma$--field of $\F{F}$.
	Then $\expec{}{X \mid \F{G}}$ exists.
\end{theorem}

This theorem helps us to find bounds on the conditional expected value if we are unable to determine it exactly.

\begin{lemma}
	\label{lem:upper_bound_cond_exp}
	Let $X$ be a random variable with $\IP{}{X=\infty}=0$ and $\Omega = \biguplus\limits_{i=1}^{\infty} A_i$ for a sequence $A_i \in \F{F}$, $\F{G} \coloneqq \sigmagen{\{A_i \mid i \in \Nats\}}$ and $Y$ a $\F{G}$--measurable function.
	We have $\expec{}{X \mid \F{G}} \leq Y$ iff
	\[
		\expec{}{X\cdot\indicator{A_i}}\leq\expec{}{Y\cdot\indicator{A_i}}, \text{ for all } i \in \Nats.
	\]
\end{lemma}

\begin{proof}
	We prove the two directions separately.
	\begin{itemize}
		\item[``$\Rightarrow$''] Let $i \in \Nats$.
			Then we have by definition of the conditional expectation
			\[
				\expec{}{X \cdot \indicator{A_i}} = \expec{}{\expec{}{X \mid \F{G}} \cdot \indicator{A_i}} \leq \expec{}{Y \cdot \indicator{A_i}},
			\]
			by monotonicity of the integral.
		\item[``$\Leftarrow$''] By \cref{thm:existence_cond_exp_nonnegative} the conditional expected value of $X$ exists and is itself a nonnegative $\F{G}$--measurable random variable.
			Consider the random variable $Z=\max(\expec{}{X \mid \F{G}}-Y,0):\Omega \to \PosRealsInf$.
			It is a result of measure theory (cf.
			\cite{bauer71measure}) that $Z$ is a $\F{G}$--measurable random variable.
			So consider the set $M=Z^{-1}((0,\infty))\in \F{G}$.
			Then we know that $M$ is a disjoint union of some of the $A_i$, so w.l.o.g.~let us assume that $A_{i_0}\subseteq M$.
			Again, a deep result from measure theory shows that $\expec{}{Z\cdot \indicator{A_{i_0}}}>0$ (cf.
			\cite{bauer71measure}) if $\IP{}{A_{i_0}}>0$.
			Then we have
			\[
				0<\expec{}{Z\cdot \indicator{A_{i_0}}} = \expec{}{(\expec{}{X \mid \F{G}}-Y)\cdot \indicator{A_{i_0}}} = \expec{}{\expec{}{X \mid \F{G}} \cdot \indicator{A_{i_0}}}-\expec{}{Y \cdot \indicator{A_{i_0}}},
			\]
			i.e., $\expec{}{Y \cdot \indicator{A_{i_0}}} < \expec{}{\expec{}{X \mid \F{G}} \cdot \indicator{A_{i_0}}}$, a contradiction.
			So, we must have $\IP{}{A_{i_0}}=0$.
			As $i_0$ was chosen arbitrarily we have $\expec{}{X \mid \F{G}} \leq Y$ a.s.
	\end{itemize}
\end{proof}

The same proof can also be done for the inverse inequality, i.e., we have the following corollary.

\begin{corollary}
	\label{coro:lower_bound_cond_exp}
	Let $X$ be a random variable with $\IP{}{X=\infty}=0$ and $\Omega = \biguplus\limits_{i=1}^{\infty} A_i$ for a sequence $A_i \in \F{F}$, $\F{G} \coloneqq \sigmagen{\{A_i \mid i \in \Nats\}}$ and $Y$ a $\F{G}$--measurable function.
	We have $\expec{}{X \mid \F{G}} \geq Y$ iff
	\[
		\expec{}{X\cdot\indicator{A_i}}\leq\expec{}{Y\cdot\indicator{A_i}}, \text{ for all } i \in \Nats.
	\]
\end{corollary}

\noindent
Recall that $\expec{}{X \mid \F{G}}$ is a random variable that is like $X$, but for those elements that are not distinguishable in the sub--$\sigma$--field $\F{G}$, it ``distributes the value of $X$ equally''.
This statement is formulated by the following lemma.
\begin{lemma}[Expected Value Does Not Change When Regarding Conditional Expected Values]
	\label{lemma:expected_value_does_not_change}
	Let $X$ be a random variable on $(\Omega, \F{F}, \mathbb{P})$ and let $\F{G}$ be a sub--$\sigma$--field of $\F{F}$.
	Then
	\[
		\expec{}{X}=\expec{}{\expec{}{X \mid \F{G}}}
	\]
\end{lemma}

\begin{proof}
	\[
		\begin{array}{rcll}
			\expec{}{\expec{}{X \mid \F{G}}} & = & \expec{}{\expec{}{X \mid \F{G}} \cdot \indicator{\Omega}} & \text{as $\indicator{\Omega}\equiv 1$}                                \\
			                                 & = & \expec{}{X \cdot \indicator{\Omega}}                      & \text{by \cref{defQQQconditional_expectation}, as $\Omega \in \F{G}$} \\
			                                 & = & \expec{}{X}                                               & \text{as $\indicator{\Omega}\equiv 1$}
		\end{array}
	\]
\end{proof}

\noindent
The following theorem shows (a) that linear operations carry over to conditional expected values w.r.t.\ sub--$\sigma$--fields, (b) that every random variable approximates itself if it is already measurable w.r.t.\ the sub--$\sigma$--field $\F{G}$, and (c) it allows to simplify multiplications with $\F{G}$--measurable random variables.
Moreover, (d) shows how to simplify expected values with several conditions.

\begin{theorem}[Properties of Conditional Expected Value \protect{\cite[p.~443]{grimmett2001probability}}]
	\label{app_thm:propert_conditional_expectation}
	Let $X,Y$ be random variables on $(\Omega, \F{F},\mathbb{P})$ and let $\F{G}$ be a sub--$\sigma$--field of $\F{F}$.
	Then the following properties hold.
	\begin{enumerate}
		\item[(a)] $\expec{}{a\cdot X+b \cdot Y \mid \F{G}}=a \cdot \expec{}{X\mid \F{G}}+b \cdot \expec{}{Y\mid \F{G}}$.
		\item[(b)] If $X$ is itself $\F{G}$--measurable then $\expec{}{X \mid \F{G}}=X$.
		\item[(c)] If $X$ is $\F{G}$--measurable then $\expec{}{X \cdot Y\mid \F{G}}=X \cdot \expec{}{Y\mid \F{G}}$.
		\item[(d)] If $\F{G} \subset \F{G}$ is a sub--$\sigma$--field of $\F{G}$ then $\expec{}{\expec{}{X\mid \F{G}}\mid \F{G}} = \expec{}{X \mid \F{G}}$.
	\end{enumerate}
\end{theorem}

\section{Proofs for Section \ref{sec:expectations_processes}}
\label{app:proofs_expectations_processes}
\renewcommand{\charwp}[3]{\Phi_{#3}}
\renewcommand{\charwpn}[4]{\Phi_{#3}^{#4}}
\noindent
We start this section with a crucial observation which will ease the proofs we conduct here.
Reconsider the filtration $(\F{F}_n^\lp)_{n \in \Nats}$ as presented in \cref{def:canonical-filtration}.
For every $n \in \Nats$ and every two distinct prefixes $\pi \neq \primed{\pi}$ of length $n+1$, their generated cylinder sets are disjoint, i.e., $Cyl(\pi)\cap Cyl(\primed{\pi})=\emptyset$ and $\biguplus_{\pi \in \States^+,\,|\pi| = n+1} Cyl(\pi)=\Omega$.
Therefore,
\begin{align*}
	\F{F}_n^\lp \overset{\textnormal{\cref{lemma:sigma_field_countable_cover}}}{\eeq} \left\{ \biguplus_{\pi \in J} Cyl(\pi) \mid J \subset \States^{n+1}\right\}.
\end{align*}
As $\F{G}_n^\lp=\F{F}_{n+1}^\lp$ (cf.~\cref{def:filtration}) we directly get
\begin{align*}
	\F{G}_n^\lp \eeq \left\{ \biguplus_{\pi \in J} Cyl(\pi) \mid J \subset \States^{n+2}\right\}.
\end{align*}

\lemmameasurability*
\begin{proof}
	We have to prove that $X_n^{f,I}$ is $\F{G}_n^\lp$--measurable, i.e., $\left(X_n^{f,I}\right)^{-1}(B)\in \F{G}_n^\lp$ for any $B \in \F{B}$.

	Consider any run $\run \in \Omega$.
	$X_n^{f,I}(\run)$ just depends on $\run[0], \cdots, \run[n+1]$, i.e., $X_n^{f,I}$ is constant on the cylinder set $Cyl(\run[0]\cdots \run[n+1])$.
	Therefore it is constant on the generators of $\F{G}_n^\lp$, i.e., on every $Cyl(\pi)$ with $|\pi| = n+2$.
	Since there are only countably many of these generators, $\left(X_n^{f,I}\right)^{-1}(B)$ for any $B \in \F{B}$ is a countable union of generators of $\F{G}_n^\lp$, i.e., $\left(X_n^{f,I}\right)^{-1}(B)\in \F{G}_n^\lp$.
\end{proof}

\thmcondexpec*
\begin{proof}
	Due to \cref{lemma:measurability}, $X_n^{f, \charwp{\guard}{C}{f}(I)}$ is $\F{G}_n^\lp$--measurable.
	So it is left to show that for any $G \in \F{G}_n^\lp$
	\begin{align}
		\expec{\State}{X_{n+1}^{f, I}\cdot \indicator{G}}=\int_{G}X_{n+1}^{f, I} \, d \left({}^{\State}\mathbb{P}\right) = \int_{G} X_n^{f, \charwp{\guard}{C}{f}(I)} \, d \left({}^{\State}\mathbb{P}\right) = \expec{\State}{X_{n}^{f, \charwp{\guard}{C}{f}(I)}\cdot \indicator{G}}.\label{cond_exp}
	\end{align}
	Due to \cref{app_lemma:conditional_expectation_countable_cover} it is enough to show \cref{cond_exp}
	for the generators of $\F{G}_n^\lp$, as any other set in $\F{G}_n^\lp$ is just a disjoint union of these generators.
	Hence, we prove \cref{cond_exp} for $G=Cyl(\pi)$ for a prefix run $\pi \in \States^{n+2}$.
	Furthermore, if $\pi \not \in \State\States^{n+1}$ the set $Cyl(\pi)$ is a nullset and hence, \cref{cond_exp} holds trivially.
	So assume that $\State_0\cdots \State_n = \pi \in \State\States^{n+1}$ and that $Cyl(\pi)$ is not a nullset, i.e., $^{\State}p(\pi)>0$.
	Note that if $\State_i \in \States_{\neg \guard}$ for some $i \leq n$, then $X_{n+1}^{f, I}$ and $X_n^{f, \charwp{\guard}{C}{f}(I)}$ are identical on $Cyl(\pi)$, as then $\termtime{\guard}(\run)\leq n\leq n+1$ for all $\run \in Cyl(\pi)$.
	So in this case \cref{cond_exp} holds trivially, too.
	Hence, we assume $\pi=\State_0 \cdots \State_{n+1} \in \States_{\guard}^{n}\States$.
	We will use a case analysis to prove the desired result.
	\begin{enumerate}
		\item $\State_{n+1}\in \States_{\neg \guard}$, i.e., $\termtime{\guard}(\run)=n+1$ for all $\run \in Cyl(\pi)$

		      \begin{align*}
			        & \int_{Cyl(\pi)}X_{n+1}^{f, I} \, d \left({}^{\State}\mathbb{P}\right)                                                                                                           \\
			      = & \int_{Cyl(\pi)}f(\State_{n+1}) \, d \left({}^{\State}\mathbb{P}\right)                                                                                                          \\
			      = & \int_{Cyl(\pi)} \iverson{\neg \guard}(\State_{n+1})\cdot f(\State_{n+1}) + \iverson{ \guard}(\State_{n+1})\cdot \wp{C}{I}(\State_{n+1}) \, d \left({}^{\State}\mathbb{P}\right) \\
			      = & \int_{Cyl(\pi)} \charwp{\guard}{C}{f}(I)(\State_{n+1}) \, d \left({}^{\State}\mathbb{P}\right)                                                                                  \\
			      = & \int_{Cyl(\pi)} X_n^{f, \charwp{\guard}{C}{f}(I)} \, d \left({}^{\State}\mathbb{P}\right)
		      \end{align*}

		\item $\State_{n+1}\in \States_{\guard}$, i.e., $\termtime{\guard}(\run)>n+1$ for all $\run \in Cyl(\pi)$

		      \begin{align*}
			        & \int_{Cyl(\pi)}X_{n+1}^{f, I}\, d \left({}^{\State}\mathbb{P}\right)                                                                                          \\
			      = & \int_{\biguplus\limits_{\State_{n+2} \in \States}Cyl(\pi \State_{n+2})}X_{n+1}^{f, I} \, d \left({}^{\State}\mathbb{P}\right)                                 \\
			      = & \sum\limits_{\State_{n+2} \in \States}\int_{Cyl(\pi \State_{n+2})}X_{n+1}^{f, I}
			      \, d \left({}^{\State}\mathbb{P}\right)                                                                                                                           \\
			      = & \sum\limits_{\State_{n+2} \in \States}\int_{Cyl(\pi \State_{n+2})}I(\State_{n+2}) \, d \left({}^{\State}\mathbb{P}\right)                                     \\
			      = & \sum\limits_{\State_{n+2} \in \States}\IP{\State}{Cyl(\pi \State_{n+2})}\cdot I(\State_{n+2})                                                                 \\
			      = & \sum\limits_{\State_{n+2} \in \States}\IP{\State}{Cyl(\pi)}\cdot \phantom{}^{\State_{n+1}}\mu_{C}(\State_{n+2}) \cdot I(\State_{n+2})                         \\
			      = & \IP{\State}{Cyl(\pi)}\cdot\sum\limits_{\State_{n+2} \in \States} \phantom{}^{\State_{n+1}}\mu_{C}(\State_{n+2}) \cdot I(\State_{n+2})                         \\
			      = & \IP{\State}{Cyl(\pi)} \cdot \wp{C}{I}(\State_{n+1})                                                                                                           \\
			      = & \IP{\State}{Cyl(\pi)} \cdot \left(\iverson{\neg \guard}(\State_{n+1})\cdot f(\State_{n+1})+\iverson{\guard}(\State_{n+1})\cdot \wp{C}{I}(\State_{n+1})\right) \\
			      = & \IP{\State}{Cyl(\pi)} \cdot \charwp{\guard}{C}{f}(I)(\State_{n+1})                                                                                            \\
			      = & \int_{Cyl(\pi)}\charwp{\guard}{C}{f}(I)(\State_{n+1}) \, d \left({}^{\State}\mathbb{P}\right)                                                                 \\
			      = & \int_{Cyl(\pi)} X_n^{f, \charwp{\guard}{C}{f}(I)} \, d \left({}^{\State}\mathbb{P}\right)
		      \end{align*}
	\end{enumerate}
\end{proof}

\corconnectionexpectation*
\begin{proof}
	Let us fix an arbitrary state $\State \in \States$.
	We will prove the result by induction.
	\begin{itemize}
		\item Induction base: $n=0$
		      \begin{align*}
			      \expec{\State}{X_0^{f,I}} & =\iverson{\neg \guard}(\State)\cdot f(\State)+\iverson{\guard}(\State)\cdot \sum \limits_{\primed{\State}\in \States}
			      \IP{\State}{Cyl(\State\primed{\State})} \cdot I(\primed{\State})                                                                                              \\
			                                & =\iverson{\neg \guard}(\State)\cdot f(\State)+\iverson{\guard}(\State)\cdot \underbrace{\sum \limits_{\primed{\State}\in \States}
			      \phantom{}^{s}\mu_{C}(\primed{\State}) \cdot I(\primed{\State})}_{=\wp{C}{I}}                                                                                 \\
			                                & =\charwp{\guard}{C}{f}(I)(\State).
		      \end{align*}
		\item Assume the result holds for a fixed $n \in \Nats$.
		      \begin{align*}
			      \expec{\State}{X_{n+1}^{f,I}} & \overset{\textnormal{\cref{lemma:expected_value_does_not_change}}}{=}\expec{\State}{\expec{\State}{X_{n+1}^{f,I}\big \vert \F{G}_n}}\overset{\textnormal{\cref{thm:cond_expec}}}{=}\expec{\State}{X_{n}^{f,\charwp{\guard}{C}{f}(I)}} \\
			                                    & \overset{I.H.}{=}\charwpn{\guard}{C}{f}{n+1}(\charwp{\guard}{C}{f}(I))(\State)= \charwpn{\guard}{C}{f}{n+2}(I)(\State).
		      \end{align*}
	\end{itemize}
\end{proof}
\renewcommand{\charwp}[3]{\Phi_{#3}}
\renewcommand{\charwpn}[4]{\Phi_{#3}^{#4}}
\lemmaaslimit*
\begin{proof}
	Let $\run$ in $\Omega$.
	\begin{align*}
		  & \lim\limits_{n \to \omega}X_n^{f,I}(\run)\cdot \indicator{(\termtime{\guard})^{-1}(\Nats)}(\run)=\lim\limits_{n \to \omega}X_{n+\termtime{\guard}(\run)}^{f,I}(\run)\cdot \indicator{(\termtime{\guard})^{-1}(\Nats)}(\run) \\
		= & \lim\limits_{n \to \omega} f(\run[\termtime{\guard}(\run)])\cdot \indicator{(\termtime{\guard})^{-1}(\Nats)}(\run)                                                                                                          \\[1.7ex]
		= & \smash{\left\{
			\begin{array}{@{}ll@{}}
				f(\run[\termtime{\guard}(\run)]), & \textnormal{if }\termtime{\guard}(\run)<\omega \\[\jot]
				0,                                & \textnormal{if }\termtime{\guard}(\run)=\omega
			\end{array}
		\right.}                                                                                                                                                                                                                        \\[.3cm]
		= & X^{f}_{\termtime{\guard}(\run)}(\run).
	\end{align*}
	If the program is universally almost--surely terminating (i.e., $\IP{\State}{\termtime{\guard}<\infty}=\IP{\State}{(\termtime{\guard})^{-1}(\Nats)}=1$ for any $\State \in \States$), then $\IP{\State}{X_n^{f,I}\cdot \indicator{(\termtime{\guard})^{-1}(\Nats)}=X_n^{f,I}}=1$ for any $\State \in \States$. Furthermore, $\IP{\State}{(\termtime{\guard})^{-1}(\Nats)}=1$ so by the previous result $\mathbf{X}^{f,I}$ converges point--wise to $X_{\termtime{\guard}}^f$ on a set with probability $1$. By definition $\mathbf{X}^{f,I}$ converges almost--surely to $X_{\termtime{\guard}}^f$.
\end{proof}

\thmlfpexpectation*
\begin{proof}
	Consider the expectation $I\coloneqq 0 $ and the indicator function $\indicator{(\termtime{\guard})^{-1}(\Nats)}$ . Then by definition $X_n^{f,0}\cdot\indicator{(\termtime{\guard})^{-1}(\Nats)} \leq X_{n+1}^{f,0}\cdot\indicator{(\termtime{\guard})^{-1}(\Nats)}$ for all $n\in \Nats$. To apply the Monotone Convergence Theorem we will calculate the expectation of $X_n^{f,0}\cdot\indicator{(\termtime{\guard})^{-1}(\Nats)}$. Note that $X_n^{f,0}$ is zero on the set $(\termtime{\guard})^{-1}(\{\omega\})$. Hence,
	\begin{align}
		\nonumber   & \charwpn{\guard}{C}{f}{n+1}(0)(\State)\overset{\textnormal{\cref{cor:connection_expectation}}}{=}\expec{\State}{X_n^{f,0}}=\expec{\State}{X_n^{f,0}\cdot \indicator{(\termtime{\guard})^{-1}(\Nats) \uplus (\termtime{\guard})^{-1}(\{\omega\})}} \\
		\nonumber = & \expec{\State}{X_n^{f,0}\cdot \indicator{(\termtime{\guard})^{-1}(\Nats)}}
		+ \underbrace{\expec{\State}{X_n^{f,0}\cdot \indicator{(\termtime{\guard})^{-1}(\{\omega\})}}}_{=0}                                                                                                                                                             \\
		\label{aux observation}
		=           & \expec{\State}{X_n^{f,0}\cdot \indicator{(\termtime{\guard})^{-1}(\Nats)}}.
	\end{align}
	Hence by the Monotone Convergence Theorem (\cref{app_thm:monotone_convergence})
	\begin{align*}
		                                                              & \expec{\State}{X^f_{\termtime{\guard}}}\overset{\textnormal{\cref{lemma:as_limit}}}{=}\expec{\State}{\lim\limits_{n \to \omega} X_n^{f,0}\cdot \indicator{(\termtime{\guard})^{-1}(\Nats)}}                \\
		\overset{\textnormal{\cref{app_thm:monotone_convergence}}}{=} & \lim\limits_{n \to \omega}\expec{\State}{X_n^{f,0} \cdot \indicator{(\termtime{\guard})^{-1}(\Nats)}}\overset{\eqref{aux observation}}{=}\lim\limits_{n \to \omega} \charwpn{\guard}{C}{f}{n+1}(0)(\State) \\
		=                                                             & \lim\limits_{n \to \omega} \charwpn{\guard}{C}{f}{n+1}(0)(\State)=(\lfp \charwp{\guard}{C}{f}) (\State).
	\end{align*}
\end{proof}

\corouniformintegrability*
\begin{proof}
	By \cref{def:uniform_integrability_stochastic_processes}, $\mathbf{X}^{f,I}$ is uniformly integrable for every $\State \in \States$ if and only if for every $\State \in \States$ the equation $\lim\limits_{n \to \omega}\expec{\State}{X_n^{f,I}}= \expec{\State}{\lim\limits_{n \to \omega}
		X_n^{f,I}}\overset{\textnormal{\cref{lemma:as_limit}}}{=}\expec{\State}{X_{\termtime{\guard}}^f}$ holds. But due to \cref{thm:lfp_expectation} and AST, we have $\expec{\State}{X^f_{\termtime{\guard}}}=(\lfp\charwp{\guard}{C}{f}) (\State)$. So $\mathbf{X}^{f,I}$ is uniformly integrable iff $\expec{\State}{X_n^{f,I}}$ converges to $(\lfp\charwp{\guard}{C}{f}) (\State)$. By \cref{cor:connection_expectation}, this is equivalent to the requirement that $\charwpn{\guard}{C}{f}n(I)$ converges to $\lfp \charwp{\guard}{C}{f}$. This is the definition of uniform integrability of $I$ for $f$, cf.\ \cref{def:uniform_integrability_invariants}.
\end{proof}
\renewcommand{\charwp}[3]{\charfun{\wpsymbol}{#1}{#2}{#3}}
\renewcommand{\charwpn}[4]{\charfunn{\wpsymbol}{#1}{#2}{#3}{#4}}

\section{Proofs for Section \ref{sec:optional-stopping-wp}}
\label{app:proofs_optional_stopping}
\renewcommand{\charwp}[3]{\Phi_{#3}}
\renewcommand{\charwpn}[4]{\Phi_{#3}^{#4}}
\lemsubmartingale*
\begin{proof}
	Let $n\in \Nats$.
	We have already proved in \cref{lemma:measurability} that $X_n^{f,I}$ is $\F{G}_n^\lp$--measurable.
	First of all, by \cref{cor:connection_expectation} we have $\expec{\State}{X_n^{f,I}}=\charwpn{\guard}{C}{f}{n+1}(I)(\State)<\infty$.
	Secondly, by \cref{thm:cond_expec} and $\charwp{\guard}{C}{f}(I) \succeq I$ we have $\expec{\State}{X_{n+1}^{f,I} \mid \F{G}_n^\lp} = X_n^{f, \charwp{\guard}{C}{f}(I)}\geq X_n^{f,I}$.
	By \cref{defQQQsubmartingale_filtration} this proves that $\mathbf{X}^{f,I}$ is a submartingale with respect to $(\F{G}_n^\lp)_{n\in \Nats}$.
\end{proof}

\thmconditionaldifferenceboundedness*
\begin{proof}
	First of all, $X_n^{0,\Diff{I}}$ is $\F{G}_n^\lp$--measurable as seen in \cref{lemma:measurability}.
	For any $\run \in \Omega$ we have
	\begin{align*}
		|X_{n+1}^{f,I}-X_n^{f,I}|(\run) = &
		\begin{cases}
			0,                             & \termtime{\guard}(\run)\leq n \\
			|f(\run[n+1]) - I(\run[n+1])|, & \termtime{\guard}(\run) = n+1 \\
			|I(\run[n+2]) - I(\run[n+1])|, & \termtime{\guard}(\run) > n+1
		\end{cases}
		\\
		=                                 &
		\begin{cases}
			0,                             & \termtime{\guard}(\run)\leq n+1 \\
			|I(\run[n+2]) - I(\run[n+1])|, & \termtime{\guard}(\run) > n+1
		\end{cases}
	\end{align*}
	as $I$ harmonizes with $f$.
	To show the result we will prove
	\begin{align}
		\expec{\State}{\big|X_{n+1}^{f,I}-X_n^{f,I}\big|\cdot \indicator{Cyl(\pi)}} = \expec{\State}{X_n^{0,\Diff{I}}\cdot \indicator{Cyl(\pi)}}\label{aux_cond_exp}
	\end{align}
	for any $\pi\in \States^{n+2}$ and use \cref{app_lemma:conditional_expectation_countable_cover} to obtain the desired result.
	Note that both sides of this equality are $0$ if $\IP{\State}{Cyl(\pi)} = 0$, so in this case the equality holds trivially.

	Take any $\pi\in \States^{n+2}$ such that $\IP{\State}{Cyl(\pi)} \neq 0$.
	Furthermore, as both random variables $\big|X_{n+1}^{f,I}-X_n^{f,I}\big|$ and $X_n^{0,\Diff{I}}$ are constant zero if all runs in $Cyl(\pi)$ have a looping time $\leq n+1$, \cref{aux_cond_exp} holds trivially in this case as well.
	Note that for any $X \in \E$, $H(X)(\State)=\iverson{\guard}(\State)\cdot \wp{C}{X}(\State)$ is zero if $\State \not \vDash \guard$.
	So, assume $\pi = \State_0 \cdots \State_{n+1} \in \State\States^{n+1}_{\guard}$.

	\begin{align*}
		  & \expec{\State}{\big|X_{n+1}^{f,I}-X_n^{f,I}\big|\cdot \indicator{Cyl(\pi)}}                                                                                                                 \\
		= & \int_{Cyl(\pi)} |X_{n+1}^{f,I}-X_n^{f,I}| \, d\left({}^{\State}\mathbb{P}\right)                                                                                                            \\
		= & \int_{\biguplus\limits_{\State_{n+2} \in \Sigma}Cyl(\pi\State_{n+2})}|X_{n+1}^{f,I}-X_n^{f,I}| \, d\left({}^{\State}\mathbb{P}\right)                                                       \\
		= & \sum\limits_{\State_{n+2}\in \States}\int_{Cyl(\pi\State_{n+2})}
		|X_{n+1}^{f,I}-X_n^{f,I}| \, d\left({}^{\State}\mathbb{P}\right)                                                                                                                                \\
		= & \sum\limits_{\State_{n+2}\in \States}\int_{Cyl(\pi\State_{n+2})}
		|I(\State_{n+2})-I(\State_{n+1})| \, d\left({}^{\State}\mathbb{P}\right)                                                                                                                        \\
		= & \sum\limits_{\State_{n+2}\in \States}\IP{\State}{Cyl(\pi\State_{n+2})} \cdot |I(\State_{n+2})-I(\State_{n+1})|                                                                              \\
		= & \sum\limits_{\State_{n+2}\in \States}\IP{\State}{Cyl(\pi)}\cdot \phantom{}^{\State_{n+1}}\mu_{C}(\State_{n+2}) \cdot |I(\State_{n+2})-I(\State_{n+1})|                                      \\
		= & \IP{\State}{Cyl(\pi)}\cdot \sum\limits_{\State_{n+2}\in \States}
		\phantom{}^{\State_{n+1}}\mu_{C}(\State_{n+2}) \cdot |I(\State_{n+2})-I(\State_{n+1})|                                                                                                          \\
		= & \IP{\State}{Cyl(\pi)}\cdot \iverson{\guard}(\State_{n+1})\cdot \sum\limits_{\State_{n+2}\in \States} \phantom{}^{\State_{n+1}}\mu_{C}(\State_{n+2}) \cdot |I(\State_{n+2})-I(\State_{n+1})| \\
		= & \IP{\State}{Cyl(\pi)}\cdot \Diff{I}(\State_{n+1})                                                                                                                                           \\
		= & \int_{Cyl(\pi)} \Diff{I}(\State_{n+1}) \, d\left({}^{\State}\mathbb{P}\right)                                                                                                               \\
		= & \int_{Cyl(\pi)} X_n^{0,\Diff{I}} \, d\left({}^{\State}\mathbb{P}\right)                                                                                                                     \\
		= & \expec{\State}{X_n^{0,\Diff{I}}\cdot \indicator{Cyl(\pi)}}.
	\end{align*}
	As we have already seen in Appendix \ref{app:proofs_expectations_processes} that $\F{G}_n^\lp = \left\{ \biguplus_{\pi \in J} Cyl(\pi) \mid J \subset \States^{n+2}\right\}$, we use \cref{app_lemma:conditional_expectation_countable_cover}
	to conclude our desired result $\expec{\State}{\big|X_{n+1}^{f,I}-X_n^{f,I}\big| ~\Big|~ \F{G}_n^\lp} \eeq X_n^{0,\Diff{I}}$.
\end{proof}

\noindent
The following auxiliary lemma is needed for the proof of \cref{thm:optional_stopping_probabilistic_programs}.

\begin{restatable}[Sufficient Condition for Uniform Integrability for a Fixed State]{lemma}{lemmaconditionaldifferenceboundednesslfp}
	\label{lemma:conditional_difference_boundedness_lfp}
	Let $I \pprec \infty$ be a conditionally difference bounded expectation that harmonizes with $f \pprec \infty$, $\charwp{\guard}{C}{f}(I) \pprec \infty$ and $\State \in \States$.
	Let the expected looping time of $\WHILEDO{\guard}{C}$ be finite for $\State \in \States$, where $C$ is AST, i.e., $\expec{\State}{\termtime{\guard}}<\infty$ and $\charwp{\guard}{C}{f}(I)(\State)<\infty$.
	Then $\charwpn{\guard}{C}{f}{n}(I)(\State)<\infty$ for all $n \in \Nats$ and
	\[
		\lim\limits_{n \to \omega} \charwpn{\guard}{C}{f}n(I)(\State) = \lfp \charwp{\guard}{C}{f} (\State)
	\]
\end{restatable}

\begin{proof}
	We present a proof based on the proof of the Optional Stopping Theorem given in \cite[Thm 12.5.(9)]{grimmett2001probability}.
	Consider the process $\mathbf{X}^{f,I}$ as studied in \cref{sec:expectations_processes}.
	As $I$ harmonizes with $f$, we have seen in \cref{thm:conditional_difference_boundedness} that if $I$ is conditionally difference bounded by the constant $c \geq 0$, then $\expec{\State}{\abs{X_{n+1}^{f,I}-X_n^{f,I}}\mid \F{G}_n^\lp}\leq c$, where $\F{G}_n^\lp$ belongs to the filtration defined in \cref{def:filtration}.
	Now we will show that $\mathbf{X}^{f,I}$ is uniformly integrable w.r.t.~$\phantom{}^{\State}\mathbb{P}$.

	Note that by definition of $\mathbf{X}^{f,I}$ (\cref{def:induced-stochastic-process}), we have $\mathbf{X}^{f,I}_{\wedge \termtime{\guard}}=\mathbf{X}^{f,I}$: Let $n \in \Nats$ and $\run \in \Omega$.
	Then if $\termtime{\guard}(\run) \leq n$, we have $X^{f,I}_{n \wedge \termtime{\guard}(\run)}(\run) = X^{f,I}_{\termtime{\guard}(\run)}(\run)=f(\run[\termtime{\guard}(\run)])=X^{f,I}_n(\run)$.
	If on the other hand $\termtime{\guard}(\run) > n$, we have $X^{f,I}_{n \wedge \termtime{\guard}(\run)}(\run) = X^{f,I}_{n}(\run)$.

	We have for any $n \in \Nats$, and any run $\run \in \Omega^\lp$:
	\begin{align*}
		X^{f,I}_n(\run) & = \abs{ X^{f,I}_{n}(\run)} = \abs{ X^{f,I}_{n \wedge \termtime{\guard}}(\run)} = \abs{
		X^{f,I}_{\min(n,\termtime{\guard}(\run))}(\run)}                                                                                                                                                           \\
		                & = \abs{X^{f,I}_0(\run) + \sum \limits_{k=0}^{\min(n,\termtime{\guard}(\run))-1} \left(X^{f,I}_{k+1}(\run) - X^{f,I}_{k}(\run)\right)}                                                    \\
		                & \leq \underbrace{\abs{X^{f,I}_0(\run)}}_{=X^{f,I}_0(\run)} + \sum \limits_{k=0}^{\min(n,\termtime{\guard}(\run))-1} \abs{X^{f,I}_{k+1}(\run) - X^{f,I}_{k}(\run)}                        \\
		                & \leq X^{f,I}_0(\run) + \sum \limits_{k=0}^{\termtime{\guard}(\run)-1} \abs{X^{f,I}_{k+1}(\run) - X^{f,I}_{k}(\run)}                                                                      \\
		                & = \underbrace{X^{f,I}_0(\run) + \sum \limits_{k=0}^{\infty} \abs{X^{f,I}_{k+1}(\run) - X^{f,I}_k(\run)}\cdot \indicator{\{\termtime{\guard}\geq k+1\}}(\run)}_{\coloneqq W^{f,I}(\run)}.
	\end{align*}
	We will show that the expectation of $W^{f,I}$ is finite.

	\begin{align*}
		\expec{\State}{W^{f,I}} & =\expec{\State}{X^{f,I}_0+\sum \limits_{k=0}^{\infty}
		\abs{\left(X^{f,I}_k - X^{f,I}_{k+1}\right)\cdot \indicator{\{\termtime{\guard}\geq k+1\}}}}                                                                                                                                                                                            \\
		                        & =\expec{\State}{X^{f,I}_0}+\sum \limits_{k=0}^{\infty} \expec{\State}{\abs{\left(X^{f,I}_k - X^{f,I}_{k+1}\right)\cdot \indicator{\{\termtime{\guard}\geq k+1\}}}}                                                                                            \\
		                        & =\expec{\State}{X^{f,I}_0}+\sum \limits_{k=0}^{\infty} \expec{\State}{\expec{\State}{\abs{\left(X^{f,I}_k - X^{f,I}_{k+1}\right)\cdot \indicator{\{\termtime{\guard}\geq k+1\}}}\mid \F{G}_k^\lp}}\tag{by \cref{lemma:expected_value_does_not_change}}        \\
		                        & =\expec{\State}{X^{f,I}_0}+\sum \limits_{k=0}^{\infty} \expec{\State}{\expec{\State}{\abs{\left(X^{f,I}_k - X^{f,I}_{k+1}\right)}\mid \F{G}_k^\lp}\cdot \indicator{\{\termtime{\guard}\geq k+1\}}} \tag{by \cref{app_thm:propert_conditional_expectation}(3)} \\
		                        & \leq \expec{\State}{X^{f,I}_0}+\sum \limits_{k=0}^{\infty} \expec{\State}{c\cdot \indicator{\{\termtime{\guard}\geq k+1\}}}                                                                                                                                   \\
		                        & =\expec{\State}{X^{f,I}_0}+\sum \limits_{k=0}^{\infty}c\cdot\expec{\State}{ \indicator{\{\termtime{\guard}\geq k+1\}}}                                                                                                                                        \\
		                        & =\expec{\State}{X^{f,I}_0}+\sum \limits_{k=0}^{\infty}c\cdot\IP{\State}{\termtime{\guard}\geq k+1}                                                                                                                                                            \\
		                        & =\expec{\State}{X^{f,I}_0}+c\cdot\sum \limits_{k=0}^{\infty}\IP{\State}{\termtime{\guard}\geq k+1}                                                                                                                                                            \\
		                        & =\expec{\State}{X^{f,I}_0}+c\cdot\sum \limits_{k=0}^{\infty}k \cdot \IP{\State}{\termtime{\guard}=k}                                                                                                                                                          \\
		                        & = \charwp{\guard}{C}{f}(I)(\State)+c\cdot \expec{\State}{\termtime{\guard}}
		< \infty~. \tag{by \cref{cor:connection_expectation} and as $\expec{\State}{\termtime{\guard}}<\infty$}
	\end{align*}
	By \cref{lem:upper_bound_uniform_integrability} the uniform integrability of $\mathbf{X}^{f,I}$ w.r.t.~$\phantom{}^{\State}\mathbb{P}$ follows.
	Therefore, we have that $\lim\limits_{n \to \omega} \charwpn{\guard}{C}{f}n(I)(\State)=\lfp \charwp{\guard}{C}{f} (\State)$ by \cref{coro:uniform_integrability}.
	Furthermore we can extract the following for any $n \in \Nats$:
	\begin{align*}
		\charwpn{\guard}{C}{f}{n+1}(I)(\State) & = \expec{\State}{X^{f,I}_n} \tag{by \cref{cor:connection_expectation} }                                                     \\
		                                       & \leq \expec{\State}{W^{f,I}} \tag{by monotonicity of the expectation}                                                       \\
		                                       & \leq \charwp{\guard}{C}{f}(I)(\State) + c \cdot \expec{\State}{\termtime{\guard}} < \infty~. \tag{by the calculation above}
	\end{align*}
	So we have just shown that $\charwpn{\guard}{C}{f}{n}(I)(\State) < \infty$ for any $n \in \Nats$.
\end{proof}

\begin{restatable}[Sufficient Condition for Uniform Integrability]{corollary}{coroconditionaldifferenceboundednesslfp}
	\label{coro:conditional_difference_boundedness_lfp}
	Let $I \pprec \infty$ be a conditionally difference bounded expectation that harmonizes with $f \pprec \infty$, $\charwp{\guard}{C}{f}(I) \pprec \infty$,
	and let the expected looping time of $\WHILEDO{\guard}{C}$ be finite for every initial state $\State \in \States$, where $C$ is AST.
	Then I is uniformly integrable for $f$ and $\Phi^{n}_f(I) \pprec \infty$ for any $n \in \Nats$.
\end{restatable}

\begin{proof}
	The result follows immediately by applying \cref{lemma:conditional_difference_boundedness_lfp} to every state $\State \in \States$.
\end{proof}

\thmoptionalstoppingprobabilisticprograms*
\begin{proof}
	That the subinvariant $I$ is a lower bound iff it is uniformly integrable for $f$ is exactly \cref{thm:uniform_integrability_least_fixed_point}.
	Nevertheless, we present a proof for the whole theorem in analogy to the Optional Stopping Theorem (\cref{thm:optional_stopping}).

	First of all, recall that $\mathbf{X}^{f,I}_{\wedge \termtime{\guard}}=\mathbf{X}^{f,I}$ holds (cf.\ the proof of \cref{coro:conditional_difference_boundedness_lfp}).

	Secondly, in any of the three cases (a) to (c), we have $\charwpn{\guard}{C}{f}n(I) \pprec \infty$ for any $n \in \Nats$: in (a) it is a precondition, in (b) it holds due to \cref{coro:conditional_difference_boundedness_lfp}, and in (c) the boundedness of $f$ and $I$ implies that $\charwpn{\guard}{C}{f}n(I)$ is bounded as well.
	So in particular, it is finite (cf.\ \cite{DBLP:series/mcs/McIverM05}).
	Therefore in any of the three cases, $\mathbf{X}^{f,I}$ is a submartingale by \cref{lemma:submartingale} as $I$ is a subinvariant.

	Furthermore, in any of the three cases (a) to (c), $\WHILEDO{\guard}{C}$ is universally almost surely terminating.
	Hence by \cref{lemma:as_limit} we have
	\[
		\lim\limits_{n \to \omega} X^{f,I}_{n \wedge \termtime{\guard}} = \lim\limits_{n \to \omega} X^{f,I}_{n} = X^{f}_{\termtime{\guard}},
	\]
	almost--surely for every $\State \in \States$.
	So if we can prove for all of the three cases (a) to (c) that $\mathbf{X}^{f,I}=\mathbf{X}^{f,I}_{\wedge \termtime{\guard}}$ is uniformly integrable for any $\State \in \States$, then we have independent of $\State \in \States$:

	\begin{align*}
		I(\State)\leq & \charwp{\guard}{C}{f}(I)(\State) \tag{as $I$ is a subinvariant}                                                \\
		=             & \expec{\State}{X_0^{f,I}}\tag{by \cref{cor:connection_expectation}}                                            \\
		\leq          & \lim\limits_{n \to \omega} \expec{\State}{X_n^{f,I}} \tag{as $\mathbf{X}^{f,I}$ is a submartingale}            \\
		=             & \expec{ \State}{\lim\limits_{n \to \omega} X_n^{f,I}} \tag{by the uniform integrability of $\mathbf{X}^{f,I}$} \\
		=             & \expec{\State}{X_{\termtime{\guard}}^{f}} \tag{by \cref{lemma:as_limit}}                                       \\
		=             & (\lfp \charwp{\guard}{C}{f}) (\State)\tag{by \cref{cor:connection_expectation}}
	\end{align*}
	i.e., $I \preceq \lfp \charwp{\guard}{C}{f} = \wp{\WHILEDO{\guard}{C}}{f}$ as desired.

	We will now use the Optional Stopping Theorem (\cref{thm:optional_stopping}) and \cref{coro:conditional_difference_boundedness_lfp} to prove the uniform integrability.
	\begin{enumerate}
		\item[(a)] Let $\State \in \States$.
			Then there is an $N(\State) \in \Nats$ with $\IP{\State}{\termtime{\guard}\leq N(\State)}$, i.e., the looping time of $\WHILEDO{\guard}{C}$ is almost--surely bounded for any $\State \in \States$.
			So by \cref{thm:optional_stopping} (a), $\mathbf{X}^{f,I}$ is uniformly integrable for any $\State \in \States$.
		\item[(b)] Due to \cref{coro:conditional_difference_boundedness_lfp}, $I$ is uniformly integrable for $f$.
			Hence $\mathbf{X}^{f,I}$ is uniformly integrable by \cref{coro:uniform_integrability}.
			As $\mathbf{X}^{f,I}$ is a submartingale, the result follows from \cref{thm:optional_stopping}.
		\item[(c)] If $f$ and $I$ are bounded, then so is the process $\mathbf{X}^{f,I}=\mathbf{X}^{f,I}_{\wedge \termtime{\guard}}$.
			By \cref{thm:optional_stopping} (c) $\mathbf{X}^{f,I}$ is uniformly integrable.
	\end{enumerate}
\end{proof}

\begin{example}[Details on \cref{ex:running_example_cdb,ex:running_example_our_rule}]

	Reconsider the program $C_{cex}$, given by \hyperlink{ex:calculations_running_example}{}
	\begin{align*}
		 & \WHILE{a \neq 0}                                                            \\
		 & \qquad \COMPOSE{\PCHOICE{\ASSIGN{a}{0}}{\sfrac{1}{2}}{\ASSIGN{b}{b + 1}}}{} \\
		 & \qquad \ASSIGN{k}{k + 1}                                                    \\
		 & \}~.
	\end{align*}
	The characteristic function of the while loop with respect to postexpectation $b$ is given by
	\begin{align*}
		\Phi_b(X) \eeq & \iverson{a = 0} \cdot b \quad \pplus \iverson{a \neq 0} \cdot \tfrac{1}{2} \cdot \Bigl(X\subst{a}{0}
		\pplus X\subst{b}{b + 1}\Bigr)\subst{k}{k+1}~.
	\end{align*}
	We have seen that
	\begin{align*}
		I \eeq b + \iverson{a \neq 0}
	\end{align*}
	and
	\begin{align*}
		\primed{I} \eeq b + \iverson{a \neq 0} \cdot \left(1+ 2^k \right)
	\end{align*}
	are fixed points of $\Phi_b$ and in \textnormal{\cref{ex:running_example_our_rule}} we proved that $I$ is indeed the least fixed point.
	So, $\primed{I}$ cannot be the least fixed point.
	Hence, $\primed{I}$ cannot be uniformly integrable and in particular, it cannot be conditionally difference bounded.
	Indeed, $\Diff{\primed{I}}$ is unbounded: Let $\State \in \States$.
	For any $x \in \Vars$ and any arithmetic expression $e$, let $\State[x/e]$ denote the state with $\State[x/e](y) = \State(y)$ for $y \in \Vars \setminus \{x\}$ and $\State[x/e](x) = \State(e)$, where $\State(e)$ is obtained by extending states from variables to arithmetic expressions in the straightforward way.
	\begin{align*}
		\Diff{\primed{I}}(\State) = & \iverson{a \neq 0}(\State)\cdot \tfrac{1}{2}\cdot\Bigl( \abs{\primed{I}-\primed{I}(\State)}\subst{k, a}{k+1, 0} (\State) + \abs{\primed{I}-\primed{I}(\State)}\subst{k, b}{k+1, b+1}(\State)\Bigr) \\
		=                           & \iverson{a \neq 0}(\State)\cdot \tfrac{1}{2}\cdot\Bigl( \abs{\primed{I}(\State\subst{k, a}{k+1, 0})-\primed{I}(\State)} + \abs{\primed{I}(\State\subst{k, b}{k+1, b+1})-\primed{I}(\State)}\Bigr)  \\
		=                           & \iverson{a \neq 0}(\State)\cdot \tfrac{1}{2}\cdot\Bigl(\abs{\State(b) + \iverson{0 \neq 0}(\State)\cdot (2^{\State(k)+1}+1) - (\State(b) + \iverson{a \neq 0}(\State)\cdot (2^{\State(k)}+1))}     \\
		                            & + \abs{\State(b)+1+ \iverson{a \neq 0}(\State)\cdot (2^{\State(k)+1}+1)-(\State(b)+ \iverson{a \neq 0}(\State)\cdot (2^{\State(k)}+1)(\State))}\Bigr)                                              \\
		=                           & \iverson{a \neq 0}(\State)\cdot \tfrac{1}{2}\cdot\Bigl( 2^{\State(k)}+1 + 1+ 2^{\State(k)}\Bigr)                                                                                                   \\
		=                           & \iverson{a \neq 0}(\State)\cdot \Bigl(2^{\State(k)}+1\Bigr)
	\end{align*}
	So we have $\Diff{\primed{I}} = \lambda\State.\iverson{a \neq 0}(\State)\cdot (2^{\State(k)}+1)$ which is unbounded.
	So $\primed{I}$ does not satisfy the preconditions of \textnormal{\cref{thm:optional_stopping_probabilistic_programs}} (b), hence our proof rule sorts out this invariant.
	Note that neither (a) (as the looping time is unbounded) nor (c) (as neither $b$ nor $\primed{I}$ are bounded) are applicable.
\end{example}

\renewcommand{\charwp}[3]{\charfun{\wpsymbol}{#1}{#2}{#3}}
\renewcommand{\charwpn}[4]{\charfunn{\wpsymbol}{#1}{#2}{#3}{#4}}

\section{Proofs for Section \ref{sec:mciver_morgan}}
\label{app:mciver_morgan}
\renewcommand{\charwp}[3]{\Phi_{#3}}
\renewcommand{\charwpn}[4]{\Phi_{#3}^{#4}}

\noindent
We will show that \cref{thm:lower_bounds_mcivermorgan} can be easily inferred from our results in \cref{sec:expectations_processes} and we can generalize (3) to a \emph{complete} proof rule.
To do so, we will make use of the \emph{Martingale Convergence Theorem} of which we present a specialized version suitable for our purposes:%
\begin{theorem}[Martingale Convergence Theorem \textnormal{\protect{\cite[Thm.12.3.(1)]{grimmett2001probability}}}]
	\label{thm:martingale_convergence}
	Let $(X_n)_{n\in \Nats}$ be a submartingale on a probability space $(\Omega,\F{F},\mathbb{P})$ with respect to a filtration $(\F{F}_n)_{n \in \Nats}$.
	If there is a constant $c\geq 0$ such that $X_n\leq c$ for every $n \in \Nats$ then there exists a random variable $X_{\omega}$ such that $\IP{}{\{\run \in \Omega \mid \lim\limits_{n \to \omega} X_n(\run)=X_{\omega}(\run)\}}=1,$ i.e., $(X_n)_{n\in \Nats}$ converges almost surely to $X_{\omega}$.
	Furthermore, $(X_n)_{n\in \Nats}$ is uniformly integrable, i.e., $\lim\limits_{n \to \omega}\expec{}{X_n}=\expec{}{X_{\omega}}.$ Moreover,
	\[
		\expec{}{X_n} \leq \expec{}{X_{\omega}} \text{ for all } n \in \Nats.
	\]
	If $(X_n)_{n\in \Nats}$ is a martingale, i.e., for all $n \in \Nats$ we have $\expec{}{X_{n+1} \mid \F{F}_n} = X_n$, then we even have
	\[
		\expec{}{X_n} = \expec{}{X_{\omega}} \text{ for all } n \in \Nats.
	\]
\end{theorem}
\noindent
Now let $f,I \in \E$ be \textbf{bounded} such that $I \preceq \charwp{\guard}{C}{f}(I)$, i.e., $I$ is a subinvariant and assume there is some $c\geq 0$ with $f,I \preceq c$.
Then the process $\mathbf{X}^{f,I}$ satisfies $X^{f,I}_n\leq c$ for every $n \in \Nats$.
By \cref{thm:martingale_convergence} there exists a random variable $X^{f,I}_{\omega}$ such that $X^{f,I}_n$ converges to $X^{f,I}_{\omega}$ almost surely.
By \cref{lemma:as_limit} we get that for any run $\run \in \Omega^{\lp}$ with $\termtime{\guard}(\run)<\omega$ we must have $X^{f,I}_{\omega}(\run)=X^{f}_{\termtime{\guard}}(\run)$, i.e., w.l.o.g.~we can assume $X^{f,I}_{\omega}\cdot \indicator{(\termtime{\guard})^{-1}(\Nats)}=X^{f}_{\termtime{\guard}}$.
As $I$ is a subinvariant we have by \cref{lemma:submartingale} that $\mathbf{X}^{f,I}$ is a submartingale.
We conclude for an arbitrary initial state $\State \in \States$
\begin{align*}
	I(\State) \lleq & \charwp{\guard}{C}{f}(I)(\State) \eeq \expec{\State}{X_0^{f,I}} \tag{\cref{cor:connection_expectation}}                                                                    \\
	\lleq           & \expec{\State}{X^{f,I}_{\omega}} \tag{\cref{thm:martingale_convergence}}                                                                                                   \\
	\eeq            & \expec{\State}{\left(\indicator{(\termtime{\guard})^{-1}(\Nats)}+\indicator{(\termtime{\guard})^{-1}(\{\omega\})}\right)\cdot X^{f,I}_{\omega}}                            \\
	\eeq            & \expec{\State}{\indicator{(\termtime{\guard})^{-1}(\Nats)}\cdot X^{f,I}_{\omega}} + \expec{\State}{\indicator{(\termtime{\guard})^{-1}(\{\omega\})}\cdot X^{f,I}_{\omega}} \\
	\eeq            & \expec{\State}{X^f_{\termtime{\guard}}} + \expec{\State}{\indicator{(\termtime{\guard})^{-1}(\{\omega\})}\cdot X^{f,I}_{\omega}}                                           \\
	\eeq            & \lfp \charwp{\guard}{C}{f}(\State) + \expec{\State}{\indicator{(\termtime{\guard})^{-1}(\{\omega\})}\cdot X^{f,I}_{\omega}}\tag{\cref{thm:lfp_expectation}}.
\end{align*}

Consequently,
\begin{equation}
	I(\State)\lleq \lfp \charwp{\guard}{C}{f}(\State) + \expec{\State}{\indicator{(\termtime{\guard})^{-1}(\{\omega\})}\cdot X^{f,I}_{\omega}}\label{eq:lower_bounding}
	.
\end{equation}

If $I$ is a \emph{fixed} point of $\charwp{\guard}{C}{f}$, then the process $\mathbf{X}^{f,I}$ is a \emph{martingale}.
By \cref{thm:cond_expec} we have for an arbitrary initial state $\State$
\[
	\expec{\State}{X_{n+1}^{f,I}} \eeq X_n^{f,\charwp{\guard}{C}{f}(I)} \eeq X_n^{f,I}.
\]
Hence, in this case $=$ instead of $\leq$ holds in \cref{eq:lower_bounding}.
We will now discuss the results of \cref{thm:lower_bounds_mcivermorgan}.
First of all, $T(\State) = \wp{\WHILEDO{\guard}{C}}{1}(\State) = \IP{\State}{\termtime{\guard}<\omega}$ by using $X^{1,T}=\indicator{(\termtime{\guard})^{-1}(\Nats)}$ and \cref{thm:lfp_expectation}.

We will now prove \thmlowerboundsmcivermorgan*

\begin{proof}
	\begin{enumerate}
		\renewcommand{\expec}[2]{\ensuremath{\phantom{}^{#1}\mathbb{E}\Bigl(#2\Bigr)}}
		\item Assume $I=\iverson{G}$ for some predicate $G$, i.e., $I(\State)\in \{0,1\}$.
		      W.l.o.g.
		      let $I(\State)=1$ as the claim holds trivially if $I(\State)=0$.
		      Then $\expec{\State}{\indicator{(\termtime{\guard})^{-1}(\{\omega\})}\cdot \underbrace{X^{f,I}_{\omega}}_{\leq 1}} \leq \expec{\State}{\indicator{(\termtime{\guard})^{-1}(\{\omega\})}}=\IP{\State}{\termtime{\guard}=\omega}.$ By \cref{eq:lower_bounding} we get $I(\State)\cdot T(\State)=T(\State) =\IP{\State}{\termtime{\guard}<\omega}=1-\IP{\State}{\termtime{\guard}=\omega}=I(\State)-\IP{\State}{\termtime{\guard}=\omega}\leq \lfp \charwp{\guard}{C}{f}(\State),$ so we have
		      \[
			      I\cdot T \ppreceq \wp{\WHILEDO{\guard}{C}}{f}.
		      \]
		      \renewcommand{\expec}[2]{\ensuremath{\phantom{}^{#1}\mathbb{E}\left(#2\right)}}
		\item Assume that for some predicate $G$ we have $\iverson{G}\preceq T$.
		      Again, w.l.o.g.
		      let $\iverson{G}(\State)=1$.
		      But then we must have $1\leq T(\State)=\IP{\State}{\termtime{\guard}<\omega}\leq 1$, i.e., $\IP{\State}{\termtime{\guard}<\omega}=1$.
		      So, $\expec{\State}{\indicator{(\termtime{\guard})^{-1}(\{\omega\})}\cdot X^{f,I}_{\omega}}=0$ and by \cref{eq:lower_bounding} we have $\iverson{G}(\State)\cdot I(\State)=I(\State)\leq \wp{\WHILEDO{\guard}{C}}{f}(\State)$, i.e.,
		      \[
			      \iverson{G} \cdot I \ppreceq \wp{\WHILEDO{\guard}{C}}{f}.
		      \]

		\item Assume there is some $\epsilon > 0$ with $\epsilon \cdot I \preceq T$.
		      By definition, $T=\lfp \charwp{\guard}{C}{1}$.
		      By \cref{eq:lower_bounding}, we have
		      \[
			      T(\State) =\lfp \charwp{\guard}{C}{1}(\State) + \expec{\State}{\indicator{(\termtime{\guard})^{-1}(\{\omega\})}\cdot X^{1,T}_{\omega}}
			      =T(\State) + \expec{\State}{\indicator{(\termtime{\guard})^{-1}(\{\omega\})}\cdot X^{1,T}_{\omega}},
		      \]
		      i.e., $\expec{\State}{\indicator{(\termtime{\guard})^{-1}(\{\omega\})}\cdot X^{1,T}_{\omega}} = 0$.
		      By definition, $\frac{T}{\epsilon}$ is a fixed point of $\charwp{\guard}{C}{\frac{1}{\epsilon}}$.
		      Thus,
		      \begin{align*}
			      \expec{\State}{\indicator{(\termtime{\guard})^{-1}(\{\omega\})}\cdot X^{f,I}_{\omega}}
			       & \leq \expec{\State}{\indicator{(\termtime{\guard})^{-1}(\{\omega\})}\cdot X^{\frac{1}{\epsilon},\frac{T}{\epsilon}}_{\omega}} \tag{$f$ is irrelevant} \\
			       & = \frac{1}{\epsilon}\cdot\expec{\State}{\indicator{(\termtime{\guard})^{-1}(\{\omega\})}\cdot X^{1,T}_{\omega}} = \frac{1}{\epsilon} \cdot 0 = 0
		      \end{align*}
		      So by \cref{eq:lower_bounding} we can conclude that $I(\State)\leq \lfp \charwp{\guard}{C}{f}(\State)$ for any state $\State$, i.e.,
		      \[
			      I \ppreceq \wp{\WHILEDO{\guard}{C}}{f}.
		      \]
	\end{enumerate}
\end{proof}
Notice that we have \emph{not} used the fact that $T$ is the termination probability but only that $T=\wp{\WHILEDO{\guard}{C}}{f}$ for some \emph{bounded} postexpectation $f$.
Furthermore, if $I$ is a lower bound, by definition $I \preceq \wp{\WHILEDO{\guard}{C}}{f}$.
Hence, we have generalized \cref{thm:lower_bounds_mcivermorgan} (3) in case of a loop with universally almost--surely terminating body to a complete characterization of lower bounds.
So we have proved the following theorem.

\thmgeneralizationmcivermorgan*

\begin{example}[Details on \cref{ex:mciver_morgan}]
	Let us consider the program $C_{rdw}$
	\begin{align*}
		 & \WHILE{x > 0}                                                     \\
		 & \qquad \PCHOICE{\ASSIGN{x}{x-1}}{\sfrac{1}{3}}{\ASSIGN{x}{x + 1}} \\
		 & \qquad \ASSIGN{y}{\max(y-1,0)}                                    \\
		 & \}~,
	\end{align*}
	with $x,y\in \Nats$ and $y \leq 100$.
	Note that this program is \emph{not} AST.
	Furthermore, the postexpectation $y$ is bounded.
	If $y \leq x$ initially then $y$ is $0$ after termination of the program.
	So, $\wp{C_{rdw}}{y}\geq \iverson{y > x}\cdot \left(\tfrac{1}{3}\right)^x\cdot (y-x)\coloneqq I$.

	Now consider $f=\iverson{y \text{ even}}\cdot 200\cdot y^2 + \iverson{y \text{ odd}}\cdot (y+5)^4$.
	We have $I' \leq \Phi_f(I')$, where $I'=400\cdot I$.

	\begin{align*}
		\Phi_f(I') & \eeq \iverson{x = 0} \cdot f + \iverson{x>0} \bigl( \tfrac{1}{3} \cdot I'\subst{x,y}{x-1,\max(y-1,0)}+ \tfrac{2}{3}\cdot I'\subst{x,y}{x+1,\max(y-1,0)} \bigr)                                                                \\
		           & \eeq \iverson{x = 0} \cdot f + \iverson{x>0} \cdot 400\cdot\bigl(\tfrac{1}{3} \cdot \iverson{\max(y-1,0) > x-1}\cdot \left(\tfrac{1}{3}\right)^{x-1}\cdot (\max(y-1,0)-(x-1))                                                 \\
		           & \qquad +\tfrac{2}{3}\cdot \iverson{\max(y-1,0) > x+1}\cdot \left(\tfrac{1}{3}\right)^{x+1}\cdot (\max(y-1,0)-(x+1)) \bigr)                                                                                                    \\
		           & \eeq \iverson{x = 0} \cdot f + \iverson{x>0} \cdot 400\cdot\bigl(\tfrac{1}{3} \cdot \iverson{y-1 > x-1}\cdot \left(\tfrac{1}{3}\right)^{x-1}\cdot (y-1-(x-1))                                                                 \\
		           & \qquad +\tfrac{2}{3}\cdot \iverson{y-1 > x+1}\cdot \left(\tfrac{1}{3}\right)^{x+1}\cdot (y-1-(x+1)) \bigr)                                                                                                                    \\
		           & \eeq \iverson{x = 0} \cdot f + \iverson{x>0} \cdot 400\cdot\bigl(\iverson{y > x}\cdot \left(\tfrac{1}{3}\right)^{x}\cdot (y-x) + \tfrac{2}{3}\cdot \iverson{y > x+2}\cdot \left(\tfrac{1}{3}\right)^{x+1}\cdot (y-x-2) \bigr) \\
		           & \eeq \iverson{x = 0} \cdot f + \iverson{x>0} \bigl(I' + \underbrace{400\cdot\tfrac{2}{3}\cdot \iverson{y > x+2}\cdot \left(\tfrac{1}{3}\right)^{x+1}\cdot (y-x-2)}_{\geq 0} \bigr)                                            \\
	\end{align*}
	If $\State(x)>0$, then obviously $I(\State) \leq \Phi_f(I')(\State)$ by the calculation above.
	If $\State(x)=0$, then we have

	\begin{align*}
		\Phi_f(I') (\State) & \eeq \iverson{x = 0}(\State) \cdot f (\State)                                                                                \\
		                    & \eeq \iverson{y \text{ is even}}(\State)\cdot 200\cdot \State(y)^2 + \iverson{y \text{ is odd}}(\State)\cdot (\State(y)+5)^4 \\
		                    & \ggeq 400 \cdot \State(y) = \iverson{y > x}(\State)\cdot \left(\tfrac{1}{3}\right)^0\cdot (\State(y)-0) = I'(\State),
	\end{align*}
	as for every even $y$ we have $200\cdot y^2 \geq 400 \cdot y$ and for every odd $y$ we have $200\cdot (y+5)^4 \geq 400 \cdot y$.

	We have $\tfrac{1}{400}\cdot I' \preceq \wp{C_{rdw}}{y}$.
	Thus, we can conclude from \textnormal{\cref{thm:generalization_mciver_morgan}} that $I' \preceq \wp{C_{rdw}}{f}$.
	Note that this is easier than relating $I'$ and the termination probability as required in \textnormal{\cref{thm:lower_bounds_mcivermorgan}} as $y$ does not influence the termination behavior of the loop.
\end{example}

\renewcommand{\charwp}[3]{\charfun{\wpsymbol}{#1}{#2}{#3}}
\renewcommand{\charwpn}[4]{\charfunn{\wpsymbol}{#1}{#2}{#3}{#4}}

\section{Details for Section \ref{sec:runtime}}
\label{app:runtime}

\subsection{Proofs}
\label{app:runtime_proofs}
\thmlowerboundsert*
\begin{proof}
	Remember the connection between $\wpsymbol$ and $\ertsymbol$ (cf.
	\cite[Thm.
		5.2]{DBLP:conf/lics/OlmedoKKM16}): For any probabilistic program $P$ we have
	\begin{align}
		\ert{P}{t} \eeq \wp{P}{t} + \ert{P}{0}\label{eq:wp_ert}
	\end{align}

	Our goal is to show $\charertnnoindex{\guard}{C}{t}{n}(I) = \charwpnnoindex{\guard}{C}{t}{n}(I) + \charertnnoindex{\guard}{C}{0}{n}(0)$ for all $n \geq 1$.
	We use induction on $n$ to prove this result.
	In the base case we have $n=1$.
	Here, we obtain

	\begin{align*}
		\charertnoindex{\guard}{C}{t}(I) = & 1 + \iverson{\neg \guard}\cdot t + \iverson{\guard}\cdot\ert{C}{I}                                                                \\
		=                                  & 1 + \iverson{\neg \guard}\cdot t + \iverson{\guard}\cdot\wp{C}{I} + \iverson{\guard}\cdot\ert{C}{0}                               \\
		=                                  & \iverson{\neg \guard}\cdot t + \iverson{\guard}\cdot\wp{C}{I} + 1 + \iverson{\neg \guard}\cdot 0 +\iverson{\guard}\cdot\ert{C}{0} \\
		=                                  & \charwpnoindex{\guard}{C}{t}(I) + \charertnoindex{\guard}{C}{0}(0).
	\end{align*}

	In the induction step we use the induction hypothesis $\charertnnoindex{\guard}{C}{t}{n}(I) = \charwpnnoindex{\guard}{C}{t}{n}(I) + \charertnnoindex{\guard}{C}{0}{n}(0)$.
	Then we have
	\begin{align*}
		  & \charertnnoindex{\guard}{C}{t}{n+1}(I)                                                                                                                                                                      \\
		= & \charertnoindex{\guard}{C}{t}(\charertnnoindex{\guard}{C}{t}{n}(I))                                                                                                                                         \\
		= & 1 + \iverson{\neg \guard}\cdot t + \iverson{\guard}\cdot\ert{C}{\charertnnoindex{\guard}{C}{t}{n}(I)} \tag{by definition}                                                                                   \\
		= & 1 + \iverson{\neg \guard}\cdot t + \iverson{\guard}\cdot\wp{C}{\charertnnoindex{\guard}{C}{t}{n}(I)} + \iverson{\guard}\cdot\ert{C}{0} \tag{by \cref{eq:wp_ert}}                                            \\
		= & 1 + \iverson{\neg \guard}\cdot t + \iverson{\guard}\cdot\wp{C}{\charwpnnoindex{\guard}{C}{t}{n}(I) + \charertnnoindex{\guard}{C}{0}{n}(0)} + \iverson{\guard}\cdot\ert{C}{0}\tag{by linearity of \wpsymbol} \\
		= & 1 + \iverson{\neg \guard}\cdot t + \iverson{\guard}\cdot\wp{C}{\charwpnnoindex{\guard}{C}{t}{n}(I)} + \iverson{\guard}\cdot\wp{C}{\charertnnoindex{\guard}{C}{0}{n}(0)} + \iverson{\guard}\cdot\ert{C}{0}   \\
		= & 1+ \charwpnnoindex{\guard}{C}{t}{n+1}(I) + \iverson{\guard}\cdot\wp{C}{\charertnnoindex{\guard}{C}{0}{n}(0)} + \iverson{\guard}\cdot\ert{C}{0}\tag{by definition}                                           \\
		= & 1+ \charwpnnoindex{\guard}{C}{t}{n+1}(I) + \iverson{\guard}\cdot\ert{C}{\charertnnoindex{\guard}{C}{0}{n}(0)}                                                                                               \\
		= & \charwpnnoindex{\guard}{C}{t}{n+1}(I) + 1 + \iverson{\neg \phi}\cdot 0 + \iverson{\guard}\cdot\ert{C}{\charertnnoindex{\guard}{C}{0}{n}(0)}\tag{by \cref{eq:wp_ert}}                                        \\
		= & \charwpnnoindex{\guard}{C}{t}{n+1}(I) + \charertnnoindex{\guard}{C}{0}{n+1}(0).                                                                                                                             \\
	\end{align*}
	So $\charertnnoindex{\guard}{C}{t}{n}(I) = \charwpnnoindex{\guard}{C}{t}{n}(I) + \charertnnoindex{\guard}{C}{0}{n}(0)$ holds for an arbitrary $n \in \Nats$ with $n \geq 1$.

	Now let $\State \in \States$.
	Then one of the following two cases occurs.
	\begin{enumerate}
		\item \underline{$\ert{\WHILEDO{\guard}{C}}{0}(\State)=\infty$}

		      In this case we have by \cref{eq:wp_ert} $\ert{\WHILEDO{\guard}{C}}{t}(\State)=\infty \geq I(\State)$.

		\item \underline{$\ert{\WHILEDO{\guard}{C}}{0}(\State)<\infty$}

		      In this case we have $\expec{\State}{\termtime{\guard}}<\infty$.
		      First of all, in \cref{thm:optional_stopping_probabilistic_programs} (b) we have seen, that if $I$ is conditionally difference bounded, $\charwpnoindex{\guard}{C}{t}(I)\pprec \infty$ and $\expec{\State'}{\termtime{\guard}}<\infty$ for every $\State' \in \States$ then we have $\lim \limits_{n \to \omega}\charwpnnoindex{\guard}{C}{t}{n}(I) = \lfp \charwpnoindex{\guard}{C}{t}$.
		      However, in \cref{thm:optional_stopping_probabilistic_programs} (b) we need that the expected looping time is finite for \emph{every} initial state $\State' \in \States$.
		      As we cannot ensure this condition, we use \cref{lemma:conditional_difference_boundedness_lfp}, a specialized result used in the proof of \cref{thm:optional_stopping_probabilistic_programs} which is indeed dependent on the initial state $\State \in \States$.

		      Furthermore, the expected runtime of the program with initial state $\State \in \States$ is finite so, the expected looping time of the program has to be finite as well, i.e., $\ert{\WHILEDO{\guard}{C}}{0}(\State)<\infty$ implies $\expec{\State}{\termtime{\guard}}<\infty$.

		      Hence, as $I$ harmonizes with $t$, $I$ is conditionally difference bounded,$\charwpnoindex{\guard}{C}{t}(I) \pprec \infty$, and $\expec{\State}{\termtime{\guard}}<\infty$.
		      Thus, we can apply \cref{lemma:conditional_difference_boundedness_lfp} and get
		      \[
			      \lim \limits_{n \to \omega}\charwpnnoindex{\guard}{C}{t}{n}(I)(\State) = \lfp \charwpnoindex{\guard}{C}{t}(\State)=\wp{\WHILEDO{\guard}{C}}{t}(\State).
		      \]
		      Hence we have
		      \begin{align*}
			      \charertnnoindex{\guard}{C}{t}{n}(I)(\State) = & \charwpnnoindex{\guard}{C}{t}{n}(I)(\State) + \charertnnoindex{\guard}{C}{0}{n}(0)(\State) \tag{by calculation above} \\
			      \xrightarrow{n \to \omega}                     & \wp{\WHILEDO{\guard}{C}}{t}(\State) + \ert{\WHILEDO{\guard}{C}}{0}(\State)                                            \\
			      =                                              & \ert{\WHILEDO{\guard}{C}}{t}(\State)
		      \end{align*}
		      But the sequence $\charertnnoindex{\guard}{C}{t}{n}(I)_{n \in \Nats}$ is monotonically increasing as $I$ is an $\ertsymbol$--subinvariant.
		      Hence, $I(\State)\leq \lim\limits_{n \to \omega}\charertnnoindex{\guard}{C}{t}{n}(I)(\State)=\ert{\WHILEDO{\guard}{C}}{t}(\State)$.
	\end{enumerate}
	Combining these results, we get $I\leq \ert{\WHILEDO{\guard}{C}}{t}$.
\end{proof}

\subsection{Details for Example \ref{ex:coupon-collector}}
\label{app:coupon-collector}

More detailed annotations for the outer loop of the coupon collector are as follows:
\begin{align*}
	 & \eqannotate{1 \pplus N \cdot \harm{N} }                                                                                                         \\
	 & \ertannotate{1 \pplus \iverson{0 < N \leq N} \cdot N \cdot \harm{N} \pplus \iverson{N < N} \cdot (N \cdot \harm{N} + N - N)}                          \\
	 & \COMPOSE{\ASSIGN{x}{N}}{}                                                                                                                 \\
	 & \preceqannotate{\iverson{0 < x \leq N} \cdot N \cdot \harm{x} \pplus \iverson{N < x} \cdot (N \cdot \harm{N} + x - N)}                          \\
	 & \eqannotate{
		1 \pplus \iverson{0 < x \leq N} \cdot \left( N \cdot \harm{x} + 3 + \tfrac{N}{x}\right) \pplus \iverson{N < x} \cdot (2 + N \cdot \harm{N} + x - N)
	}
	\\
	 & \eqannotate{
		1 \pplus \iverson{x = 1} \cdot (3 + 2 \cdot N ) \pplus \iverson{1 < x \leq N} \cdot \left( 3 + N \cdot \left(\harm{x} + \tfrac{1}{x}\right)\right)
	}                                                                                                                                            \\
	 & \qquad \annocolor{\boldsymbol{
			\pplus \iverson{x = N + 1} \cdot \left( 3 + N \cdot \harm{N} \right) \pplus \iverson{N + 1 < x} \cdot (2 + N \cdot \harm{N} + x - N)
	}}                                                                                                                                           \\
	 & \eqannotate{
		1 \pplus \iverson{x = 1} \cdot (3 + 2 \cdot N ) \pplus \iverson{1 < x \leq N} \cdot \left( 3 + N \cdot \left(\harm{x} - \tfrac{1}{x}\right) + \tfrac{2 \cdot N}{x} \right)
	}                                                                                                                                            \\
	 & \qquad \annocolor{\boldsymbol{
			\pplus \iverson{x = N + 1} \cdot \left( 3 + N \cdot \left(\harm{N + 1} - \tfrac{1}{N + 1}\right) \right) \pplus \iverson{N + 1 < x} \cdot (2 + N \cdot \harm{N} + x - N)
	}}                                                                                                                                           \\
	 & \eqannotate{
		1 \pplus \iverson{0 < x} \cdot \Bigl( 2 + t + \iverson{x \leq N} \cdot \tfrac{2 \cdot N}{x}
		\Bigr)
	}                                                                                                                                            \\
	 & \phiannotate{
		1 \pplus \iverson{x \leq 0} \cdot 0 \pplus \iverson{0 < x} \cdot \Bigl( 2 + t + \iverson{0 < x \leq N} \cdot \tfrac{2 \cdot N}{x}
		\Bigr)
	}                                                                                                                                            \\
	 & \WHILE{0 < x}                                                                                                                             \\
	 & \qquad \eqannotate{
		2 \pplus t \pplus \iverson{0 < x \leq N} \cdot \tfrac{2 \cdot N}{x}
	}                                                                                                                                            \\
	 & \qquad \ertannotate{
		1 \pplus 1 \pplus t \pplus \iverson{0 < x < N + 1} \cdot 2 \cdot \Max{\tfrac{N}{x}}{1}
	}                                                                                                                                            \\
	 & \qquad \ASSIGN{i}{N+1}                                                                                                                    \\
	 & \qquad \ertannotate{
		1 \pplus t \pplus \iverson{0 < x < i} \cdot 2 \cdot \Max{\tfrac{N}{x}}{1}
	} \tag{by \cref{lem:batz-special-case} below}                                                                                                \\
	 & \qquad \WHILEDO{0 < x < i}{ \ASSIGN{i}{\mathrm{Unif}[1..N]} }                                                                             \\
	 & \qquad \eqannotate{\underbrace{
		1 \pplus \iverson{1 < x \leq N + 1} \cdot N \cdot \left(\harm{x} - \tfrac{1}{x}\right) \pplus \iverson{N + 1 < x} \cdot (N \cdot \harm{N} + x - 1 - N)
	}_{~{}~{}~{}~{}~{}\eqqcolon t}}                                                                                                              \\
	 & \qquad \ertannotate{1 \pplus \iverson{0 < x - 1 \leq N} \cdot N \cdot \harm{x - 1} \pplus \iverson{N < x - 1} \cdot (N \cdot \harm{N} + (x - 1) - N)} \\
	 & \qquad \ASSIGN{x}{x-1}                                                                                                                    \\
	 & \qquad \starannotate{\iverson{0 < x \leq N} \cdot N \cdot \harm{x} \pplus \iverson{N < x} \cdot (N \cdot \harm{N} + x - N) ~}                         \\
	 & \annotate{0}
\end{align*}
For the inner loop, we make use of the following Lemma, for which we also give a detailed proof in the following:
\begin{lemma}
	\label{lem:batz-special-case}
	Let $t \in \E$ be a runtime (i.e.\ an expectation) such that $t$ does \emph{not} depend on program variable~$i$.
	Then the following expected runtime annotation is valid:
	\begin{align*}
		 & \annotate{
			1 \pplus t \pplus \iverson{0 < x < i} \cdot 2\cdot \Max{\tfrac{N}{x}}{1}
		}                                                       \\
		 & \WHILEDO{0 < x < i}{\ASSIGN{i}{\mathrm{Unif}[1..N]}} \\
		 & \annotate{t}
	\end{align*}
\end{lemma}
\begin{proof}
	We employ \cite[Theorem 4]{DBLP:conf/esop/BatzKKM18} for so-called $t$--independent and identically distributed loops ($t$-i.i.d.~loops for short) (see \cite[Definition 5]{DBLP:conf/esop/BatzKKM18}).
	In order to verify the $t$-i.i.d.-ness of the loop $\WHILEDO{0 < x < i}{\ASSIGN{i}{\mathrm{Unif}[1..N]}}$, we have to establish that neither
	\begin{align*}
		\wp{\ASSIGN{i}{\mathrm{Unif}[1..N]}}{\iverson{0 < x < i}}
		\eeq \frac{1}{N} \cdot \sum_{m = 1}^{N} \iverson{0 < x < m}
		\eeq \iverson{0 < x} \cdot \Max{1 - \tfrac{x}{N}}{0} \tag{$\dagger$}
	\end{align*}
	nor
	\begin{align*}
		 & \wp{\ASSIGN{i}{\mathrm{Unif}[1..N]}}{\iverson{\neg(0 < x < i)} \cdot t}                                                                                                        \\
		 & \eeq \wp{\ASSIGN{i}{\mathrm{Unif}[1..N]}}{\iverson{\neg(0 < x < i)}} \cdot t \tag{by \cite[Lemma 1]{DBLP:conf/esop/BatzKKM18}}                                                 \\
		 & \eeq \wp{\ASSIGN{i}{\mathrm{Unif}[1..N]}}{1 - \iverson{0 < x < i}} \cdot t                                                                                                     \\
		 & \eeq \bigl( \wp{\ASSIGN{i}{\mathrm{Unif}[1..N]}}{1} - \wp{\ASSIGN{i}{\mathrm{Unif}[1..N]}}{\iverson{0 < x < i}} \bigr) \cdot t \tag{by \cite[Corollary 4.22]{thesis:kaminski}} \\
		 & \eeq \left( 1 - \iverson{0 < x} \cdot \Max{1 - \tfrac{x}{N}}{0} \right) \cdot t                                                                                                \\
		 & \eeq \left(\iverson{x \leq 0} \pplus \iverson{0 < x} \cdot \Min{\tfrac{x}{N}}{1}\right) \cdot t \tag{$\ddagger$}
	\end{align*}
	depend on program variable $i$ which is indeed the case by the assumption that $t$ does not depend on $i$.
	Additionally to $t$-i.i.d.-ness, \cite[Theorem 4]{DBLP:conf/esop/BatzKKM18} requires us to establish that
	\begin{align*}
		\ert{\ASSIGN{i}{\mathrm{Unif}[1..N]}}{0} \eeq 1
	\end{align*}
	does not depend on variable $i$ and that the loop body terminates almost-surely, i.e.\
	\begin{align*}
		\wp{\ASSIGN{i}{\mathrm{Unif}[1..N]}}{1} \eeq 1~.
	\end{align*}
	Both conditions are obviously true.

	Having established all preconditions of \cite[Theorem 4]{DBLP:conf/esop/BatzKKM18}, we can now make the following $\ertsymbol$-annotations (recall that such annotations are best read from bottom to top):
	\begin{align*}
		 & \annotate{
			1 \pplus t \pplus \iverson{0 < x < i} \cdot 2 \cdot \Max{\tfrac{N}{x}}{1}
		}                                                                                                                                                                                                \\
		 & \annotate{
			1 \pplus \iverson{0 < x < i} \cdot \left( \frac{2}{\Min{\tfrac{x}{N}}{1}} + t \right) \pplus \iverson{\neg(0 < x < i)} \cdot t
		}                                                                                                                                                                                                \\
		 & \annotate{
			1 \pplus \iverson{0 < x < i} \cdot \frac{1 + 1 + \Min{\tfrac{x}{N}}{1} \cdot t}{\Min{\tfrac{x}{N}}{1}}
			\pplus \iverson{\neg(0 < x < i)} \cdot t
		}
		\tag{see $(\ddagger)$}                                                                                                                                                                           \\
		 & \annotate{
		1 \pplus \iverson{0 < x < i} \cdot \frac{1 + \ert{\ASSIGN{i}{\mathrm{Unif}[1..N]}}{0} + \wp{\ASSIGN{i}{\mathrm{Unif}[1..N]}}{\iverson{\neg(0 < x < i)} \cdot t}}{1 - \Max{1 - \tfrac{x}{N}}{0}}} \\
		 & \qquad\annocolor{\boldsymbol{\pplus \iverson{\neg(0 < x < i)} \cdot t}
		}
		\tag{see \cite{DBLP:journals/jacm/KaminskiKMO18}}                                                                                                                                                \\
		 & \annotate{
			1 \pplus \iverson{0 < x < i} \cdot \frac{1 + \ert{\ASSIGN{i}{\mathrm{Unif}[1..N]}}{\iverson{\neg(0 < x < i)} \cdot t}}{1 - \Max{1 - \tfrac{x}{N}}{0}} \pplus \iverson{\neg(0 < x < i)} \cdot t
		}
		\tag{with $\tfrac{0}{0} = 0$ by~\cite[Theorem 4]{DBLP:conf/esop/BatzKKM18} and $(\dagger)$}                                                                                                      \\
		 & \WHILE{0 < x < i}                                                                                                                                                                             \\
		 & \qquad \ASSIGN{i}{\mathrm{Unif}[1..N]}                                                                                                                                                        \\
		 & \}                                                                                                                                                                                            \\
		 & \annotate{t}
	\end{align*}
	It is important to note that, again, any loop semantics needed to be applied only a finite number of times.
	In particular, it was not necessary to find the limit of a sequence or anything alike.
\end{proof}

}
\end{document}